\documentclass[11pt,draftcls,onecolumn]{IEEEtran}
\usepackage{graphicx}
\usepackage{epstopdf}
\usepackage{epsfig}
\usepackage{amsmath}
\usepackage{amssymb}
\usepackage{mathrsfs}
\usepackage{mathtools}
\usepackage{float}
\usepackage{url}
\usepackage{enumerate}
\usepackage{stmaryrd}
\usepackage{relsize}

\usepackage[left=20mm,right=20mm,top=24mm,bottom=24mm]{geometry}

\newtheorem{corollary}{Corollary}
\newtheorem{definition}{Definition}

\newtheorem{theorem}{Theorem}
\newtheorem{proposition}{Proposition}
\newtheorem{lemma}{Lemma}
\newtheorem{remark}{Remark}
\newtheorem{example}{Example}

\newcommand{\expectation}{\ensuremath{\mathbb{E}}}
\newcommand{\prob}{\ensuremath{\mathbb{P}}}
\newcommand{\naturals}{\ensuremath{\mathbb{N}}}
\newcommand{\Reals}{\ensuremath{\mathbb{R}}}
\newcommand{\set}{\ensuremath{\mathcal}}

\MHInternalSyntaxOn
\renewcommand{\dcases}
 {
  \MT_start_cases:nnnn
    {\quad}
    {$\m@th\displaystyle##$\hfil}
    {$\m@th\displaystyle##$\hfil}
    {\lbrace}
 }
\MHInternalSyntaxOff

\begin{document}

\title{\huge{Improved Bounds on Lossless Source Coding\\
and Guessing Moments via R\'{e}nyi Measures}}
\author{Igal Sason \qquad Sergio Verd\'{u}
\thanks{
I. Sason is with the Andrew and Erna Viterbi Faculty of Electrical Engineering,
Technion--Israel Institute of Technology, Haifa 32000, Israel (e-mail:
sason@ee.technion.ac.il).}
\thanks{
S. Verd\'{u} is with the Department of Electrical Engineering, Princeton
University, Princeton, New Jersey 08544, USA (e-mail: verdu@princeton.edu).}
\thanks{
This manuscript has been submitted to the {\em IEEE Transactions on Information Theory}
in November~1, 2017, and accepted for publication in January~28, 2018. It is published 
in the {\em IEEE Trans. on Information Theory}, vol.~64, no. 6, pp.~4323--4346, June 2018.
This is a post-print which includes few corrections of printing typos.}
\thanks{
Communicated by S. Watanabe, Associate Editor for Shannon Theory.}
\thanks{
Parts of this work have been presented at the 2018 International Zurich Seminar on
Information and Communication, Zurich, Switzerland, February 21-23, 2018, and in
the 2018 IEEE International Symposium on Information Theory, Vail, Colorado, USA,
June 17--22, 2018.}
\thanks{
This work has been supported by the Israeli Science Foundation (ISF) under
Grant 12/12, by ARO-MURI contract number W911NF-15-1-0479
and in part by the Center for Science of Information, an NSF Science and
Technology Center under Grant CCF-0939370.}
}

\maketitle

\begin{abstract}
This paper provides upper and lower bounds on the optimal guessing moments
of a random variable taking values on a finite set when side information
may be available. These moments quantify the number of guesses required for
correctly identifying the unknown object and, similarly to Arikan's bounds,
they are expressed in terms of the Arimoto-R\'{e}nyi conditional entropy.
Although Arikan's bounds are asymptotically tight, the improvement of the
bounds in this paper is significant in the non-asymptotic regime. Relationships
between moments of the optimal guessing function and the MAP error probability
are also established, characterizing the exact locus of their attainable values.
The bounds on optimal guessing moments serve to improve non-asymptotic
bounds on the cumulant generating function of the codeword lengths for
fixed-to-variable optimal lossless source coding without prefix constraints.
Non-asymptotic bounds on the reliability function of discrete
memoryless sources are derived as well. Relying on these techniques,
lower bounds on the cumulant generating function of the codeword lengths are
derived, by means of the smooth R\'{e}nyi entropy, for source codes that allow
decoding errors.
\end{abstract}


\section{Introduction}
\label{section: introduction}

\subsection{Prior work}
The problem of guessing discrete random variables has found a variety of applications
in information theory, coding theory, cryptography, and searching and sorting algorithms.
The central object of interest is the distribution of the number of guesses
required to identify a realization of a random variable $X$, taking values on a
finite or countably infinite set $\set{X} = \{1, \ldots, |\set{X}|\}$, by asking
questions of the form ``Is $X$ equal to $x$?''.
A {\em guessing function} is a one-to-one
function $g \colon \set{X} \to \set{X}$, which can be viewed as a permutation
of the elements of $\set{X}$ in the order in which they are guessed. We can
envision a generic algorithm that outputs $g^{-1}(1)$; a
supervisor checks whether $X = g^{-1}(1)$, if so then the algorithm halts; otherwise,
the algorithm outputs $g^{-1}(2)$ and the process repeats until the value of $X$ is guessed
correctly. Therefore, the number of guesses is $g(x)$ when the true outcome is $x \in \set{X}$.

Lower and upper bounds on the minimal expected number of required guesses for correctly
identifying the realization of $X$, expressed as a function of the Shannon entropy $H(X)$,
have been respectively derived by Massey \cite{Massey94} and by McEliece and Yu
\cite{McElieceYu95}. More generally, given a probability mass function $P_X$ on $\set{X}$,
it is of interest to minimize the generalized guessing moment
\begin{align}  \label{eq: 20171019-1}
\expectation[g^{\rho}(X)] = \underset{x \in \set{X}}{\sum} P_X(x) g^{\rho}(x), \quad \rho > 0.
\end{align}
For an arbitrary positive $\rho$, the $\rho$-th moment of the number
of guesses is minimized by selecting the guessing function to be a {\em ranking function}
$g_X$, for which $g_X(x)=k$ if $P_X(x)$ is the $k$-th largest mass.
Upper and lower bounds on the $\rho$-th moment of ranking functions, expressed
in terms of the R\'{e}nyi entropy $H_{\alpha}(X)$ of order $\alpha = \frac1{1+\rho}$,
were derived by Arikan \cite{Arikan96}, followed by a refined upper bound
by Bozta\c{s} \cite{Boztas97}.
Although if $|\set{X}|$ is small, it is straightforward to evaluate numerically the guessing
moments, the benefit of bounds
expressed in terms of R\'{e}nyi entropies is particularly relevant when dealing with
a random vector $X^n = (X_1, \ldots, X_n)$ whose letters belong to a finite
alphabet $\set{A}$; computing all the probabilities of the mass
function $P_{X^n}$ over the set $\set{A}^n$, and then sorting
them in decreasing order for the calculation of the $\rho$-th moment of the optimal
guessing function for the elements of $\set{A}^n$ has exponential complexity in $n$.
Therefore, it becomes infeasible even for moderate values of $n$. In contrast, regardless
of the value of $n$, bounds on guessing moments which depend on the R\'{e}nyi entropy are
readily computable if for example $\{X_i\}_{i=1}^n$ are independent; in which case,
the R\'{e}nyi entropy of the vector is equal to the sum of the R\'{e}nyi entropies of its
components (hence, the exponential complexity is reduced to linear
complexity in $n$; furthermore, in the i.i.d. case, the complexity in calculating the
R\'{e}nyi entropy of $X^n$ is independent of $n$).
Arikan's bounds are asymptotically tight for random vectors of length $n$ as $n \to \infty$,
so another benefit of these bounds is that they provide the correct exponential growth rate of
the guessing moments for sufficiently large $n$. In \cite{Arikan96}, Arikan generalized his
bounds to allow side information, leading to asymptotically tight bounds which are expressed
in terms of the Arimoto-R\'{e}nyi conditional entropy \cite{Arimoto75}.

The guessing problem has been studied in the information-theoretic literature in various
contexts, which include: guessing subject to distortion \cite{ArikanM98-1}, joint
source-channel coding and guessing with application to sequential decoding \cite{ArikanM98-2},
guessing with a prior access to a malicious oracle \cite{BurinS_arXiv17}, a large
deviations approach to guessing and source compression (\cite{ChristiansenD13, HanawalS11, Sundaresan07b}),
guessing with limited memory \cite{Huleihel17}, guesswork exponents for Markov sources
\cite{MaloneS}, guessing in secrecy problems (\cite{Salam17, Yona17}), and guessing under
source uncertainty \cite{Sundaresan07}.

For uniquely-decodable lossless source coding, Campbell (\cite{Campbell65, Campbell66})
proposed the normalized cumulant generating function of the codeword lengths as a generalization
to the frequently used design criterion of normalized average code length. Campbell's
motivation in \cite{Campbell65} was to control the contribution of the longer codewords
via a free parameter in the cumulant generating function: if the value of this parameter
tends to zero, then the resulting design criterion becomes the normalized average
code length while by increasing the value of the free parameter, the penalty for longer
codewords is more severe, and the resulting code optimization yields a reduction in the fluctuations
of the codeword lengths. In \cite{Campbell65}, Campbell obtained asymptotically tight upper and lower bounds
on the minimum normalized cumulant generating function for discrete memoryless stationary
sources with finite alphabet. These bounds, expressed in terms of the R\'{e}nyi entropy,
imply that for sufficiently long source sequences, it is possible to make the normalized
cumulant generating function of the codeword lengths approach the R\'{e}nyi entropy as closely
as desired by a proper fixed-to-variable uniquely-decodable source code; moreover, a converse
result in \cite{Campbell65} shows that there is no uniquely-decodable source code for which the normalized
cumulant generating function of its codeword lengths lies below the R\'{e}nyi entropy.
In addition, this type of bounds was studied in the context of various
coding problems, including guessing (see, e.g., \cite{Arikan96, ArikanM98-1, ArikanM98-2,
BracherHL15, BracherLP17, BunteL14a, CV2014a, CV2014b, HanawalS11, Kuzuoka16,
MerhavAr99, Merhav11, Merhav17, Sundaresan07, verdubook}).

Kontoyiannis and Verd\'{u} \cite{KYSV14} studied the behavior of the best achievable rate and
other fundamental limits in variable-rate lossless source compression without prefix constraints.
In the non-asymptotic regime, the fundamental limits of fixed-to-variable lossless compression
with and without prefix constraints were shown to be tightly coupled. Reference \cite{KYSV14} obtains
non-asymptotic upper and lower bounds on the distribution of codeword lengths, along with
a rigorous proof of the Gaussian approximation put forward in 1962 by Strassen \cite{Strassen} for memoryless
sources. An alternative approach was followed by Courtade and Verd\'{u} in \cite{CV2014a}, where they
derived non-asymptotic bounds for the normalized cumulant generating function of the codeword
lengths for optimal variable-length lossless codes without prefix constraints; these bounds
are used in \cite{CV2014a} to obtain simple proofs of the asymptotic normality and the reliability
function of memoryless sources allowing countably infinite alphabets.

In \cite{KosPV14}, Kostina {\em et al.} studied the fundamental limits of the minimum average
length of lossless and lossy variable-length compression, allowing a nonzero error probability
$\varepsilon \in [0,1)$ for almost lossless compression. The bounds in \cite{KosPV14} were used to obtain
a Gaussian approximation on the speed of convergence of the minimum average length,
which was shown to be quite accurate for all but small blocklengths. In \cite{KogaY05}, Koga
and Yamamoto followed an information-spectrum approach to obtain asymptotic properties of the codeword
lengths for prefix fixed-to-variable source codes, allowing decoding errors. This work was refined
in the non-asymptotic setting by Kuzuoka \cite{Kuzuoka16}, which bounds the cumulant generating
function of the codeword lengths via the smooth R\'{e}nyi entropy.

\subsection{Organization of the paper}
Section~\ref{section: preliminaries} defines the
R\'{e}nyi information measures, and summarizes those properties used in this paper.
Section~\ref{section: bounds for guessing} derives upper and lower bounds on the minimal
guessing moments of a random variable taking a finite number of values where side information
on its value may be available. In the non-asymptotic regime, these bounds improve earlier
results by Arikan \cite{Arikan96} and Bozta\c{s} \cite{Boztas97}.
Section~\ref{section: bounds for guessing} also derives tight lower and upper bounds
which establish relationships between the MAP error probability in $M$-ary hypothesis testing,
and the moments of the optimal guessing function for correctly identifying $X$ when side
information $Y$ is available. The bounds on the guessing moments are applied
in Section~\ref{section: lossless source coding} to optimal variable-length lossless data
compression. We derive improved bounds on the normalized cumulant generating function of the
codeword lengths for fixed-to-variable optimal codes, and on the non-asymptotic reliability function of
discrete memoryless sources, tightening earlier results by Courtade and Verd\'{u} \cite{CV2014a}.
Following up the aforementioned work by Kuzuoka \cite{Kuzuoka16},
Section~\ref{sec: almost-lossless source coding} relies on the techniques in
Sections~\ref{section: bounds for guessing} and~\ref{section: lossless source coding}
in order to derive improved lower bounds on the cumulant generating function of the codeword lengths
for fixed-to-variable source coding allowing errors via the use of the smooth
R\'{e}nyi entropy (\cite{Cachin}, \cite{Koga13}, \cite{RennerW04}, \cite{RennerW05}).
The bounds in Section~\ref{sec: almost-lossless source coding} are derived for
source codes allowing a given maximal or average error probability.

\section{Preliminaries}
\label{section: preliminaries}

The information measures used in this paper apply to discrete random variables.

\begin{definition} \cite{Renyientropy}
\label{definition: Renyi entropy}
Let $X$ be a discrete random variable taking values on a finite
or countably infinite set $\set{X}$, and let $P_X$ be its probability
mass function. The R\'{e}nyi entropy of order
$\alpha \in (0,1) \cup (1, \infty)$
is given by\footnote{Unless explicitly stated, the logarithm base
can be chosen by the reader, with $\exp$ indicating the inverse function
of $\log$.}
\begin{align} \label{eq: Renyi entropy}
H_{\alpha}(X) & = H_{\alpha}(P_X) = \frac1{1-\alpha} \, \log
\sum_{x \in \set{X}} P_X^{\alpha}(x).
\end{align}
By its continuous extension,
\begin{align}
\label{eq: RE of zero order}
& H_0(X) = \log \, \bigl| \{x \in \set{X} \colon
P_X(x) > 0 \} \bigr|, \\
\label{eq: Shannon entropy}
& H_1(X) = H(X), \\
\label{eq: RE of infinite order}
& H_{\infty}(X) = \log \frac1{p_{\max}}
\end{align}
where $H(X)$ is the entropy of $X$, and
\begin{align}  \label{eq: 20171031-a}
p_{\max} = \underset{x \in \set{X}}{\max} \, P_X(x).
\end{align}
\end{definition}
All definitions in this section extend in a natural way to random vectors.
\begin{lemma} \cite{Renyientropy}
\label{lemma: RE - indep.}
Let $X^n = (X_1, \ldots, X_n)$ be an $n$-dimensional random vector with independent components.
Then, for all $\alpha \in [0, \infty]$,
\begin{align} \label{eq: 20171016-a}
H_{\alpha}(X^n) = \sum_{i=1}^n H_{\alpha}(X_i).
\end{align}
\end{lemma}
Note however that in contrast to the Shannon entropy, $H_{\alpha}(X^n)$ may exceed the
right side of \eqref{eq: 20171016-a} when $\alpha \neq 1$ and $\{X_i\}$ are dependent.

\begin{definition} \cite{Renyientropy}
\label{definition: Renyi divergence}
Given probability mass functions $P$ and $Q$ on a finite
or countably infinite set $\set{X}$, the {\em R\'{e}nyi
divergence of order $\alpha > 0$} is defined as follows:
\begin{itemize}
\item
If $\alpha \in (0,1) \cup (1, \infty) $, then
\begin{align}
\label{eq:RD1}
D_{\alpha}(P\|Q) &= \frac1{\alpha-1} \; \log \sum_{x \in \set{X}}
P^\alpha(x) \, Q^{1-\alpha}(x).
\end{align}
\item By the continuous extension of $D_{\alpha}(P \| Q)$, the R\'{e}nyi divergences of
orders 0, 1, and $\infty$ are defined as
\begin{align}
\label{eq: d0}
& D_0(P \| Q) = -\log Q\bigl(\{x \in \set{X} \colon P(x) > 0 \} \bigr), \\
\label{def:d1}
& D_1(P\|Q) = D(P\|Q), \\
\label{def:dinf}
& D_{\infty}(P\|Q) = \log \, \sup_{x \in \set{X}} \frac{P(x)}{Q(x)},
\end{align}
\end{itemize}
where $D(P\|Q)$ denotes the relative entropy.
\end{definition}
Properties of the R\'{e}nyi divergence were studied in
\cite{Atar15}, \cite{ErvenH14}, \cite[Section~8]{SV16}
and \cite{Shayevitz_ISIT11}. The R\'{e}nyi divergence
of negative orders is defined by extending \eqref{eq:RD1}
to $\alpha \in (-\infty, 0)$ \cite[Section~5]{ErvenH14}.

\begin{lemma}\cite[Section~5]{ErvenH14} \label{lemma: RD}
The following properties are satisfied by the R\'{e}nyi divergence:
\begin{itemize}
\item For all $\alpha \neq 0$, $D_{\alpha}(P\|Q)=0$ if and only if $P=Q$.
\item
$D_{\alpha}(P \| Q) \in [0, \infty]$ for $\alpha \in [0, \infty]$, and
$D_{\alpha}(P \| Q) \in [-\infty, 0]$ for $\alpha \in [-\infty, 0)$ (with the continuous extension where
$D_{-\infty}(P\|Q) \triangleq \underset{\alpha \to -\infty}{\lim} D_{\alpha}(P\|Q)$ \cite[(81)]{ErvenH14}).
\end{itemize}
\end{lemma}

\begin{definition} \label{definition: binary RD}
For all $\alpha \in (0,1) \cup (1, \infty)$, the {\em binary R\'{e}nyi divergence of
order $\alpha$}, denoted by $d_{\alpha}(p\|q)$ for $(p,q) \in [0,1]^2$,
is defined as $D_{\alpha}([p ~1-p] \| [q ~1-q])$. It is the continuous
extension to $[0,1]^2$ of
\begin{equation}
\label{eq1: binary RD}
d_\alpha (p \| q ) =\frac1{\alpha-1} \; \log \Bigl(p^{\alpha} q^{1-\alpha}
+ (1-p)^{\alpha} (1-q)^{1-\alpha} \Bigr).
\end{equation}
\end{definition}

\begin{definition} \cite{Arimoto75}
\label{definition: AR conditional entropy}
Let $P_{XY}$ be defined on $\set{X} \times \set{Y}$,
where $X$ is a discrete random variable.
The \textit{Arimoto-R\'{e}nyi conditional entropy of order
$\alpha \in [ 0, \infty]$} of $X$ given $Y$ is defined as follows:
\begin{itemize}
\item
If $\alpha \in (0,1) \cup (1, \infty) $, then
\begin{align}
\label{eq1: Arimoto - cond. RE}
H_{\alpha}(X | Y) &= \frac{\alpha}{1-\alpha} \,
\log \, \expectation \left[
\left( \, \sum_{x \in \set{X}} P_{X|Y}^{\alpha}(x|Y)
\right)^{\frac1{\alpha}} \right] \\
\label{eq2: Arimoto - cond. RE}
&= \frac{\alpha}{1-\alpha} \, \log
\, \sum_{y \in \set{Y}} P_Y(y) \, \exp \left(
\frac{1-\alpha}{\alpha} \;
H_{\alpha}(X | Y=y) \right),
\end{align}
where \eqref{eq2: Arimoto - cond. RE}
applies if $Y$ is a discrete random variable.
\item By its continuous extension, the Arimoto-R\'enyi conditional
entropy of orders $0$, $1$, and $\infty$ are defined as
\begin{align}
\label{eq2: cond. RE at 0:gral}
H_0(X|Y) &= \sup_{y \in \set{Y}} \, H_0(X \, | \, Y=y), \\
\label{eq: cond. RE at 1}
H_1(X|Y) &= H(X|Y), \\
\label{eq: cond. RE at infinity}
H_{\infty}(X|Y) &= \log \, \frac1{\expectation
\Bigl[ \underset{x \in \set{X}}{\max} \, P_{X|Y}(x|Y) \Bigr]}.
\end{align}
\end{itemize}
\end{definition}
Properties of the Arimoto-R\'{e}nyi conditional entropy were studied in
\cite{FehrB14}, \cite{SakaiY_arXiv17} and \cite{SV18}.

As in \cite[Section~4]{SV18}, we find several useful results satisfied by the
Arimoto-R\'{e}nyi conditional entropy of negative orders.

Another R\'{e}nyi information measure used in this paper is
the smooth R\'{e}nyi entropy, introduced by Renner and Wolf \cite{RennerW05}
(after a different definition in \cite{RennerW04}).
\begin{definition} \cite{RennerW05}
\label{def: smooth Renyi entropy}
Let $X$ be a discrete random variable taking values on $\set{X}$, and let
$P_X$ denote the probability mass function of $X$. Let $\alpha \in (0,1) \cup (1, \infty)$
and $\varepsilon \in [0,1)$. The {\em $\varepsilon$-smooth R\'{e}nyi entropy of order $\alpha$}
is given by
\begin{align} \label{eq: smooth RE 1}
H_{\alpha}^{(\varepsilon)}(X) = \frac1{1-\alpha} \; \min_{\mu \in B^{(\varepsilon)}(P_X)}
\log \sum_{x \in \set{X}} \mu^{\alpha}(x)
\end{align}
where $B^{(\varepsilon)}(P_X)$ is the following subset of sub-probability measures defined on $\set{X}$:
\begin{align} \label{eq: smooth RE 2}
B^{(\varepsilon)}(P_X) \triangleq \left\{ \mu \colon \set{X} \to [0,1]:
\, \sum_{x \in \set{X}} \mu(x) \geq 1-\varepsilon,
\quad \mu(x) \leq P_X(x), \; \forall x \in \set{X} \right\}.
\end{align}
The $\varepsilon$-smooth R\'{e}nyi entropy becomes the R\'{e}nyi entropy when $\varepsilon=0$, i.e.,
\begin{align}  \label{eq: smooth RE 3}
H_{\alpha}^{(0)}(X) = H_{\alpha}(X)
\end{align}
for all $\alpha \in (0,1) \cup (1, \infty)$.
\end{definition}
Properties of $H_{\alpha}^{(\varepsilon)}(X)$ were studied in \cite{Koga13} and \cite{RennerW05}.

\begin{lemma} \cite[Theorem~1]{Koga13}
\label{lemma: Koga13}
Let $X$ be a random variable taking values on a finite set
$\set{X}=\{x_1, \ldots, x_M\}$, whose elements are
ordered such that
\begin{align} \label{eq: 20171026-a}
P_X(x_1) \geq P_X(x_2) \geq \ldots \geq P_X(x_M),
\end{align}
and let $\varepsilon \in [0,1)$. Then,
\begin{enumerate}[a)]
\item \label{lemma: Koga13-a}
For $\alpha \in (0,1)$, the minimum in the right side of \eqref{eq: smooth RE 1} is achieved
by $\mu_1 \in B^{(\varepsilon)}(P_X)$ given by
\begin{align}  \label{eq: 20171026-b}
\mu_1(x_i) =
\begin{dcases}
P_X(x_i), & \quad i = 1, \ldots, J_{\varepsilon}-1 \\
1-\varepsilon - \sum_{j=1}^{J_{\varepsilon}-1} P_X(x_j), & \quad i=J_{\varepsilon} \\
0, & \quad i = J_{\varepsilon}+1, \ldots, M
\end{dcases}
\end{align}
with
\begin{align}  \label{eq: 20171026-c}
J_{\varepsilon} = \min \left\{1 \leq j \leq M \colon \sum_{i=1}^j P_X(x_i) \geq 1-\varepsilon \right\}.
\end{align}
\item \label{lemma: Koga13-b}
For $\alpha > 1$, the minimum in the right side of \eqref{eq: smooth RE 1} is achieved
by $\mu_2 \in B^{(\varepsilon)}(P_X)$ given by
\begin{align}  \label{eq: 20171026-d}
\mu_2(x_i) =
\begin{dcases}
\beta, & \quad i = 1, \ldots, K_{\beta} \\
P_X(x_i), & \quad i = K_{\beta}+1, \ldots, M
\end{dcases}
\end{align}
where $\beta \in (0,1)$ and $K_{\beta} \in \{1, \ldots, M\}$ are jointly selected such that
\begin{align}
\label{eq: 20171026-e}
& \sum_{i=1}^M \mu_2(x_i) = 1-\varepsilon, \\
\label{eq: 20171026-f}
& K_{\beta} = \max \bigl\{ 1 \leq j \leq M \colon P_X(x_j) \geq \beta \bigr\}.
\end{align}
\end{enumerate}
\end{lemma}
\begin{remark} \label{remark: koga13}
The sub-probability measures $\mu_1, \mu_2 \in B^{(\varepsilon)}(P_X)$ in Lemma~\ref{lemma: Koga13}
are {\em independent} of $\alpha$ for $\alpha \in (0,1)$ or $\alpha > 1$, respectively.
\end{remark}

\begin{lemma} \label{lemma: 20171029}
Under the assumption in Lemma~\ref{lemma: Koga13}, the following inequalities hold for $\varepsilon \in (0,1)$:
\begin{enumerate}[a)]
\item If $\alpha \in (0,1)$
\begin{align}  \label{eq: 20171029-a}
\frac1{\alpha-1} \, \log \frac1{1-\varepsilon} \leq H_{\alpha}^{(\varepsilon)}(X)
\leq \frac1{\alpha-1} \, \log \left( \frac1{1-\varepsilon} \right) + \log \left( \frac1{\mu_1(x_{_{\hspace*{-0.5mm} J_{\varepsilon}}})} \right)
\end{align}
with $J_{\varepsilon}$ as defined in \eqref{eq: 20171026-c}.
\item If $\alpha > 1$
\begin{align}  \label{eq: 20171029-b}
\frac1{\alpha-1} \, \log \frac1{1-\varepsilon} \leq H_{\alpha}^{(\varepsilon)}(X)
\leq \frac1{\alpha-1} \, \log \left( \frac1{1-\varepsilon} \right) + \log \left(\frac1{\min\{P_{\min}, \beta\}}\right)
\end{align}
where $P_{\min}$ denotes the minimal non-zero mass of $P_X$, and $\beta$ is defined in \eqref{eq: 20171026-d}--\eqref{eq: 20171026-f}.
\end{enumerate}
\end{lemma}
\begin{proof}
From \eqref{eq: 20171026-b} and \eqref{eq: 20171026-d}, it follows that for all $x \in \set{X}$
\begin{align}
\label{eq: 20171029-c}
& \mu_1(x) \leq \mu_1^{\alpha}(x) \leq \mu_1(x) \, \mu_1^{\alpha-1}(x_{_{\hspace*{-0.5mm} J_{\varepsilon}}}), \quad \quad \alpha \in (0,1), \\
\label{eq: 20171029-d}
& \mu_2(x) \, \min\{P_{\min}^{\alpha-1}, \beta^{\alpha-1}\} \leq \mu_2^{\alpha}(x) \leq \mu_2(x), \quad \alpha > 1.
\end{align}
The bounds in \eqref{eq: 20171029-a} and \eqref{eq: 20171029-b} can be verified from Definition~\ref{def: smooth Renyi entropy},
\eqref{eq: 20171026-b}, \eqref{eq: 20171026-e}, \eqref{eq: 20171029-c} and \eqref{eq: 20171029-d}.
\end{proof}
\begin{remark} \label{remark: 20171029}
The left inequality in \eqref{eq: 20171029-a} can be obtained from \cite[Lemma~2]{RennerW05}.
\end{remark}

The following result readily follows from Lemma~\ref{lemma: 20171029}.
\begin{corollary} \label{corollary: 20171028}
If $X$ takes values on a finite set, then for all $\varepsilon \in (0,1)$,
\begin{align}
\label{eq: 20171028-b}
& \lim_{\alpha \uparrow 1} H_{\alpha}^{(\varepsilon)}(X) = -\infty, \\
\label{eq: 20171028-c}
& \lim_{\alpha \downarrow 1} H_{\alpha}^{(\varepsilon)}(X) = +\infty.
\end{align}
\end{corollary}

\begin{remark} \label{remark: 20171028}
In contrast to $H_{\alpha}(X)$ which is non-negative, continuous and monotonically non-increasing
in $\alpha$ when $X$ takes values on a finite set, Corollary~\ref{corollary: 20171028} shows that these
properties do not fully extend to $H_{\alpha}^{(\varepsilon)}(X)$ with $\varepsilon \in (0,1)$.
Bearing in mind the discontinuity shown in Corollary~\ref{corollary: 20171028}, $H_{\alpha}^{(\varepsilon)}(X)$
is monotonically non-increasing on both $\alpha \in (0,1)$ and $\alpha \in (1,\infty)$ \cite[Lemma~1]{RennerW05}.
\end{remark}

\section{Improved Bounds on Guessing Moments}
\label{section: bounds for guessing}

This section provides improved upper and lower bounds on the guessing moments
of a discrete random variable. The upper bounds in this section correspond to
the case where the guessing function is a ranking function.

\subsection{Key result}
\label{subsec: key result}

\begin{theorem} \label{thm: key result}
Given a discrete random variable $X$ taking values on a set $\set{X}$,
a function $g \colon \set{X} \to (0, \infty)$, and a scalar $\rho \neq 0$, then
\begin{enumerate}[1)]
\item \label{item: key1}
\begin{align}
& \sup_{\beta \in (-\rho, +\infty) \setminus \{0\}}
\frac{1}{\beta} \left[ H_{\frac{\beta}{\beta+\rho}}(X)
- \log \sum_{x \in \set{X}} g^{-\beta}(x) \right] \nonumber \\[0.1cm]
\label{eq: key result}
\leq & \frac1\rho \log \expectation[g^{\rho}(X)] \\[-0.1cm]
\label{eq: key result-b}
\leq & \inf_{\beta \in (-\infty, -\rho) \setminus \{0\}}
\frac{1}{\beta} \left[ H_{\frac{\beta}{\beta+\rho}}(X)
- \log \sum_{x \in \set{X}} g^{-\beta}(x) \right].
\end{align}
\item \label{item: key3}
For $\tau \in \Reals$, define the probability mass function
\begin{align} \label{eq: 20171013a}
Q_{\tau}(x) = \frac{g^{-\tau}(x)}{\underset{a \in \set{X}}{\sum} g^{-\tau}(a)}, \quad x \in \set{X},
\end{align}
provided that the sum in the right side of \eqref{eq: 20171013a} is finite. The following results hold:
\begin{enumerate}[a)]
\item \label{item: key3-a}
If $P_X = Q_{\rho}$ and $\set{X}$ is a finite set, then
\begin{align}
\label{eq: 20171014-a}
\frac1\rho \log \expectation[g^{\rho}(X)] =
-\frac1\rho \log \left( \frac1{|\set{X}|} \sum_{x \in \set{X}} g^{-\rho}(x) \right).
\end{align}
\item \label{item: key3-b}
If $P_X = Q_{\nu}$ with $\nu>0$ and $\nu \neq \rho$, then the supremum in the left side of
\eqref{eq: key result} is attained at $\beta = \nu - \rho$, and the inequality in \eqref{eq: key result}
becomes an identity. Conversely, if \eqref{eq: key result} is an identity and the supremum is attained at
$\beta^{\ast} \in (-\rho, +\infty) \setminus \{0\}$, then $P_X = Q_{\rho+\beta^{\ast}}$.
\item \label{item: key3-c}
If $P_X = Q_{\nu}$ with $\nu<0$ and $\nu \neq \rho$, then the infimum in the right side of
\eqref{eq: key result-b} is attained at $\beta = \nu - \rho$, and the inequality in \eqref{eq: key result-b}
becomes an identity. Conversely, if \eqref{eq: key result-b} is an identity and the infimum is attained at
$\beta^{\ast} \in (-\infty, -\rho) \setminus \{0\}$, then $P_X = Q_{\rho+\beta^{\ast}}$.
\end{enumerate}
\end{enumerate}
\end{theorem}

\begin{proof}
It is instructive to prove a weaker result first where instead of optimizing with respect to $\beta$,
we simply take $\beta=1$ in the upper/lower bound.

Let $\alpha \neq 0$, and let $R_{\alpha}$ be the scaled probability mass function defined by
\begin{align} \label{R dist}
& R_{\alpha}(x) = \frac{P_X^{\frac1{\alpha}}(x)}{\underset{a \in \set{X}}{\sum} P_X^{\frac1{\alpha}}(a)}, \quad x \in \set{X}.
\end{align}
Then,
\begin{align}
D_{1+\rho}(R_{1+\rho} \| Q_1)
& = \frac1{\rho} \, \log \left( \sum_{x \in \set{X}} P_X(x) g^{\rho}(x) \right)
+ \log \left( \sum_{x \in \set{X}} \frac1{g(x)} \right)
- \frac{1+\rho}{\rho} \, \log \left( \sum_{x \in \set{X}} P_X^{\frac1{1+\rho}}(x) \right) \label{eq2} \\
&= \frac1{\rho} \log \expectation[g^{\rho}(X)] + \log \left( \sum_{x \in \set{X}} \frac1{g(x)} \right)
- H_{\frac1{1+\rho}}(X) \label{eq3}
\end{align}
where \eqref{eq2} follows from \eqref{eq:RD1}, \eqref{eq: 20171013a} and
\eqref{R dist}; \eqref{eq3} is due to \eqref{eq: Renyi entropy}.
Since $D_{1+\rho}(R_{1+\rho} \| Q_1) \geq 0$ for $\rho > -1$ (see Lemma~\ref{lemma: RD}),
the non-negativity of the right side of \eqref{eq3} implies that
\begin{align} \label{eq: 20171013b}
\frac1\rho \, \log \expectation[g^{\rho}(X)] \geq
\frac{1}{\beta} \left[ H_{\frac{\beta}{\beta+\rho}}(X)
- \log \sum_{x \in \set{X}} g^{-\beta}(x) \right]
\end{align}
holds for $\beta=1$ and $\rho \in (-1, 0) \cup (0, \infty)$; furthermore, since
$D_{1+\rho}(R_{1+\rho} \| Q_1) \leq 0$ for $\rho < -1$ (see Lemma~\ref{lemma: RD}), the
non-positivity of the right side of \eqref{eq3} implies that
\begin{align} \label{eq: 20171013c}
\frac1\rho \, \log \expectation[g^{\rho}(X)] \leq
\frac{1}{\beta} \left[ H_{\frac{\beta}{\beta+\rho}}(X)
- \log \sum_{x \in \set{X}} g^{-\beta}(x) \right]
\end{align}
holds for $\beta=1$ and $\rho \in (-\infty, -1)$.

We proceed to show \eqref{eq: 20171013b} and \eqref{eq: 20171013c} for the range of $\beta$
specified in \eqref{eq: key result} and \eqref{eq: key result-b} respectively.
Consider next $\rho \neq 0$ and $\beta \neq 0$. Let $\widetilde{\rho} \triangleq \frac{\rho}{\beta}$.

To prove \eqref{eq: key result}:
\begin{enumerate}[i)]
\item If $\beta \in (-\rho, 0)$ with $\rho \in (0, \infty)$, then $\widetilde{\rho} \in (-\infty, -1)$.
Since $1 \in (-\infty, -\widetilde{\rho})$, \eqref{eq: 20171013b} follows from the specialized version
of \eqref{eq: 20171013c} with $\beta \leftarrow 1$ and $(\rho, g) \leftarrow (\widetilde{\rho}, g^{\beta})$.
\item If $\beta \in (0, \infty)$ with $\rho \in (0, \infty)$, then $\widetilde{\rho} \in (0, \infty)$;
and if $\beta \in (-\rho, \infty)$ with $\rho \in (-\infty, 0)$, then $\widetilde{\rho} \in (-1, 0)$.
In both cases, $1 \in (-\widetilde{\rho}, \infty)$ and, consequently, \eqref{eq: 20171013b}
follows from its specialized version with $\beta \leftarrow 1$ and $(\rho, g) \leftarrow (\widetilde{\rho}, g^{\beta})$.
\end{enumerate}

To prove \eqref{eq: key result-b}:
\begin{enumerate}[i)]
\setcounter{enumi}{2}
\item If $\beta \in (-\infty, -\rho)$ with $\rho \in (0, \infty)$, then $\widetilde{\rho} \in (-1, 0)$;
and if $\beta \in (-\infty, 0)$ with $\rho \in (-\infty, 0)$, then $\widetilde{\rho} \in (0, \infty)$.
In both cases $1 \in (-\widetilde{\rho}, \infty)$, and therefore
\eqref{eq: 20171013c} follows from the specialized version of \eqref{eq: 20171013b}
with $\beta \leftarrow 1$ and $(\rho, g) \leftarrow (\widetilde{\rho}, g^{\beta})$.
\item If $\beta \in (0, -\rho)$ with $\rho \in (-\infty, 0)$, then $\widetilde{\rho} \in (-\infty, -1)$; this
yields \eqref{eq: 20171013c} from its specialized version with $\beta \leftarrow 1$
and $(\rho, g) \leftarrow (\widetilde{\rho}, g^{\beta})$.
\end{enumerate}

To prove Item~\ref{item: key3}):
\begin{itemize}
\item
Suppose the set $\set{X}$ is finite. By letting $\tau = \nu$ in \eqref{eq: 20171013a}, the identity in \eqref{eq: 20171014-a} follows easily.
\item
Suppose that $P_X = Q_{\nu}$ with $\nu>0$ and $\nu \neq \rho$. Let $\beta^\ast = \nu - \rho$ and
let $(\rho, g) \leftarrow (\frac{\rho}{\beta^\ast}, g^{\beta^\ast})$, which yields $Q_1 \leftarrow Q_{\beta^\ast}$
(see \eqref{eq: 20171013a}), $R_{1+\rho} \leftarrow R_{\nu/{\beta^\ast}}$, and $\frac1{1+\rho} \leftarrow \frac{\beta^\ast}{\beta^\ast+\rho}$.
Then, from \eqref{eq2}--\eqref{eq3},
\begin{align}
\label{eq: 20171014-b}
D_{\nu/{\beta^\ast}}(R_{\nu/{\beta^\ast}} \| Q_{\beta^\ast})
&= \frac{\beta^\ast}{\rho} \log \expectation[g^{\rho}(X)] + \log \left( \sum_{x \in \set{X}} \frac1{g^{\beta^\ast}(x)} \right)
- H_{\frac{\beta^\ast}{\beta^\ast+\rho}}(X).
\end{align}
Since by assumption $P_X = Q_{\nu}$, it is easy to verify from \eqref{eq: 20171013a} and \eqref{R dist} that $R_{\nu/{\beta^\ast}} = Q_{\beta^\ast}$.
Hence, both sides of \eqref{eq: 20171014-b} are equal to zero, which therefore implies that the supremum in the left side of
\eqref{eq: key result} is attained at $\beta = \beta^\ast \in (-\rho, +\infty) \setminus \{0\}$, and the inequality in
\eqref{eq: key result} becomes an identity. This proves the first part of Item~\ref{item: key3-b}). To prove its second part,
assume that \eqref{eq: key result} is an identity and the supremum is attained at $\beta^{\ast}$ in the range of $\beta$
specified in \eqref{eq: key result}. This implies that the right side of \eqref{eq: 20171014-b} is equal to zero, and
so is its left side. In view of Lemma~\ref{lemma: RD}, since $\nu=0$, it follows that $R_{\nu/{\beta^\ast}} = Q_{\beta^\ast}$
which then gives that $P_X = Q_{\nu}$.
\item
The proof of Item~\ref{item: key3-c}) is analogous to the proof of Item~\ref{item: key3-b}).
\end{itemize}
\end{proof}

\begin{remark} \label{remark: alternative proof - Th 1}
Theorem~\ref{thm: key result} can be proved in the following alternative way:
For $\beta \neq 0$, let
\begin{align} \label{eq: 20171013d}
\alpha = \frac{\rho+\beta}{\beta}.
\end{align}
It can be verified from \eqref{eq:RD1}, \eqref{eq: 20171013a} and \eqref{R dist} that
\begin{align} \label{eq4}
D_{\alpha}(R_\alpha \| Q_{\beta}) &= \frac{\beta}{\rho} \, \log \expectation[g^{\rho}(X)]
+ \log \left( \sum_{x \in \set{X}} g^{-\beta}(x) \right) - H_{\frac{\beta}{\beta+\rho}}(X),
\end{align}
and then, in view of Lemma~\ref{lemma: RD}, the following cases are considered for the free parameter $\beta$:
\begin{enumerate}[i)]
\item If $\beta \in (-\rho, \infty)$ and $\beta>0$, then $\alpha > 0$ and $D_{\alpha}(R_\alpha \| Q_\beta) \geq 0$;
hence, the right side of \eqref{eq4} is non-negative, and dividing it by $\beta$ gives \eqref{eq: 20171013b};
\item If $\beta \in (-\rho, \infty)$ and $\beta<0$, then $\alpha < 0$, and therefore $D_{\alpha}(R_\alpha \| Q_\beta) \leq 0$;
since the right side of \eqref{eq4} is non-positive, dividing it by $\beta$ gives \eqref{eq: 20171013b};
\item If $\beta \in (-\infty, -\rho)$ and $\beta<0$, then $\alpha > 0$ and $D_{\alpha}(R_\alpha \| Q_\beta) \geq 0$;
hence, the right side of \eqref{eq4} is non-negative, and dividing it by $\beta$ gives \eqref{eq: 20171013c};
\item If $\beta \in (-\infty, -\rho)$ and $\beta>0$, then $\alpha < 0$, and therefore $D_{\alpha}(R_\alpha \| Q_\beta) \leq 0$;
since the right side of \eqref{eq4} is non-positive, dividing it by $\beta$ gives \eqref{eq: 20171013c}.
\end{enumerate}
This gives the lower and upper bounds in \eqref{eq: key result} and \eqref{eq: key result-b}, respectively,
after an optimization of the right sides of \eqref{eq: 20171013b}
and \eqref{eq: 20171013c} over the free parameter $\beta \in (-\rho, \infty) \setminus \{0\}$ and
$\beta \in (-\infty, -\rho) \setminus \{0\}$, respectively. Item~\ref{item: key3}) can be proved
in a similar way to our earlier proof by relying on \eqref{eq: 20171013d}, \eqref{eq4},
and Lemma~\ref{lemma: RD}; note that in view of \eqref{eq: 20171013a}, \eqref{R dist} and
\eqref{eq: 20171013d}, if $\beta = \nu-\rho$, then $R_\alpha=Q_{\beta}$ if and only if
$P_X = Q_\nu$.
\end{remark}

\begin{remark}
For $\rho > 0$, the supremum over $\beta$ in the right side of \eqref{eq: key result}
involves negative orders of the R\'{e}nyi entropy whenever $\beta \in (-\rho, 0)$.
As shown in Example~\ref{example: optimal beta}, the optimal value of
$\beta \in (-\rho, \infty) \setminus \{0\}$ can be negative; furthermore, for
every such $\beta$, Theorem~\ref{thm: key result} asserts the existence of a
probability mass function for which \eqref{eq: key result} is achieved with
equality. Allowing R\'{e}nyi entropy of negative orders in
Theorem~\ref{thm: key result} is therefore beneficial.
\end{remark}

\par
The particularization of Item~\ref{item: key1}) in Theorem~\ref{thm: key result}
to $\beta=1$ yields the following, generally looser, bound:
\begin{corollary}\cite[Lemma~2]{CV2014a} \label{corollary: beta=1}
Let $X$ and $g$ be as in Theorem~\ref{thm: key result},
and $\rho \in (-1, 0) \cup (0, \infty)$. Then,
\begin{align} \label{eq: CV2014a - Lemma2}
\frac1{\rho} \, \log \expectation \bigl[ g^{\rho}(X) \bigr]
\geq H_{\frac1{1+\rho}}(X)-\log \sum_{x \in \set{X}} \frac1{g(x)}.
\end{align}
\end{corollary}

\begin{remark}
The proof of Corollary~\ref{corollary: beta=1}, which serves
as a first step in proving Theorem~\ref{thm: key result},
differs from its proof in \cite[Lemma~2]{CV2014a} (see also \cite[Theorem~1]{Arikan96}
for a specialized version where $\rho > 0$). The proofs in \cite{Arikan96} and
\cite{CV2014a} rely on the reverse H\"{o}lder inequality, whereas
the proof here is based on the R\'{e}nyi divergence.
\end{remark}

\subsection{Lower bounds}
\label{subsec: LB guessing}

Theorem~\ref{thm: key result} is applied in this section to derive lower bounds on
guessing moments with or without side information, improving the bounds in \cite{Arikan96}.

\begin{theorem} \label{theorem: improving Arikan's bound}
Let $X$ be a random variable taking values on the
finite set $\set{X}=\{1, \ldots, M\}$, and let
$g \colon \set{X} \to \set{X}$ be
an arbitrary guessing function.
Then, for every $\rho \neq 0$,
\begin{align} \label{eq: improving Arikan's bound}
\frac1{\rho} \, \log \expectation\bigl[g^{\rho}(X)\bigr]
\geq \sup_{\beta \in (-\rho, \infty) \setminus \{0\}} \frac1{\beta} \,
\Bigl[ H_{\frac{\beta}{\beta+\rho}}(X) - \log u_M(\beta) \Bigr]
\end{align}
with
\begin{align}  \label{eq: u}
u_M(\beta) =
\begin{dcases}
\log_{\mathrm{e}} M + \gamma + \frac1{2M}-\frac5{6(10M^2+1)}
& \quad \beta=1, \\[0.4cm]
\min\left\{\zeta(\beta) -
\frac{(M+1)^{1-\beta}}{\beta-1} -\frac{(M+1)^{-\beta}}{2}, \, u_M(1) \right\}
& \quad \beta > 1, \\[0.4cm]
1 + \frac{1}{1-\beta} \, \left[\left(M+\tfrac12\right)^{1-\beta}-\left(\tfrac32\right)^{1-\beta}\right]
& \quad |\beta|<1, \\[0.4cm]
\frac{M^{1-\beta}-1}{1-\beta} + \tfrac12 \left(1+M^{-\beta}\right)
& \quad \beta \leq -1
\end{dcases}
\end{align}
where $\gamma \approx 0.5772$ is the Euler-Mascheroni constant,
and $\zeta(\beta) = \overset{\infty}{\underset{n=1}{\sum}} \frac1{n^\beta}$
is the Riemann zeta function for $\beta>1$.
\end{theorem}

\begin{proof}
Since $g \colon \set{X} \to \set{X}$ is a one-to-one function with $\set{X}=\{1, \ldots, M\}$,
\begin{align}
\sum_{x \in \set{X}} g^{-\beta}(x) = \sum_{i=1}^M \frac1{i^\beta}.
\end{align}
In view of \eqref{eq: key result}, we derive \eqref{eq: improving Arikan's bound}
by proving that
\begin{subequations}
\begin{align}
\label{eq1: UB on harmonic sum}
& \sum_{i=1}^M \frac1{i^\beta} \leq u_M(\beta),
\quad \beta \geq 0, \\
\label{eq2: UB on harmonic sum}
& \sum_{i=1}^M \frac1{i^\beta} \geq u_M(\beta),
\quad \beta \leq 0
\end{align}
\end{subequations}
where $u_M(\beta)$ is given in \eqref{eq: u}. Note that the restriction
$\beta > -\rho$ in \eqref{eq: improving Arikan's bound} is due to
\eqref{eq: key result}.
The proof of \eqref{eq1: UB on harmonic sum} and \eqref{eq2: UB on harmonic sum}
is deferred to Appendix~\ref{appendix: UB on harmonic sum}.
\end{proof}

\begin{remark} \label{remark: Arikan1}
Specializing Theorem~\ref{theorem: improving Arikan's bound} to $\beta=1$
and using $u_M(1) \leq 1 + \log_{\mathrm{e}} M$ for $M \geq 2$, we obtain
\begin{align} \label{eq: Arikan's bound}
\frac1{\rho} \; \log \expectation\bigl[g^{\rho}(X)\bigr]
\geq H_{\frac{1}{1+\rho}}(X) - \log \bigl(1 + \log_{\mathrm{e}} M\bigr)
\end{align}
for $\rho \in (-1, \infty)$. This bound was obtained in the range
$\rho > 0$ by Arikan \cite[(1)]{Arikan96}.
\end{remark}

The following remark justifies the utility of Theorem~\ref{theorem: improving Arikan's bound}.
\begin{remark}
Since Theorem~\ref{thm: key result} also applies to guessing functions, it gives a
lower bound on $\frac1{\rho} \log \expectation\bigl[g^{\rho}(X)\bigr]$ where
$u_M(\beta)$ in \eqref{eq: improving Arikan's bound}--\eqref{eq: u} is replaced
by the finite sum $\overset{M}{\underset{j=1}{\sum}} \frac1{j^\beta}$ for
$\beta \in (-\rho, \infty) \setminus \{0\}$. While numerical
evidence shows that it is slightly better than the bound in
Theorem~\ref{theorem: improving Arikan's bound} for large $M$,
the latter bound is much easier to compute if $M$ is large.
\end{remark}

\par
The following simple example illustrates the improvement afforded
by the lower bound in Theorem~\ref{theorem: improving Arikan's bound},
as well as the sub-optimality of $\beta=1$.
\begin{example}
\label{example: optimal beta}
Let $X$ be geometrically distributed restricted to $\{1, \ldots , M\}$ with
the probability mass function
\begin{align} \label{eq: geometric dist.}
P_X(k) = \frac{(1-a) \, a^{k-1}}{1-a^M}, \quad k \in \{1, \ldots, M\}
\end{align}
where $a = \tfrac{24}{25}$ and $M = 128$. Since $P_X$ is monotonically
decreasing on $\{1, \ldots, M\}$, the optimal guessing function
is $g_X(i)=i$ for $i \in \{1, \ldots, M\}$.
Figure~\ref{figure: 6th moment} compares
$ \tfrac16 \, \log_{\mathrm{e}} \expectation \bigl[ g_X^6(X) \bigr]$
with its lower bounds in \eqref{eq: improving Arikan's bound} and
\eqref{eq: Arikan's bound}.
\begin{figure}[h]
\begin{center}
\includegraphics[width=10cm]{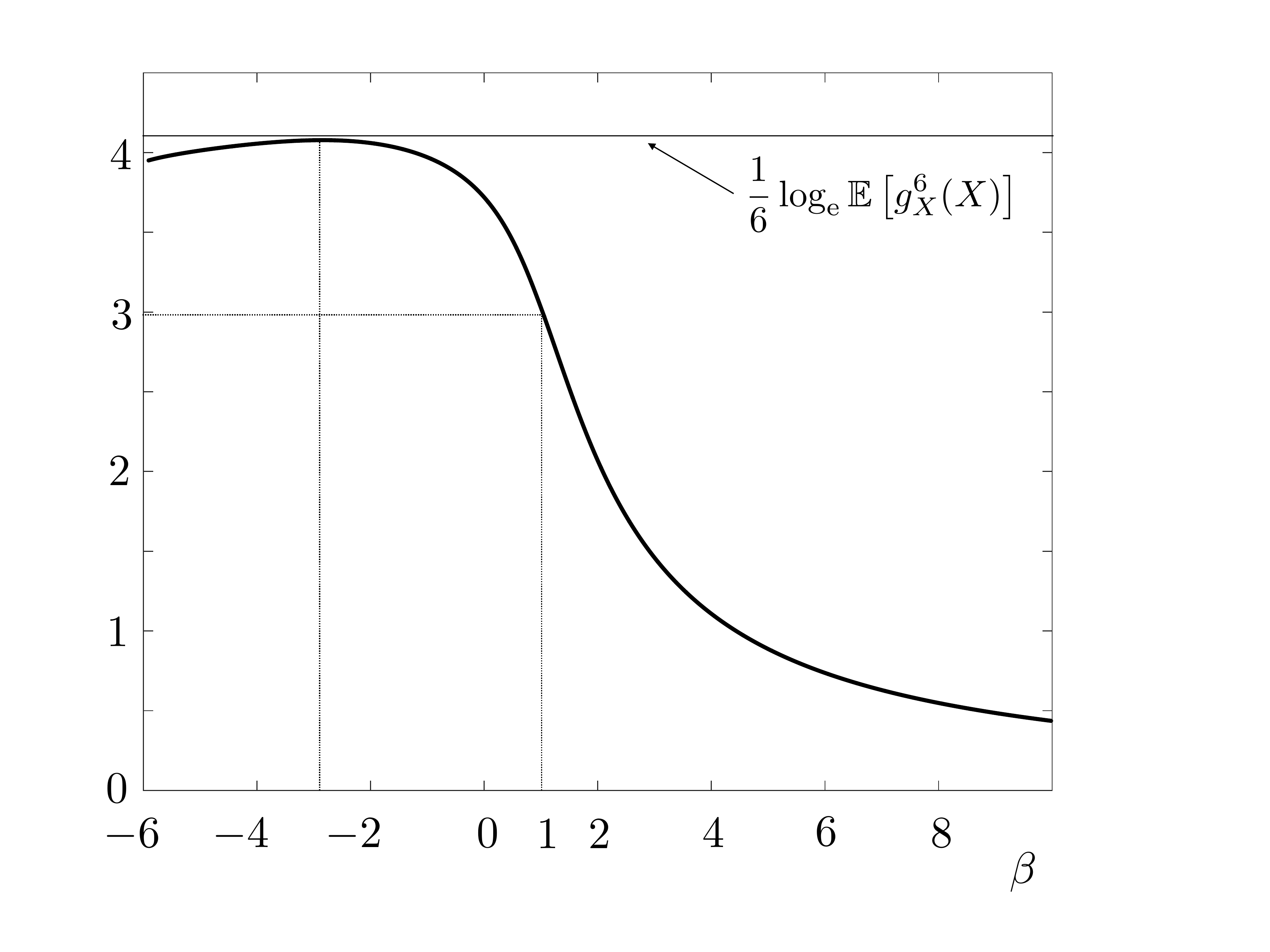}
\end{center}
\vspace*{-1cm}
\caption{\label{figure: 6th moment} Comparison of
$ \tfrac16 \log_{\mathrm{e}} \expectation \bigl[ g_X^6(X) \bigr] = 4.084$ for the
random variable $X$ in Example~\ref{example: optimal beta} with its lower bounds.
The lower bound in \eqref{eq: Arikan's bound} is equal to $2.953$, while
Theorem~\ref{theorem: improving Arikan's bound} improves this lower bound to $4.078$
(attained at $\beta = -2.85$).}
\end{figure}
The numerical results exemplify the sub-optimality of $\beta=1$ in the right side of
\eqref{eq: improving Arikan's bound}.
Not only is the improvement over the lower bound in \cite[(1)]{Arikan96} significant,
but the improved lower bound is very close to the actual value of the sixth guessing
moment.
Furthermore, the improved lower bound is attained at $\beta = -2.85$,
thereby showing the benefit of allowing negative orders in the definition of
R\'{e}nyi entropy.
\end{example}

\begin{remark}
Massey \cite{Massey94} obtained a lower bound on the expected value of
the minimal number of guesses for correctly guessing the value of a discrete
random variable $X$ which does not necessarily take a finite number of values:
\begin{align}
\label{eq: Massey94}
H(X) &\leq \expectation[g(X)] \, h\left(\frac1{\expectation[g(X)]}\right) \\
\label{eq-b: Massey94}
&\leq \log \expectation[g(X)] + \log \mathrm{e},
\end{align}
where $h$ is the binary entropy function,
and \eqref{eq: Massey94} is achieved with equality if and only if $X$ is
geometrically distributed. The idea behind \eqref{eq: Massey94} is that
$H(X) = H(g(X))$ (since $X$ and $g(X)$ are in one-to-one correspondence),
and the entropy of the positive integer-valued random variable $g(X)$ with
a given value of its expectation is maximized when $X$ is geometrically
distributed, which gives the right side of \eqref{eq: Massey94}. Note that
\eqref{eq: Massey94} provides an implicit lower bound on $\expectation[g(X)]$
as a function of the Shannon entropy $H(X)$, whereas \eqref{eq-b: Massey94}
gives a looser though explicit lower bound.

Massey's bound \eqref{eq: Massey94} can be generalized to obtain a
lower bound on the $\rho$-th moment of the guessing function as a function
of $H_{\alpha}(X)$ for arbitrary $\alpha, \rho >0$.
To this end, we first rely on the basic equality
\begin{align} \label{20171015-a}
H_{\alpha}(X) = H_{\alpha}(g(X))
\end{align}
since $g$ is a one-to-one function, and then we seek the distribution of
the random variable $Z = g(X)$, taking values on $\{1, \ldots, M\}$ with possibly
$M=\infty$, which maximizes $H_{\alpha}(Z)$ under the equality constraint
\begin{align}  \label{20171015-b}
\expectation[Z^\rho] = \theta
\end{align}
for fixed $\theta > 1$.
The problem of maximizing the R\'{e}nyi entropy under general equality constraints was
studied by Bunte and Lapidoth (see \cite[Theorem~II.1]{BunteL_Eilat14}) in the setting
of continuous random variables. In the discrete setting, a similar proof
which relies on the non-negativity of Sundaresan's divergence \cite{Sundaresan07}
provides the maximizing distribution of the R\'{e}nyi entropy $H_{\alpha}(Z)$
for a fixed value of $\expectation[Z^\rho]$:
\begin{align}  \label{eq: pmf}
P_Z(n) =
\begin{dcases}
c \left(1+ \lambda n^\rho\right)^{\frac1{\alpha-1}} \, 1\{1 \leq n \leq M \},
& \quad \alpha \in (0,1), \; \lambda > -M^{-\rho}  \\
c \left(1+ \lambda n^\rho\right)^{\frac1{\alpha-1}} \, 1\{n \in \set{T}\},
& \quad \alpha > 1
\end{dcases}
\end{align}
where $\set{T} \triangleq \set{T}_{\lambda, \rho} = \{1, \ldots, M\} \cap \{n: 1 + \lambda n^{\rho} \geq 0\}$,
and $\lambda \in \Reals$ is set to satisfy \eqref{20171015-b}
with the normalizing constant $c$ in \eqref{eq: pmf}.
Hence, we get from \eqref{20171015-a} and \eqref{eq: pmf} that
\begin{align} \label{20171015-b2}
H_{\alpha}(X) \leq H_{\alpha}(Z)
\end{align}
with equality for the maximizing distribution $P_Z$ in \eqref{eq: pmf} which is selected to
satisfy \eqref{20171015-b}. An analogous way to view this problem is to
minimize $\expectation[g^{\rho}(X)]$ for a given value of $H_{\alpha}(X)$.
This leads to an extension of Massey's bound in \eqref{eq: Massey94} which, however,
does not lend itself to a closed-form generalized bound and its
computation is rather involved (especially, for $\alpha > 1$). The lower bound in
Theorem~\ref{theorem: improving Arikan's bound}, on the other hand, can be calculated
more easily.
\end{remark}

\subsection{Upper bounds}
\label{subsec: UB guessing}

The average number of guesses is minimized by taking the guessing
function to be the ranking function $g_X$, for which $g_X(x)=k$
if $P_X(x)$ is the $k$-th largest mass \cite{Massey94}. Although
the tie breaking affects the choice of $g_X$, the distribution of
$g_X(X)$ does not depend on how ties are resolved. Not only does
this strategy minimize the average number of guesses, but it also
minimizes the $\rho$-th moment of the number of guesses for every
$\rho > 0$.

\begin{theorem} \cite[Proposition~4]{Arikan96}
\label{thm: Arikan's UB}
Let $X$ be a discrete random variable taking values on a set
$\set{X}$, and let $g_X$ be the ranking function according to $P_X$.
Then, for all $\rho > 0$,
\begin{align} \label{eq: Arikan96 - prop. 4 without conditioning}
\expectation[g_{X}^{\rho}(X)]
\leq \exp\Bigl(\rho H_{\frac1{1+\rho}}(X) \Bigr).
\end{align}
\end{theorem}

The following result tightens Theorem~\ref{thm: Arikan's UB}.
\begin{theorem} \label{theorem: UB-NSI}
Under the assumptions in Theorem~\ref{thm: Arikan's UB}, for all $\rho \geq 0$,
\begin{align} \label{eq1: UB-NSI}
\expectation[g_{X}^{\rho}(X)]
\leq \frac1{1+\rho} \left[ \exp\Bigl(\rho H_{\frac1{1+\rho}}(X) \Bigr)
- 1 \right] + \exp \Bigl( (\rho-1)^+ H_{\frac1{\rho}}(X) \Bigr)
\end{align}
where $(x)^+ \triangleq \max\{x,0\}$ for $x \in \Reals$.
\end{theorem}

\begin{proof}
From \cite[Lemma~2]{Boztas97}, if $\rho \geq 0$, then
\begin{align}
\expectation\bigl[g_{X}^{1+\rho}(X)\bigr] -
\expectation\bigl[\bigl(g_{X}(X)-1\bigr)^{1+\rho}\bigr] \leq
\left( \sum_{x \in \set{X}} P_X^{\frac1{1+\rho}}(x)\right)^{1+\rho}
= \exp\Bigl(\rho H_{\frac1{1+\rho}}(X) \Bigr).
\label{eq: from Boztas' paper}
\end{align}
At this point we deviate from the analysis in \cite{Boztas97}.
For $\rho \geq 1$, let $r \colon [1, \infty) \to \Reals$
and $v \colon [0, \infty) \to \Reals$ be given by
\begin{align}
& r(u) = \frac{u^{1+\rho}-(u-1)^{1+\rho}-1}{1+\rho} - (u-1)^\rho,
\quad u \geq 1 \label{eq: r fun} \\
& v(u) = u^{\rho}, \quad u \geq 0. \label{eq: v fun}
\end{align}
Expressing the derivative of \eqref{eq: r fun} with the aid of \eqref{eq: v fun}
implies that, for all $u \geq 1$,
\begin{align}
r'(u) = v(u) - v(u-1) - v'(u-1) \geq 0,
\end{align}
where the inequality is due to the convexity of $v(\cdot)$ in $(0,\infty)$.
Since $r(1)=0$, it follows that $r(\cdot)$ is non-negative in $[1,\infty)$.
Invoking the result in \eqref{eq: from Boztas' paper}, along with $\expectation[r(g_X(X))] \geq 0$,
implies that for $\rho \geq 1$
\begin{align}  \label{haifa}
\expectation\bigl[\bigl(g_{X}(X)-1\bigr)^{\rho}\bigr] \leq
\frac1{1+\rho} \left[ \exp\Bigl(\rho H_{\frac1{1+\rho}}(X) \Bigr) - 1 \right].
\end{align}
A replacement of $\rho$ by $\rho-1$ in
\eqref{eq: from Boztas' paper} gives that for $\rho \geq 1$
\begin{align}   \label{tel-aviv}
\expectation\bigl[g_X^{\rho}(X)\bigr] - \expectation\bigl[\bigl(g_X(X)-1\bigr)^{\rho}\bigr]
\leq \exp\Bigl((\rho-1) \, H_{\frac1{\rho}}(X) \Bigr).
\end{align}
Hence, for $\rho \geq 1$, adding \eqref{haifa} and \eqref{tel-aviv} yields
\begin{align} \label{eq4: UB-NSI}
\expectation[g_{X}^{\rho}(X)]
\leq \frac1{1+\rho} \left[ \exp\Bigl(\rho H_{\frac1{1+\rho}}(X) \Bigr)
- 1 \right] + \exp \Bigl( (\rho-1) H_{\frac1{\rho}}(X) \Bigr).
\end{align}
This proves \eqref{eq1: UB-NSI} for $\rho \geq 1$.
We next consider the case where $\rho \in (0,1)$.
Lemma~\ref{lemma: 0<rho<1} in Appendix~\ref{appendix: UB-NSI} yields
\begin{align} \label{eq1: for 0<rho<1}
u^{\rho} \leq \frac{u^{1+\rho} - (u-1)^{1+\rho} + \rho}{1+\rho}, \quad u \geq 1,
\end{align}
and therefore, for $\rho \in (0,1)$,
\begin{align} \label{june17}
\expectation\bigl[g_X^{\rho}(X)\bigr] \leq
\frac1{1+\rho} \, \left( \expectation\left[g_X^{1+\rho}(X)\right]
- \expectation\left[\bigl(g_X(X)-1\bigr)^{1+\rho}\right] \right)
+ \frac{\rho}{1+\rho}.
\end{align}
Combining \eqref{eq: from Boztas' paper} and \eqref{june17} gives
\eqref{eq1: UB-NSI} for $\rho \in (0,1)$. The case $\rho=0$, for
which \eqref{eq1: UB-NSI} holds with equality, is trivial.
\end{proof}

\begin{remark}
Bound \eqref{eq1: UB-NSI} strictly (unless $X$ is deterministic)
improves the bound in \eqref{eq: Arikan96 - prop. 4 without conditioning}.
This holds for $\rho \in (0,1)$ since
$H_{\frac1{1+\rho}}(X) \geq 0$; for $\rho \in [1, \infty)$,
the difference between the bounds in
\eqref{eq: Arikan96 - prop. 4 without conditioning}
and \eqref{eq1: UB-NSI} is equal to
\begin{align}
& \exp\Bigl(\rho \, H_{\frac1{1+\rho}}(X) \Bigr) - \left\{ \frac1{1+\rho}
\left[ \exp\Bigl(\rho H_{\frac1{1+\rho}}(X) \Bigr) - 1 \right]
+ \exp \Bigl( (\rho-1) H_{\frac1{\rho}}(X) \Bigr) \right\} \nonumber \\
& = \frac{\rho}{1+\rho} \left( \, \sum_{x \in \set{X}} P_X^{\frac1{1+\rho}}(x)
\right)^{1+\rho} - \left( \, \sum_{x \in \set{X}} P_X^{\frac1{\rho}}(x)
\right)^{\rho} + \frac1{1+\rho} \nonumber \\
& \geq \frac{\rho}{1+\rho} \left( \, \sum_{x \in \set{X}} P_X^{\frac1{1+\rho}}(x)
\right)^{1+\rho} - \left( \, \sum_{x \in \set{X}} P_X^{\frac1{1+\rho}}(x)
\right)^{\rho} + \frac1{1+\rho} \label{nahalal} \\
& = s\left(\sum_{x \in \set{X}} P_X^{\frac1{1+\rho}}(x)\right)
\label{ramat-gan}
\end{align}
where \eqref{nahalal} holds since
$P_X^{\frac1{\rho}}(x) \leq P_X^{\frac1{1+\rho}}(x)$,
and \eqref{ramat-gan} holds with $s \colon [1, \infty) \to \Reals$ given by
\begin{align}
s(u) = u^\rho \left(\frac{\rho u}{1+\rho} - 1 \right) + \frac1{1+\rho}, \quad u \geq 1.
\end{align}
In \eqref{ramat-gan}, note that
$\underset{x \in \set{X}}{\sum} P_X^{\frac1{1+\rho}}(x) \geq \underset{x \in \set{X}}{\sum} P_X(x) = 1$.
The right side of \eqref{ramat-gan} is non-negative since $s(1)=0$, and
\begin{align}
s'(u)= \rho (u-1) u^{\rho-1} > 0
\end{align}
for all $u > 1$.
This shows that the bound in \eqref{eq1: UB-NSI} is at least as tight as the bound
in \eqref{eq: Arikan96 - prop. 4 without conditioning} with a strict improvement unless
$X$ is deterministic.
\end{remark}

In the range $\rho \in [0,2]$, we can tighten \eqref{eq1: UB-NSI} according to
the following result.
\begin{theorem} \label{thm: 0-2}
Under the assumptions in Theorem~\ref{thm: Arikan's UB},
\begin{enumerate}[a)]
\label{label1.5: UB-NSI}
\item For $\rho \in [0,1]$
\begin{align}  \label{eq1.5: UB-NSI}
\expectation[g_{X}^{\rho}(X)]
\leq \frac1{1+\rho} \, \exp\Bigl(\rho H_{\frac1{1+\rho}}(X) \Bigr)
+ \frac{\rho - (1-\rho)(2^{\rho}-1)(1-p_{\max})}{1+\rho}.
\end{align}
\item \label{label2: UB-NSI}
For $\rho \in [1,2]$
\begin{align} \label{eq2: UB-NSI}
\expectation[g_{X}^{\rho}(X)] \leq \frac1{1+\rho} \,
\exp\Bigl(\rho H_{\frac1{1+\rho}}(X) \Bigr)
+ \frac1\rho \,\exp\left((\rho-1) H_{\frac1\rho}(X) \right)
+ \frac{\rho^2-\rho-1}{\rho (1+\rho)}.
\end{align}
\end{enumerate}
Furthermore, both \eqref{eq1.5: UB-NSI} and \eqref{eq2: UB-NSI} hold with equality
if $X$ is deterministic.
\end{theorem}

\begin{proof}
See Appendix~\ref{appendix: UB-NSI}.
\end{proof}

\begin{remark}
Particularizing \eqref{eq2: UB-NSI} to $\rho=1$ and $\rho=2$, we
recover the bounds on the first and second moments in
\cite[Theorem~3]{Boztas97}. Furthermore, the bounds in \eqref{eq1.5: UB-NSI}
and \eqref{eq2: UB-NSI} provide a continuous transition at $\rho=1$.
\end{remark}

\begin{theorem} \label{thm: rho>=2}
Under the assumptions in Theorem~\ref{thm: Arikan's UB}, for $\rho \geq 2$,
\begin{align} \label{eq: rho>=2}
\expectation[g_{X}^{\rho}(X)] \leq 1 + \sum_{j=0}^{\lfloor \rho \rfloor}
c_j(\rho) \left[ \exp \left( (\rho-j) \, H_{\frac1{1+\rho-j}}(X)\right) - 1 \right],
\end{align}
where $\{c_j(\rho)\}$ is given by
\begin{align} \label{eq: c_j}
c_j(\rho) =
\begin{dcases}
\hfil \frac1{1+\rho}  & \quad j=0 \\[0.1cm]
\hfil \tfrac12        & \quad j=1 \\[0.1cm]
\frac{\rho \ldots (\rho-j+2)}{2^j}  & \quad
j \in \{2, \ldots, \lfloor \rho \rfloor - 1\} \\[0.1cm]
\frac{\rho \ldots (\rho-j+2)}{2^{j- 1} \; (\rho-j+1)}
& \quad j = \lfloor \rho \rfloor
\end{dcases}
\end{align}
and $\lfloor x \rfloor$ denotes the largest integer that is smaller than or equal to $x$.
\end{theorem}
\begin{proof}
See Appendix~\ref{appendix: thm for rho>=2}.
\end{proof}

\begin{remark}
In contrast to \cite[Theorem~3]{Boztas97}, the results in
Theorems~\ref{thm: 0-2} and~\ref{thm: rho>=2} provide an
explicit upper bound on $\expectation[g_{X}^\rho(X)]$ for
$\rho \in (0, \infty)$ as a function of R\'{e}nyi entropies of $X$.
Note also that the upper bounds in \eqref{eq2: UB-NSI} and
\eqref{eq: rho>=2} coincide at $\rho=2$.
\end{remark}

\begin{remark} \label{remark: numerical}
Numerical evidence shows that none of the bounds in
\eqref{eq1: UB-NSI} and \eqref{eq: rho>=2} supersedes the other
for $\rho > 2$ (as it is next illustrated in
Examples~\ref{example: 3rd moment} and~\ref{example: 20th moment}).
Since \eqref{eq2: UB-NSI} and \eqref{eq: rho>=2} coincide at $\rho=2$,
Theorem~\ref{thm: 0-2}-\ref{label2: UB-NSI}) implies that the bound
in \eqref{eq: rho>=2} is tighter than \eqref{eq1: UB-NSI} for this
value of $\rho$.
\end{remark}

\begin{example}
\label{example: 3rd moment}
Let $X \in \{1, \ldots, 32\}$ have
the probability distribution in \eqref{eq: geometric dist.} with
$a = 0.9$ and $M = 32$.
Table~\ref{Table: third moment} compares
$\tfrac1{3} \log_{\mathrm{e}} \expectation[g_{X}^{3}(X)]$
to its various lower and upper bounds.
Notice that in this example, the upper bound in \eqref{eq: rho>=2}
improves the bound in \eqref{eq1: UB-NSI}.
\begin{table}[h!]
\caption{\label{Table: third moment} Comparison of
$\tfrac13 \log_{\mathrm{e}} \expectation[g_{X}^3(X)]$ and
bounds in Example~\ref{example: 3rd moment}.}
\renewcommand{\arraystretch}{1.5}
\centering
\begin{tabular}{|c|c||c||c|c|c|} \hline
\eqref{eq: Arikan's bound} & Theorem~\ref{theorem: improving Arikan's bound}
&  $\tfrac13 \log_{\mathrm{e}} \expectation[g_{X}^3(X)]$ & \eqref{eq: rho>=2}
& \eqref{eq1: UB-NSI} & \eqref{eq: Arikan96 - prop. 4 without conditioning} \\
lower bound & lower bound & exact value & upper bound & upper bound & upper bound
\\[0.1cm] \hline
$1.864$ & $2.593$ & $2.609$ & $2.920$ & $2.939$ & $3.360$ \\
\hline
\end{tabular}
\end{table}
\end{example}

\begin{example}
\label{example: 20th moment}
Let $X \in \{1, \ldots, 16\}$ have
the probability distribution in \eqref{eq: geometric dist.} with
$a = 0.9$ and $M = 16$. Table~\ref{Table: 20th moment}
compares $\tfrac1{20} \log_{\mathrm{e}} \expectation[g_{X}^{20}(X)]$
to its various lower and upper bounds.
\begin{table}[h!]
\caption{\label{Table: 20th moment} Comparison of
$\tfrac1{20} \log_{\mathrm{e}} \expectation[g_{X}^{20}(X)]$
and bounds in Example~\ref{example: 20th moment}.}
\renewcommand{\arraystretch}{1.5}
\centering
\begin{tabular}{|c|c||c||c|c|c|} \hline
\eqref{eq: Arikan's bound} & Theorem~\ref{theorem: improving Arikan's bound}
& $\tfrac1{20} \log_{\mathrm{e}} \expectation[g_{X}^{20}(X)]$ & \eqref{eq: rho>=2}
& \eqref{eq1: UB-NSI} & \eqref{eq: Arikan96 - prop. 4 without conditioning} \\
lower bound & lower bound & exact value & upper bound & upper bound & upper bound
\\[0.1cm] \hline
$1.439$ & $2.602$ & $2.606$ & $2.662$ & $2.657$ & $2.767$ \\
\hline
\end{tabular}
\end{table}
Note that in Table~\ref{Table: 20th moment}, the lower bound in \eqref{eq: Arikan's bound} is
quite weaker than the lower bound in Theorem~\ref{theorem: improving Arikan's bound},
which shows an excellent match with its exact value. In contrast to Example~\ref{example: 3rd moment},
Example~\ref{example: 20th moment} shows that the upper bound \eqref{eq: rho>=2} may be
weaker than \eqref{eq1: UB-NSI}.
\end{example}

\subsection{Improved bounds on guessing moments with side information}
\label{subsection: guessing - SI}

This subsection extends the lower and upper bounds in Sections~\ref{subsec: LB guessing}
and~\ref{subsec: UB guessing} to allow side information $Y$
for guessing the value of $X$. These bounds tighten the results in \cite[Theorem~1]{Arikan96}
and \cite[Proposition~4]{Arikan96} for all $\rho > 0$.

\begin{theorem} \label{theorem: LB guessing - SI}
Let $X$ and $Y$ be discrete random variables taking values on the sets
$\set{X} = \{1, \ldots, M\}$ and $\set{Y}$, respectively.
For all $y \in \set{Y}$, let $g(\cdot | y)$ be a guessing function
of $X$ given that $Y=y$. Then, for $\rho \in (0, \infty)$,
\begin{align} \label{eq: improving Arikan's bound - SI}
\frac1{\rho} \; \log \expectation \bigl[g^{\rho}(X|Y)\bigr]
\geq \sup_{\beta \in (-\rho, 0) \cup (0, \infty)} \frac1{\beta} \,
\Bigl[ H_{\frac{\beta}{\beta+\rho}}(X|Y) - \log u_M(\beta) \Bigr]
\end{align}
with $u_M(\cdot)$ as defined in \eqref{eq: u}.
\end{theorem}

\begin{proof}
\begin{align}
\label{shalom}
\expectation[g^{\rho}(X|Y)]
& = \sum_{y \in \set{Y}} P_Y(y) \, \expectation[g^{\rho}(X|y)] \\
\label{hi}
& \geq \sum_{y \in \set{Y}} P_Y(y) \, \exp \left(
\sup_{\beta \in (-\rho, 0) \cup (0, \infty)}
\frac{\rho}{\beta} \left[ H_{\frac{\beta}{\beta+\rho}}(X|Y=y)
- \log u_M(\beta) \right] \right) \\
\label{hello}
& \geq \sup_{\beta \in (-\rho, 0) \cup (0, \infty)} \;
\sum_{y \in \set{Y}} P_Y(y) \, \exp \left( \frac{\rho}{\beta}
\left[ H_{\frac{\beta}{\beta+\rho}}(X|Y=y) - \log u_M(\beta) \right] \right) \\
\label{hola}
&= \sup_{\beta \in (-\rho, 0) \cup (0, \infty)} \; \exp \left( \frac{\rho}{\beta}
\left[ H_{\frac{\beta}{\beta+\rho}}(X|Y) - \log u_M(\beta) \right] \right)
\end{align}
where \eqref{hi} follows from \eqref{eq: improving Arikan's bound} with the
conditional guessing function $g(\cdot|Y=y) \colon \set{X} \to \set{X}$ where
$|\set{X}|=M$; \eqref{hola} follows from \eqref{eq2: Arimoto - cond. RE} with
$\alpha = \frac{\beta}{\rho+\beta}$.
\end{proof}

\begin{remark} \label{remark: Arikan2}
Similarly to Remark~\ref{remark: Arikan1}, the loosening of the result in
Theorem~\ref{theorem: LB guessing - SI} by replacing
the supremum over $\beta \in (-\rho, 0) \cup (0, \infty)$ in the right side of
\eqref{eq: improving Arikan's bound - SI} with the value $\beta=1$
and further using the inequality $u_M(1) \leq 1 + \log_{\mathrm{e}} M$ for
$M \geq 2$ (see \eqref{eq: u}) yields Arikan's result in \cite[(2)]{Arikan96}.
As explained in \cite{Arikan96}, Theorem~\ref{theorem: LB guessing - SI} can
be used to obtain an improved non-asymptotic lower bound on the moments of the
number of computational steps used to decode tree codes by an arbitrary
sequential decoder.
\end{remark}

\begin{theorem} \label{theorem: UB-SI}
Let $X$ and $Y$ be discrete random variables taking values on sets
$\set{X}$ and $\set{Y}$, respectively.
For all $y \in \set{Y}$, let $g_{X|Y}(\cdot | y)$ be a ranking function
of $X$ given that $Y=y$. Then, for $\rho \in (0, \infty)$,
\begin{align}
\label{eq1: UB-SI}
\expectation[g_{X|Y}^{\rho}(X|Y)]
\leq \frac1{1+\rho} \left[ \exp\Bigl(\rho H_{\frac1{1+\rho}}(X|Y) \Bigr) - 1 \right] +
\exp \Bigl( (\rho-1)^+ \, H_{\frac1{\rho}}(X|Y) \Bigr).
\end{align}
For $\rho \in (0,1)$, \eqref{eq1: UB-SI} can be tightened to
\begin{align} \label{eq: 20171012a}
\expectation[g_{X|Y}^{\rho}(X|Y)] \leq \frac1{1+\rho} \,
\exp\Bigl(\rho H_{\frac1{1+\rho}}(X|Y) \Bigr) +
\frac{\rho - (1-\rho) (2^\rho-1) (1-p^\ast)}{1+\rho}
\end{align}
with
\begin{align} \label{p^ast}
p^\ast \triangleq \underset{y \in \set{Y}}{\sup} \, \underset{x \in \set{X}}{\max} \, p_{X|Y}(x|y).
\end{align}
For $\rho \in [1,2]$, \eqref{eq1: UB-SI} can be tightened to
\begin{align} \label{eq2: UB-SI}
\expectation[g_{X|Y}^{\rho}(X|Y)] \leq \frac1{1+\rho} \,
\exp\Bigl(\rho H_{\frac1{1+\rho}}(X|Y) \Bigr)
+ \frac1\rho \,\exp\left((\rho-1) H_{\frac1\rho}(X|Y) \right)
+ \frac{\rho^2-\rho-1}{\rho (1+\rho)}.
\end{align}
Moreover, for $\rho \geq 2$,
\begin{align} \label{eq: rho>=2 - SI}
\expectation[g_{X|Y}^{\rho}(X|Y)] \leq 1 + \sum_{j=0}^{\lfloor \rho \rfloor}
c_j(\rho) \left[ \exp \left( (\rho-j) \, H_{\frac1{1+\rho-j}}(X|Y)\right) - 1 \right],
\end{align}
with $\{c_j(\rho)\}$ defined in \eqref{eq: c_j}.
\end{theorem}

\begin{proof}
From Theorem~\ref{theorem: UB-NSI} and \eqref{shalom}, for all $\rho \in [1, \infty)$,
\begin{align}
& \expectation[g_{X|Y}^{\rho}(X|Y)] \nonumber \\
& \leq \frac1{1+\rho} \left[ \sum_{y \in \set{Y}} P_Y(y) \,
\exp \Bigl( \rho H_{\frac1{1+\rho}}(X \, | \, Y=y) \Bigr) - 1 \right]
+ \sum_{y \in \set{Y}} P_Y(y) \,
\exp \Bigl( (\rho-1) \, H_{\frac1\rho}(X \, | \, Y=y) \Bigr) \label{eilat} \\
\label{arad}
& = \frac1{1+\rho} \left[ \exp \Bigl( \rho H_{\frac1{1+\rho}}(X|Y) \Bigr) - 1 \right]
+ \exp \Bigl( (\rho-1) \, H_{\frac1\rho}(X|Y) \Bigr)
\end{align}
where \eqref{eilat} follows from \eqref{eq1: UB-NSI}, and \eqref{arad} follows from
\eqref{eq2: Arimoto - cond. RE} with $\alpha = \frac1{1+\rho}$ and $\alpha = \frac1\rho$.
The proof of \eqref{eq1: UB-SI} for $\rho \in (0,1)$ is similar.

Extending \eqref{eq1.5: UB-NSI}, \eqref{eq2: UB-NSI} and \eqref{eq: rho>=2} to
\eqref{eq: 20171012a}, \eqref{eq2: UB-SI} and \eqref{eq: rho>=2 - SI}, respectively,
relies on \eqref{eq2: Arimoto - cond. RE} and it is similar to the extension of the
result in Theorem~\ref{theorem: UB-NSI} to \eqref{eq1: UB-SI}.
\end{proof}

\begin{remark}
The bound in \eqref{eq1: UB-SI} improves the result in \cite[Proposition~4]{Arikan96}
since the superiority of the bound in Theorem~\ref{theorem: UB-NSI} over the result
in Theorem~\ref{thm: Arikan's UB} is preserved after the averaging in \eqref{shalom}.
\end{remark}

\begin{remark}
Following Remark~\ref{remark: numerical}, neither of the bounds in
\eqref{eq1: UB-SI} and \eqref{eq: rho>=2 - SI} supersedes the other
for $\rho \geq 2$. Since \eqref{eq2: UB-SI} and \eqref{eq: rho>=2 - SI}
coincide for $\rho=2$, it follows from Theorem~\ref{theorem: UB-SI} that
the upper bound in \eqref{eq: rho>=2 - SI} is tighter than \eqref{eq1: UB-SI}
for this value of $\rho$. Note that both bounds are tight if $X$, conditioned
on $Y$, is deterministic, for which they are equal to
$\expectation[g_{X|Y}^{\rho}(X|Y)]=1$. Finally, note that the transition
from \eqref{eq: 20171012a} to \eqref{eq2: UB-SI} is continuous at $\rho=1$.
\end{remark}

\subsection{Relationship between guessing moments and minimum probability of error}
\label{subsection: guessing moments versus epsilon}
Let $X$ and $Y$ be discrete random variables,\footnote{The assumption that $Y$ is
a discrete random variable can be easily dispensed with.} taking values on the sets
$\set{X}$ and $\set{Y}$ respectively.
The minimum probability of error of $X$ given $Y$, denoted by
$\varepsilon_{X|Y}$, is achieved by the maximum-a-posteriori
(MAP) decision rule. Hence,
\begin{align} \label{eq: epsilon_X|Y}
\varepsilon_{X|Y} = \sum_{y \in \set{Y}} P_Y(y)
\left[ 1 - \max_{x \in \set{X}} P_{X|Y}(x|y) \right].
\end{align}
In contrast, the moments of the ranking function $\expectation[g_{X|Y}^{\rho}(X|Y)]$
quantify the number of guesses required for correctly identifying the unknown object $X$
on the basis of $Y$. It is therefore natural to establish relationships
between both quantities. First note that by definition,
\begin{align} \label{eq: 20171016-b}
1 - \varepsilon_{X|Y} = \prob[g_{X|Y}(X|Y)=1].
\end{align}

In \cite{SV18}, we derived lower and upper bounds
on $\varepsilon_{X|Y}$ as a function of $H_\alpha (X|Y)$
for an arbitrary order $\alpha$. In this section,
Theorems~\ref{theorem: LB guessing - SI}--\ref{theorem: UB-SI} provide
lower and upper bounds on guessing moments of a ranking function
$g_{X|Y}(X|Y)$ as a function of Arimoto-R\'enyi conditional
entropies. As a natural continuation to these studies, we derive
tight lower and upper bounds on $\expectation[g_{X|Y}^{\rho}(X|Y)]$ as
a function of $\varepsilon_{X|Y}$.

\begin{theorem} \label{theorem: g_X|Y - epsilon_X|Y}
Let $X$ and $Y$ be discrete random variables taking values on sets
$\set{X} = \{1, \ldots, M\}$ and $\set{Y}$, respectively.
Then, for $\rho > 0$,
\begin{align} \label{NY}
f_{\rho}(\varepsilon_{X|Y}) & \leq \expectation[g_{X|Y}^{\rho}(X|Y)] \\
\label{NY2}
& \leq 1 + \left( \tfrac1{M-1} \sum_{j=2}^M j^\rho - 1 \right) \varepsilon_{X|Y}
\end{align}
where the function $f_{\rho} \colon [0,1) \to [0, \infty)$ is given by
\begin{align}
\label{eq: f-thm10}
& f_{\rho}(u) = (1-u) \sum_{j=1}^{k_u} j^\rho + [1 - (1-u)k_u] (k_u+1)^{\rho}, \quad u \in [0,1)\\
\label{eq: k-thm10}
& k_u = \left\lfloor \frac1{1-u} \right\rfloor.
\end{align}
Furthermore, the lower and upper bounds in \eqref{NY} and \eqref{NY2} are tight:
\begin{itemize}
\item Let $p_{\max}(y) = \underset{{x \in \set{X}}}{\max} P_{X|Y}(x|y)$ for $y \in \set{Y}$.
The lower bound is attained if and only if $p_{\max}(y) = p_{\max}$
is fixed for all $y \in \set{Y}$, and conditioned on $Y=y$, $X$ has
$\left\lfloor\frac1{p_{\max}}\right\rfloor$ masses equal to $p_{\max}$,
and an additional mass equal to $1-p_{\max} \left\lfloor\frac1{p_{\max}}\right\rfloor$
whenever $\frac1{p_{\max}}$ is not an integer.
\item The upper bound is attained if and only if regardless of $y \in \set{Y}$,
conditioned on $Y=y$, $X$ is equiprobable among its $M-1$ conditionally
least likely values on $\set{X}$.
\end{itemize}
\end{theorem}

\begin{proof}
For $y \in \set{Y}$ and $i \in \{1, \ldots, M\}$, let $x_i(y) \in \set{X}$ satisfy
$g_{X|Y}\bigl(x_i(y)|y\bigr)=i$.
By the definition of $g_{X|Y}(\cdot | \cdot)$ on $\set{X} \times \set{Y}$,
it follows that
\begin{align} \label{170824-2}
P_{X|Y}\bigl(x_1(y)|y\bigr) \geq P_{X|Y}\bigl(x_2(y)|y\bigr)
\geq \ldots \geq P_{X|Y}\bigl(x_M(y)|y\bigr).
\end{align}
For $\rho \in (0, \infty)$,
\begin{align}
\label{yiannis}
\expectation[ g_{X|Y}^{\rho}(X|Y)]
&= \sum_{y \in \set{Y}} P_Y(y) \, \expectation[ g_{X|Y}^{\rho}(X|y)] \\
\label{yossi}
&= \sum_{y \in \set{Y}} \left\{ P_Y(y) \left( \max_{x \in \set{X}}
P_{X|Y}(x|y) + \sum_{i=2}^M i^\rho P_{X|Y}\bigl(x_i(y) | y \bigr) \right) \right\} \\
\label{yuval}
&= 1 - \varepsilon_{X|Y} + \sum_{y \in \set{Y}} \left\{ P_Y(y)
\sum_{i=2}^M i^\rho P_{X|Y}\bigl(x_i(y) | y \bigr) \right\} \\
\label{ido}
&\leq 1 - \varepsilon_{X|Y} + \frac1{M-1} \sum_{y \in \set{Y}}
\left\{ P_Y(y) \left(\sum_{i=2}^M i^\rho\right) \left( \sum_{i=2}^M P_{X|Y}\bigl(x_i(y) | y \bigr) \right) \right\} \\
\label{emre}
&= 1 - \varepsilon_{X|Y} + \frac1{M-1} \sum_{y \in \set{Y}}
\left\{ P_Y(y) \left( 1 - \max_{x \in \set{X}} P_{X|Y}(x|y) \right) \right\} \sum_{i=2}^M i^\rho \\
\label{rudi}
&= 1 + \left( \frac1{M-1} \sum_{i=2}^M i^\rho - 1 \right) \varepsilon_{X|Y}
\end{align}
where \eqref{yossi} and \eqref{emre} follow from \eqref{170824-2};
\eqref{yuval} and \eqref{rudi} follow from \eqref{eq: epsilon_X|Y};
\eqref{ido} follows from Lemma~\ref{lemma-2: f} (see Appendix~\ref{app: d}),
for which we take the strictly monotonically increasing function $f$
to be given by
\begin{align} \label{eq: 20171016-c}
f(i) = (i+1)^\rho, \quad i = 1, \ldots, M-1,
\end{align}
and $q_i \leftarrow P_{X|Y}\bigl(x_{i+1}(y) | y \bigr)$.
This proves the upper bound in \eqref{NY2}.
A necessary and sufficient condition for attaining the upper bound in \eqref{NY2}
follows from Lemma~\ref{lemma-2: f} and due to the strict
monotonicity of the function in \eqref{eq: 20171016-c} for $\rho > 0$. Hence, it follows
that \eqref{ido} is satisfied with equality if and only if, for every $y \in \set{Y}$,
$q_i(y) \triangleq P_{X|Y}\bigl(x_{i+1}(y) | y \bigr)$ is fixed for all $i \in \{1, \ldots, M-1\}$;
due to \eqref{170824-2}, this is equivalent to the requirement that regardless of $y \in \set{Y}$,
conditioned on $Y=y$, the $M-1$ least probable values of $X$ are equiprobable.

To show \eqref{NY}, we write the conditional moment of the ranking function as
\begin{align} \label{dor}
\expectation[ g_{X|Y}^{\rho}(X|Y)]
&= \sum_{y \in \set{Y}} \left\{ P_Y(y)
\sum_{i=1}^M i^\rho P_{X|Y}\bigl(x_i(y) | y \bigr) \right\},
\end{align}
and we denote the conditional error probability given the observation $Y=y$ by
$\varepsilon_{X|Y}(y)= 1-  p_{\max}(y)$. Note from \eqref{170824-2} that $x_1(y)$
is a mode of $P_{X|Y}(\cdot|y)$, and we have
\begin{align} \label{thm 10 - a}
P_{X|Y}\bigl(x_1(y)|y\bigr) = p_{\max}(y) = 1 - \varepsilon_{X|Y}(y).
\end{align}
The inner sum in the right side
of \eqref{dor} is minimized, for
a given value of $\varepsilon_{X|Y}(y)$, by
\begin{align}  \label{thm 10 - b}
P_{X|Y}(x_i(y)|y) =
\begin{dcases}
1-\varepsilon_{X|Y}(y) & \quad i = 1, \ldots, \left\lfloor \frac1{1-\varepsilon_{X|Y}(y)} \right\rfloor \\[0.1cm]
1-(1-\varepsilon_{X|Y}(y)) \left\lfloor \frac1{1-\varepsilon_{X|Y}(y)} \right\rfloor & \quad i =
\left\lfloor \frac1{1-\varepsilon_{X|Y}(y)} \right\rfloor +1 \\
0 & \quad \text{otherwise.}
\end{dcases}
\end{align}
In order to show it, note that according to \eqref{170824-2}, any perturbation of $P_{X|Y}(\cdot|y)$ in \eqref{thm 10 - b}
necessarily shifts mass from $x_i(y)$ to $x_j(y)$ with $j>i$; since $\{k^\rho\}_{k \geq 1}$ is positive and monotonically
increasing in $k$ for $\rho > 0$, this can only increase the inner sum in the right side of \eqref{dor}.

From the minimizing conditional distribution in \eqref{thm 10 - b}, for a given value of $\varepsilon_{X|Y}(y)$, and
\eqref{eq: f-thm10}--\eqref{eq: k-thm10}
\begin{align}  \label{thm 10 - c}
\sum_{i=1}^M i^\rho P_{X|Y}\bigl(x_i(y) | y \bigr) \geq f_{\rho}\bigl( \varepsilon_{X|Y}(y) \bigr).
\end{align}
Due to the convexity of $f_{\rho} \colon [0,1) \to [0, \infty)$ in \eqref{eq: f-thm10}--\eqref{eq: k-thm10} for $\rho>0$
(see Lemma~\ref{lemma: f is convex} in Appendix~\ref{app: d}), and since
\begin{align}  \label{thm 10 - d}
\varepsilon_{X|Y} = \sum_{y \in \set{Y}} P_Y(y) \varepsilon_{X|Y}(y),
\end{align}
the lower bound in \eqref{NY} follows from \eqref{dor}, \eqref{thm 10 - c} and Jensen's inequality.
This also yields the necessary and sufficient condition for the attainability of the bound in \eqref{NY},
as it specified in the statement of the theorem (note that, from \eqref{thm 10 - a}, $p_{\max}(y)$
is fixed for all $y \in \set{Y}$ if and only if $\varepsilon_{X|Y}(y)$ is so).
\end{proof}

\begin{remark} \label{remark: extreme points}
The lower and upper bounds in \eqref{NY} and \eqref{NY2} coincide in each of the extreme cases
$\varepsilon_{X|Y} = 0$ and $\varepsilon_{X|Y} = 1-\frac1M$. The former and latter values refer,
respectively, to the following cases:
\begin{itemize}
\item $X$ is a deterministic function of $Y$,
\item $X$ and $Y$ are independent, and $X$ is equiprobable.
\end{itemize}
\end{remark}

\begin{example} \label{example: 4*4}
Let $X$ and $Y$ be random variables which take values on $\set{X} = \{1, 2, 3, 4\}$,
and let
\begin{align} \label{matrix 4X4}
\bigl[P_{XY}(x,y)\bigr]_{(x,y) \in \set{X}^2}
= \frac{1}{52}\left( \begin{array}{rrrr}
10  &   1  &   1  &   1 \\
 1  &  10  &   1  &   1 \\
 1  &   1  &  10  &   1 \\
 1  &   1  &   1  &  10
\end{array}
\right).
\end{align}
It follows from \eqref{matrix 4X4} that $P_Y(y) = \tfrac14$ for all $y \in \set{X}$, and
we can select the conditional ranking function of $X$ given $Y$ to satisfy $g_{X|Y}(x|1) = x$
for all $x \in \set{X}$ (since 1 is the most likely
value of $X$ given $Y=1$, and 2,3,4 are conditionally equiprobable
given $Y=1$); moreover, symmetry in \eqref{matrix 4X4} yields
\begin{align}
\label{eq: moment rho - SI}
\expectation[g_{X|Y}^{\rho}(X|Y)]
&= \tfrac1{13} \left(10+2^\rho + 3^\rho + 4^\rho\right).
\end{align}
The result in \eqref{eq: moment rho - SI} coincides with the upper bound in \eqref{NY2}
since, regardless of $y$, $P_{X|Y}(x| y) = \tfrac1{13}$ for the $|\set{X}|-1=3$ least
probable values of $X$ given $Y=y$. This can be verified directly since it follows from
\eqref{eq: epsilon_X|Y} and \eqref{matrix 4X4} that $\varepsilon_{X|Y} = \frac{3}{13}$,
and then the upper bound in \eqref{NY2} (with $M=4$) is equal to the right side
of \eqref{eq: moment rho - SI}.
\end{example}

\vspace*{0.2cm}
In the next example, we illustrate the locus of attainable values of
$(\varepsilon_{X|Y}, \log_{\mathrm{e}} \expectation[g_{X|Y}^{\rho}(X|Y)])$ for a fixed $M$.
The extreme cases identified in Remark~\ref{remark: extreme points} correspond to the
points $(0,0)$ and $\left(1-\frac1M, \log_{\mathrm{e}} \left(\frac1M \sum_{j=1}^M j^\rho\right) \right)$;
in these extreme cases, there is a one-to-one correspondence between the two quantities.

\begin{example} \label{example: rho=1}
Let $X$ be a random variable taking values on $\set{X}$ with $|\set{X}|=M \geq 2$.
Letting $\rho=1$ in \eqref{NY}--\eqref{eq: k-thm10} yields
\begin{align}  \label{eq: 20171014-c}
& 1+ \tfrac12 (1+\varepsilon_{X|Y}) \left\lfloor \frac1{1 - \varepsilon_{X|Y}} \right\rfloor
- \tfrac12 (1-\varepsilon_{X|Y}) \left\lfloor \frac1{1 - \varepsilon_{X|Y}} \right\rfloor^2
\leq \expectation[g_{X|Y}(X|Y)] \leq 1 + \tfrac12 M \varepsilon_{X|Y}
\end{align}
where both upper and lower bounds on the expected number of guesses in \eqref{eq: 20171014-c}
are attainable for any value of $\varepsilon_{X|Y} \in [0, 1-\frac1M]$.
The plots in Figure~\ref{fig: rho=1} illustrate the tight bounds in \eqref{eq: 20171014-c}
for a fixed $M$.
\begin{figure}[h]
\hspace*{0.5cm}
\includegraphics[width=8.0cm]{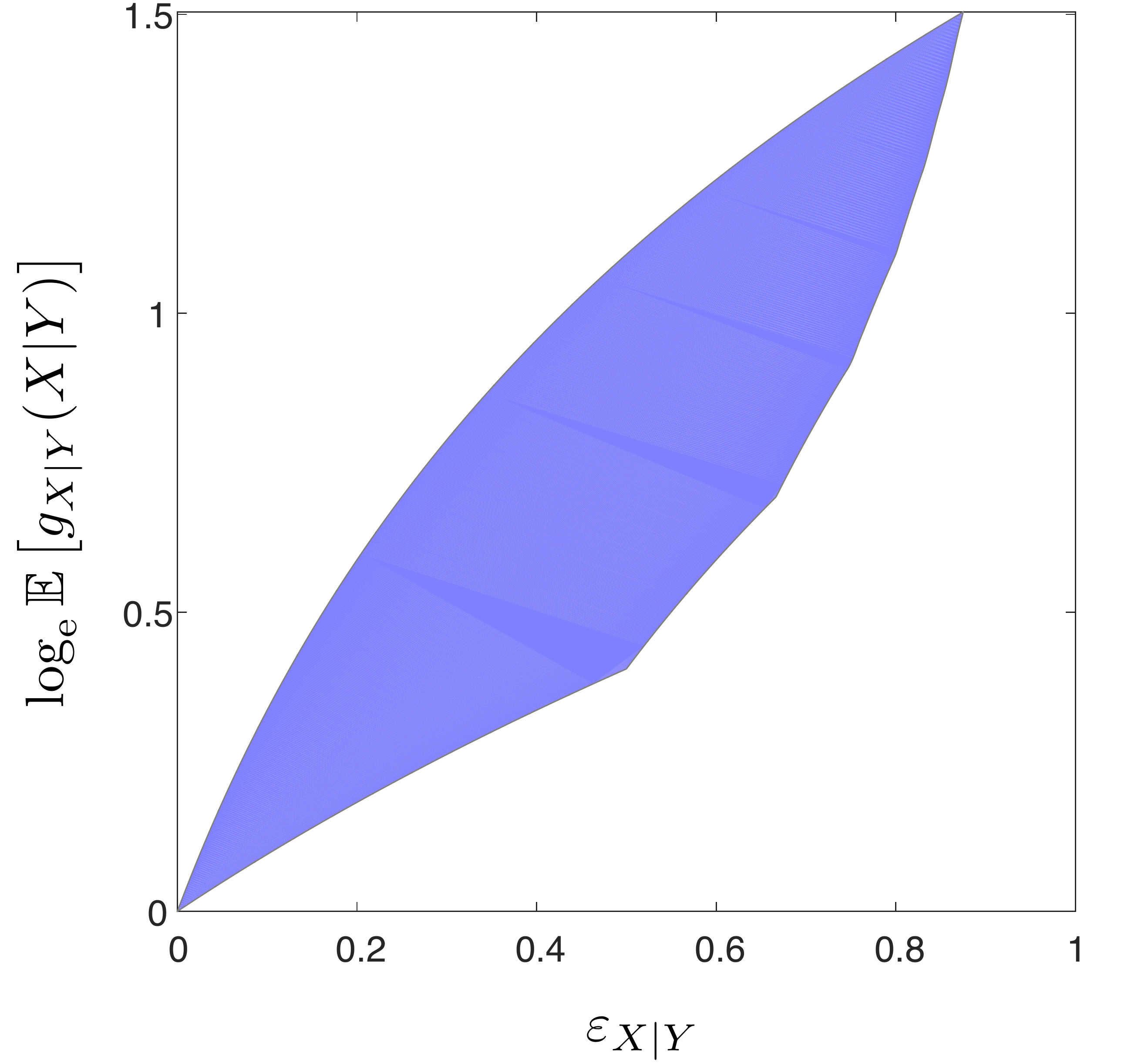}
\hspace*{0.5cm}
\includegraphics[width=8.2cm]{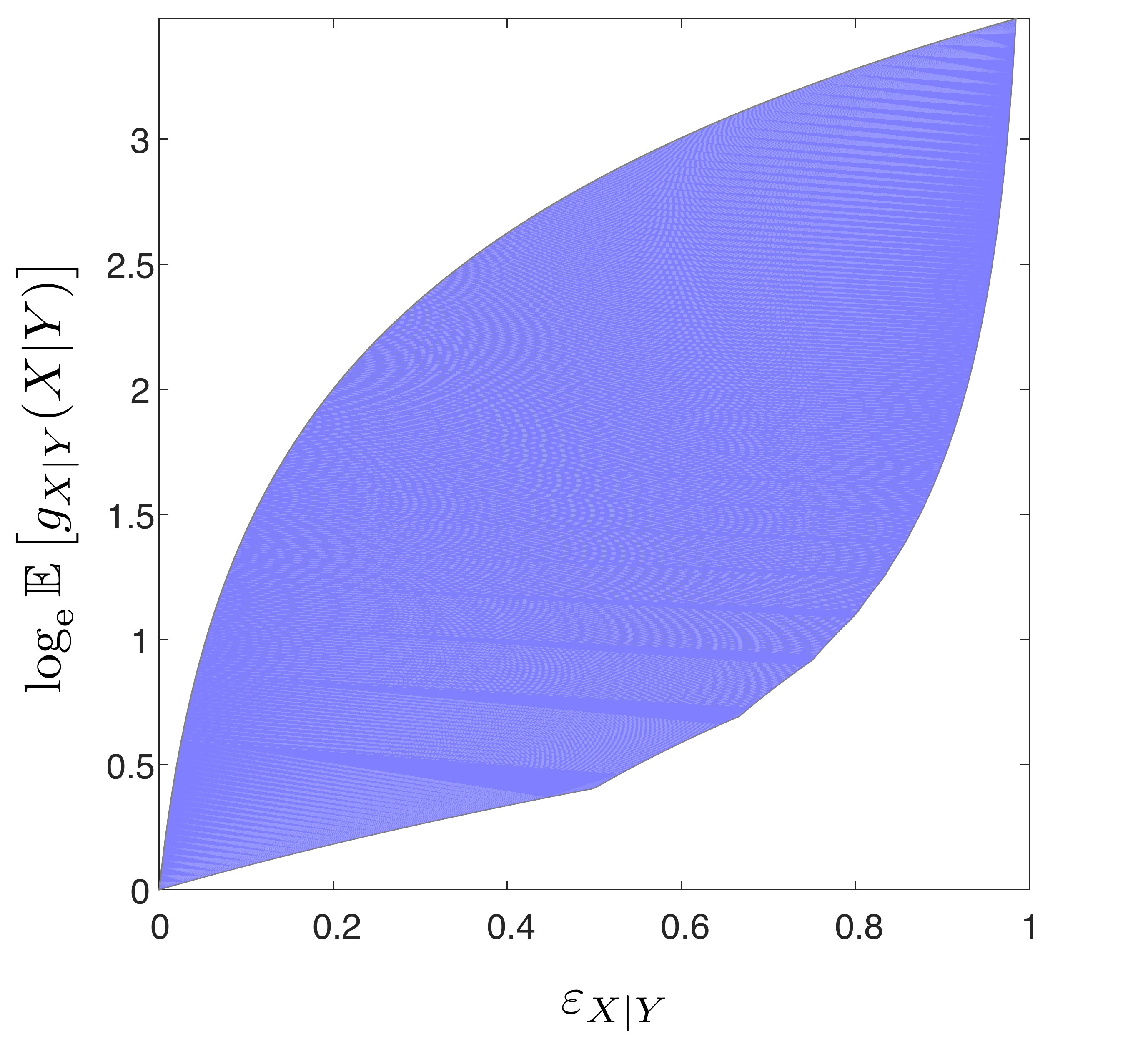}
\hspace*{-0.3cm}
\caption{\label{fig: rho=1} Example~\ref{example: rho=1}: locus of attainable values of
$(\varepsilon_{X|Y}, \log_{\mathrm{e}} \expectation[g_{X|Y}(X|Y)])$. The random variable
$X$ takes $M=8$ (left plot) or $M=64$ (right plot) possible values.}
\end{figure}
\end{example}

In view of Theorems~\ref{theorem: improving Arikan's bound} and~\ref{theorem: g_X|Y - epsilon_X|Y},
the next result provides an explicit lower bound on $\varepsilon_{X|Y}$ as a function of
$H_{\alpha}(X|Y)$ for any non-zero $\alpha < 1$.
\newpage
\begin{theorem} \label{thm: 20171020}
Let $X$ and $Y$ be discrete random variables taking values on sets
$\set{X} = \{1, \ldots, M\}$ and $\set{Y}$, respectively. Then,
for all $\alpha \in (-\infty, 0) \cup (0,1)$,
\begin{align}
\label{eq: 20171020-b}
\varepsilon_{X|Y} \geq \sup_{\rho > 0} \left \{ \frac{\exp\left(\bigl(\frac1{\alpha}-1\bigr) \,
\Bigl[ H_{\alpha}(X|Y) - \log u_M\Bigl(\frac{\alpha \rho}{1-\alpha}\Bigr) \Bigr]
\right)-1}{\tfrac1{M-1} \overset{M}{\underset{j=2}{\sum}} j^\rho - 1} \right\}
\end{align}
with $u_M(\cdot)$ as defined in \eqref{eq: u}.
\end{theorem}

\begin{proof}
Combining \eqref{eq: improving Arikan's bound - SI} and \eqref{NY2} yields, for every
$\rho > 0$ and $\beta \in (-\rho, 0) \cup (0, \infty)$,
\begin{align} \label{eq: 20171020-a}
\varepsilon_{X|Y} \geq \frac{\exp\left(\frac{\rho}{\beta} \,
\Bigl[ H_{\frac{\beta}{\beta+\rho}}(X|Y) - \log u_M(\beta) \Bigr]
\right)-1}{\tfrac1{M-1} \overset{M}{\underset{j=2}{\sum}} j^\rho - 1}.
\end{align}
Fix $\alpha \in (-\infty, 0) \cup (0,1)$,
and let the above free parameters $\beta$ and $\rho$ satisfy
$\rho = \bigl(\frac1{\alpha}-1\bigr) \beta$ (note that $\alpha$ cannot be zero or larger than 1,
as otherwise $\beta \notin (-\rho, 0) \cup (0, \infty)$). Supremizing the right side of
\eqref{eq: 20171020-a} over $\rho > 0$ yields \eqref{eq: 20171020-b}.
\end{proof}

In \cite{SV18}, we derived the following lower bounds on $\varepsilon_{X|Y}$ as a function of $H_{\alpha}(X|Y)$:
\begin{enumerate}[1)]
\item A generalization of Fano's inequality in \cite[Theorem~3]{SV18} holds for $\alpha > 0$, and it is given by
\begin{align} \label{eq1: generalized Fano}
H_{\alpha}(X | Y) \leq \log M -
d_{\alpha}\bigl( \varepsilon_{X|Y} \| 1-\tfrac1M \bigr)
\end{align}
with $d_{\alpha}(\cdot \| \cdot)$ as defined in \eqref{eq1: binary RD}.
\item
An explicit lower bound in \cite[Theorem~6]{SV18} holds for $\alpha < 0$, and it is given by
\begin{align} \label{eq: LB via HolderI}
\varepsilon_{X|Y} \geq \exp \left( \frac{1-\alpha}{\alpha} \,
\Bigl[ H_{\alpha}(X|Y) - \log(M-1) \Bigr] \right).
\end{align}
\end{enumerate}
Example~\ref{example: 20171022} includes numerical comparisons of the lower bounds on $\varepsilon_{X|Y}$ in
\eqref{eq: 20171020-b}, \eqref{eq1: generalized Fano} and \eqref{eq: LB via HolderI}.

\begin{remark} \label{remark: Shannon}
Shannon's inequality \cite{Shannon58}
(see also \cite{verduITA2011}) gives an explicit lower bound on
$\varepsilon_{X|Y}$ as a function of $H(X|Y)$ when $M$ is finite:
\begin{align} \label{eq: Shannon}
\varepsilon_{X|Y} \geq \frac16 \, \frac{H(X|Y)}{\log M + \log \log M - \log H(X|Y)},
\end{align}
and the right side of \eqref{eq: Shannon}
does not depend on the base of the logarithm. The bound in \eqref{eq: 20171020-b}
becomes trivial in the limit where $\alpha \uparrow 1$ since, for any fixed $\rho > 0$,
\eqref{eq: u} implies that $u_M\Bigl(\frac{\alpha \rho}{1-\alpha}\Bigr) \to 1$,
and therefore the lower bound on $\varepsilon_{X|Y}$ tends to zero in this case.
Nevertheless, numerical experimentation shows that the convergence of this bound in
\eqref{eq: 20171020-b} to zero is only affected by values of $\alpha$ very close to~1,
as it is illustrated in Example~\ref{example: 20171022} with a comparison to Shannon's
lower bound in \eqref{eq: Shannon}.
\end{remark}

\begin{example} \label{example: 20171022}
Let $X$ and $Y$ be random variables taking values on $\set{X} = \{1, 2, 3, 4\}$,
and let
\begin{align} \label{20171022-l}
\bigl[P_{XY}(x,y)\bigr]_{(x,y) \in \set{X}^2}
= \frac{1}{100}\left( \begin{array}{rrrr}
 9  &   3  &   4  &   9 \\
 9  &   9  &   3  &   4 \\
 4  &   9  &   9  &   3 \\
 3  &   4  &   9  &   9
\end{array}
\right).
\end{align}
It can be verified from \eqref{eq: epsilon_X|Y} that $\varepsilon_{X|Y} = \tfrac{16}{25} = 0.640$.
\begin{table}[h!]
\caption{\label{Table: 20171024} Example~\ref{example: 20171022}: Lower bounds on
$\varepsilon_{X|Y}$.}
\renewcommand{\arraystretch}{1.5}
\centering
\begin{tabular}{|c|c||c|c|} \hline
$\alpha$ & \eqref{eq: 20171020-b} & \eqref{eq1: generalized Fano}
& \eqref{eq: LB via HolderI} \\[0.1cm] \hline
$-1$ & 0.463 & -- & 0.447 \\
$-\tfrac12$ & 0.475 & -- & 0.355 \\
$-\tfrac14$ & 0.482 & -- & 0.206 \\
\; $\tfrac15$  & 0.494 & 0.523 & -- \\
\; $\tfrac12$  & 0.502 & 0.530 & -- \\
\; $\tfrac45$  & 0.510 & 0.536 & -- \\
\hline
\end{tabular}
\end{table}
Table~\ref{Table: 20171024} shows a slight advantage of the
lower bound in \eqref{eq1: generalized Fano} over \eqref{eq: 20171020-b} for $\alpha \in (0,1)$,
and a superiority of the lower bound in \eqref{eq: 20171020-b} over \eqref{eq: LB via HolderI}
for some negative values of $\alpha$.

In view of Remark~\ref{remark: Shannon}, the lower bound in \eqref{eq: 20171020-b} for $\alpha$
close to~1 is compared with Shannon's lower bound in \eqref{eq: Shannon}.
For $\alpha = 0.99$, the lower bound in \eqref{eq: 20171020-b} is equal to 0.515 (note that it
is slightly looser than \eqref{eq1: generalized Fano}, which is equal to 0.540); on the
other hand, the lower bound in \eqref{eq: Shannon} is equal to 0.146.
\end{example}

Theorem~\ref{theorem: g_X|Y - epsilon_X|Y} establishes relationships between the $\rho$-th
moment of the optimal guessing function, for fixed $\rho > 0$, and the MAP error probability.
This characterizes the exact locus of their attainable values, as it is shown in Figure~\ref{fig: rho=1}
for $\rho=1$. A natural question is whether the MAP error
probability can be uniquely determined on the basis of the knowledge of these $\rho$-th moments
for all $\rho > 0$. The following result answers this question in the affirmative, also suggesting
an easy way to determine the MAP error probability on the basis of the knowledge of these $\rho$-th
moments at an arbitrarily small right neighborhood of $\rho=0$.
\begin{theorem}  \label{thm: 20171022}
Let $X$ and $Y$ be discrete random variables taking values on sets
$\set{X} = \{1, \ldots, M\}$ and $\set{Y}$, respectively. For an integer
$k \geq 0$, denote
\begin{align}  \label{eq: 20171022-a}
z_k = \frac{\mathrm{d}^k}{\mathrm{d}\rho^k} \; \expectation[g_{X|Y}^{\rho}(X|Y)] \Bigl|_{\rho=0}.
\end{align}
Then,
\begin{align}  \label{eq: 20171025-b}
\varepsilon_{X|Y} = 1 - \frac1{c_M}  \, \left|\begin{array}{cccc}
 z_0  &   1  &   \cdots  &   1 \\
 z_1  &  \log_{\mathrm{e}} 2  &   \cdots  &   \log_{\mathrm{e}} M \\
 \vdots  & \vdots & \vdots & \vdots \\
 z_{M-1}  &  \log_{\mathrm{e}}^{M-1} 2   &  \cdots  &   \log_{\mathrm{e}}^{M-1} M
\end{array}
\right|
\end{align}
with
\begin{align} \label{eq: 20171031-b}
c_M = \begin{dcases}
\log_{\mathrm{e}} 2, & \quad M=2,  \\
\prod_{k=2}^M \log_{\mathrm{e}} k \, \prod_{2 \leq i < j \leq M} \log_{\mathrm{e}} \Bigl(\frac{j}{i}\Bigr),  & \quad M \geq 3.
\end{dcases}
\end{align}
\end{theorem}

\begin{proof}
Let $x_i(y) \in \set{X}$ be the element that satisfies \eqref{170824-2}
for $i \in \{1, \ldots, M\}$ and $y \in \set{Y}$, and consider the
ranking function $g_{X|Y}$ that satisfies
\begin{align}  \label{eq: 20171022-b}
g_{X|Y}(x_i(y) | y) = i
\end{align}
for all $i$ and $y$ as above. Then,
\begin{align} \label{eq: 20171022-c}
\expectation[g_{X|Y}^{\rho}(X|Y)] &= \sum_{y \in \set{Y}}
\left\{ P_Y(y) \sum_{x \in \set{X}} P_{X|Y}(x|y) \, g_{X|Y}^{\rho}(x|y) \right\} \\
\label{eq: 20171022-d}
&= \sum_{y \in \set{Y}} \left\{ P_Y(y) \sum_{i=1}^M P_{X|Y}(x_i(y) | y) \, i^\rho \right\}
\end{align}
where \eqref{eq: 20171022-c} and \eqref{eq: 20171022-d} follow from \eqref{yiannis}
and \eqref{eq: 20171022-b}, respectively. Swapping the order of summation in
\eqref{eq: 20171022-d} yields
\begin{align} \label{eq: 20171022-e}
\expectation[g_{X|Y}^{\rho}(X|Y)] &= \sum_{i=1}^M \left\{ i^{\rho} \sum_{y \in \set{Y}} P_{X,Y}(x_i(y), y) \right\} \\
\label{eq: 20171022-f}
&= \sum_{i=0}^{M-1} (i+1)^{\rho} \, u_i
\end{align}
where $u_i$ is the inner sum in the right side of \eqref{eq: 20171022-e} with $i$ replaced by $i+1$.
Taking the $k$-th derivative of \eqref{eq: 20171022-f} at $\rho=0$, and recalling \eqref{eq: 20171022-a} yields
\begin{align}
z_k = \begin{dcases}
\sum_{i=0}^{M-1} u_i, & \quad k = 0, \\
\sum_{i=1}^{M-1} u_i \log_{\mathrm{e}}^{k}(i+1), & \quad k \in \{1, \ldots, M-1\},
\end{dcases}
\end{align}
which gives the set of $M$ linear equations
\begin{align}  \label{eq: 20171022-g}
\left( \begin{array}{cccc}
 1  &   1  &   \cdots  &   1 \\
 0  &  \log_{\mathrm{e}} 2  &   \cdots  &   \log_{\mathrm{e}} M \\
 0  & \vdots & \vdots & \vdots \\
 0  &  \log_{\mathrm{e}}^{M-1} 2   &  \cdots  &   \log_{\mathrm{e}}^{M-1} M
\end{array}
\right)
\left( \begin{array}{c}
u_0 \\
u_1 \\
\vdots \\
u_{M-1}
\end{array}
\right) =
\left( \begin{array}{c}
z_0 \\
z_1 \\
\vdots \\
z_{M-1}
\end{array}
\right)
\end{align}
with the $M$ unknown variables $\mathbf{u}^\top = (u_0, \ldots, u_{M-1})$.
The equations in \eqref{eq: 20171022-g} are linearly independent since
by Lemma~\ref{lemma: 20171024} in Appendix~\ref{app: d}
\begin{align}  \label{eq: 20171022-h}
\left| \begin{array}{cccc}
 1  &   1  &   \cdots  &   1 \\
 0  &  \log_{\mathrm{e}} 2  &   \cdots  &   \log_{\mathrm{e}} M \\
 0  & \vdots & \vdots & \vdots \\
 0  &  \log_{\mathrm{e}}^{M-1} 2   &  \cdots  &   \log_{\mathrm{e}}^{M-1} M
\end{array}
\right| = c_M \neq 0
\end{align}
with $c_M$ as defined in \eqref{eq: 20171031-b}, so $u_0$ in \eqref{eq: 20171022-g}
is uniquely determined. We have
\begin{align}  \label{eq: 20171022-j}
\varepsilon_{X|Y} &= 1 - \sum_{y \in \set{Y}} P_{X,Y}(x_1(y),y) \\
\label{eq: 20171022-k}
&= 1-u_0
\end{align}
where \eqref{eq: 20171022-j} follows from \eqref{eq: epsilon_X|Y} and since, by definition,
$x_1(y)$ is a mode of $P_{X|Y}(\cdot | y)$ for all $y \in \set{Y}$;
\eqref{eq: 20171022-k} holds by the definition of $u_0$ as the inner sum in the right
side of \eqref{eq: 20171022-e} with $i=1$. Finally, \eqref{eq: 20171025-b} follows from
\eqref{eq: 20171022-g}--\eqref{eq: 20171022-k} and Cramer's rule.
\end{proof}

\begin{remark}
Theorem~\ref{thm: 20171022} can be, alternatively, first proved without side information.
Then, by replacing $P_X$ by $P_{X|Y}(\cdot|y)$, \eqref{eq: 20171022-j}--\eqref{eq: 20171022-k}
hold for $\varepsilon_{X|Y=y}$; finally, it holds by averaging in view of the linearity of
$\mathbf{u}$ in $\mathbf{z}$.
\end{remark}

\section{Non-Asymptotic Bounds for Optimal Fixed-to-Variable Lossless Compression}
\label{section: lossless source coding}

This section applies the improved bounds on guessing moments in
Section~\ref{section: bounds for guessing} to analyze
non-prefix one-to-one binary optimal codes, which do not satisfy Kraft's
inequality. These are one-shot codes that assign distinct codewords
to source strings; their average length per source symbol, which is smaller
than the Shannon entropy of the source, is analyzed
in \cite[Section~7]{DrmotaS} and \cite{WSSV11}.
Preliminary material is introduced in
Section~\ref{subsec - lossless preliminaries}, improved bounds
on the distribution of optimal codeword lengths are derived in
Section~\ref{subsec: distribution of lengths}, and improved
non-asymptotic bounds for fixed-to-variable
codes are derived in Section~\ref{subsec: FV lossless codes}.

\subsection{Basic setup, notation and preliminaries}
\label{subsec - lossless preliminaries}

\begin{definition}
A variable-length lossless compression binary code for a discrete
set $\set{X}$ is an injective mapping:
\begin{align}
f \colon \set{X} \to \{0, 1\}^{\ast}
= \{\varnothing, 0, 1, 00, 01, 10, 11, 000, \ldots \}
\end{align}
where $f(x)$ is the codeword which is assigned to $x \in \set{X}$; the
length of this codeword is denoted by $\ell(f(x))$ where $\ell \colon \{0,1\}^{\ast}
\to \{0, 1, 2, \ldots\}$ with the convention that $\ell(\varnothing)=0$.
\end{definition}

\begin{definition} \cite{verdubook}
\label{def: compact}
A variable-length lossless source code is {\em compact} whenever it contains
a codeword only if all shorter codewords also belong to the code.
\end{definition}

\begin{definition} \cite{verdubook}
\label{def: efficient}
Given a probability mass function $P_X$ on $\set{X}$, a variable-length
lossless source code is {\em $P_X$-efficient} if for all $(a,b) \in \set{X}^2$,
\begin{align}
\ell(f(a)) < \ell(f(b)) \Longrightarrow P_X(a) \geq P_X(b).
\end{align}
\end{definition}

\begin{definition} \cite{verdubook}
\label{def: opt}
Given a probability mass function $P_X$ on $\set{X}$, a variable-length
lossless source code is {\em $P_X$-optimal} if it is both compact and
$P_X$-efficient.
\end{definition}

The optimality in Definition~\ref{def: opt} is justified in
Proposition~\ref{prop: opt}.
Let $f_X^{\ast} \colon \set{X} \to \{0, 1\}^{\ast}$ be a $P_X$-optimal
variable-length lossless source code. If $|\set{X}| < \infty$, then
\begin{enumerate}[a)]
\item $\varnothing$ is assigned to the most likely element in $\set{X}$.
\item All the $2^\ell$ binary strings of length $\ell$ are assigned to
the $2^\ell$-th through $(2^{\ell+1}-1)$-th most likely elements with
$\ell \in \{1, \ldots, \lfloor \log_2(1+|\set{X}|) \rfloor - 1\}$. For
example, 0 and 1 (or 1 and 0) are assigned, respectively,
to the second and third most likely elements in $\set{X}$.
\item If $\log_2(1+|\set{X}|)$ is not an integer, then codewords of
length $\lfloor \log_2(1+|\set{X}|) \rfloor$ are assigned to each of
the remaining $1+|\set{X}|-2^{\lfloor \log_2(1+|\set{X}|) \rfloor}$
elements in $\set{X}$.
\end{enumerate}

As long as $|\set{X}| > 1$, there is more than one $P_X$-optimal code
since compactness and $P_X$-efficiency are preserved by swapping codewords of
the same length (and, if $|\set{X}|=2$, then the second most likely element
can be either assigned $0$ or $1$).
In the presence of ties among probabilities, the value of $\ell(f_X^{\ast}(x))$
for some $x \in \set{X}$ may depend on the choice of $f_X^{\ast}$. The following
result provides several relevant properties of optimal codes.

\begin{proposition} \label{prop: opt} (\cite{KYSV14, verdubook})
Fix a probability mass function $P_X$ on a finite set $\set{X}$. The following
results hold for $P_X$-optimal codes $f_X^{\ast} \colon \set{X} \to \{0, 1\}^{\ast}$:
\begin{enumerate}[a)]
\item \label{opt-1}
The distribution of $\ell(f_X^{\ast}(X))$ is invariant to the actual choice of
$f_X^{\ast}$, and it only depends on $P_X$.
\item \label{opt-2}
For every lossless data compression code $f$, and for all $r \geq 0$,
\begin{align} \label{opt. dist}
\prob\bigl[\ell\bigl(f(X)\bigr) \leq r \bigr] \leq \prob\bigl[\ell\bigl(f_X^{\ast}(X)\bigr) \leq r \bigr].
\end{align}
Furthermore, the inequality in \eqref{opt. dist} is strict for some
$r \geq 0$ if $f$ is not $P_X$-optimal.
\item  \label{opt-3}
\vspace*{-0.2cm}
\begin{align} \label{Lemma 1 - CV14}
\sum_{x \in \set{X}} 2^{-\ell(f_X^{\ast}(x))} \leq \log_2(1+|\set{X}|)
\end{align}
with equality if and only if $|\set{X}|+1$ is a positive integral power of~2. Furthermore, all
compact codes for $\set{X}$ achieve the same value of $\underset{x \in \set{X}}{\sum} 2^{-\ell(f(x))}$,
which is larger than that achieved by a non-compact code.
\end{enumerate}
\end{proposition}

\begin{definition}
The {\em cumulant generating function} of the codeword lengths of $P_X$-optimal binary codes
is given by
\begin{align} \label{eq: cumulant}
\Lambda^{\ast}(\rho) \triangleq
\log \expectation \bigl[ 2^{\rho \, \ell(f_X^{\ast}(X))} \bigr],
\quad \rho \in \Reals.
\end{align}
\end{definition}

\begin{remark}
\eqref{eq: cumulant} is actually a scaled cumulant generating function.
The cumulant generating function of a random variable $X$ is given by
\begin{align} \label{eq: 20171018-a}
\Lambda_X(\rho) = \log_{\mathrm{e}} \expectation \bigl[ e^{\rho X} \bigr], \quad \rho \in \Reals
\end{align}
whereas, following Campbell \cite{Campbell65}, it is more natural to study the function given by
\begin{align} \label{eq: 20171018-b}
\widetilde{\Lambda}_X(\rho) = \log \expectation \bigl[ 2^{\rho X} \bigr].
\end{align}
Note, however, that \eqref{eq: 20171018-a}
and \eqref{eq: 20171018-b} satisfy
\begin{align} \label{eq: 20171018-c}
\widetilde{\Lambda}_X(\rho) =  \Lambda_X(\rho \log_{\mathrm{e}} 2) \, \log \mathrm{e},
\end{align}
which implies that they can be obtained from each other by proper linear scalings of the axes.
\end{remark}

As mentioned in the introduction, the cumulant generating function of the
codeword lengths provides an important design criterion. In particular, it
yields the average length via the equality
\begin{align} \label{eq: 20171017-a}
\lim_{\rho \to 0} \frac{\Lambda^{\ast}(\rho)}{\rho} = \expectation[\ell(f_X^{\ast}(X))].
\end{align}
\begin{theorem} \cite[Theorem~1]{CV2014a}
\label{thm1 - CourtadeV14}
If $\rho \in (-\infty, -1]$, then
\begin{align} \label{eq1: Th. 1 CV-ISIT14}
H_{\infty}(X) - \log \log_2 (1+|\set{X}|) \leq
-\Lambda^{\ast}(\rho) \leq H_{\infty}(X),
\end{align}
and, if $\rho \in (-1,0) \cup (0,\infty)$, then
\begin{align}  \label{eq2: Th. 1 CV-ISIT14}
H_{\frac1{1+\rho}}(X) - \log \log_2 (1+|\set{X}|) \leq
\frac{\Lambda^{\ast}(\rho)}{\rho} \leq H_{\frac1{1+\rho}}(X).
\end{align}
\end{theorem}
By invoking the Chernoff bound and using Theorem~\ref{thm1 - CourtadeV14},
the following result holds.
\begin{theorem} \cite[Theorem~2]{CV2014a}
\label{thm2 - CourtadeV14}
For all $H(X) < R < \log |\set{X}|$
\begin{align} \label{eq1: Th. 2 in CV2014a}
\log \frac1{\mathbb{P}[\ell(f_X^{\ast}(X)) \geq R]}
& \geq \sup_{\rho > 0} \Bigl\{ \rho R - \rho H_{\frac1{1+\rho}}(X) \Bigr\} \\
\label{eq2: Th. 2 in CV2014a}
& = D(X_\alpha \| X)
\end{align}
where $\alpha \in (0,1)$ is a function of $R$ chosen so that $R = H(X_\alpha)$,
and $X_\alpha$ has the scaled probability mass function
\begin{align} \label{eq: scaled distribution}
P_{X_\alpha}(x) = \frac{P_X^{\alpha}(x)}{\underset{a \in \set{X}}{\sum}
P_X^{\alpha}(a)}, \quad x \in \set{X}.
\end{align}
\end{theorem}

\subsection{Improved bounds on the distribution of the optimal code lengths}
\label{subsec: distribution of lengths}
We derive bounds on the cumulant generating function and the complementary
cumulative distribution of optimal lengths for lossless compression of a
random variable $X$ which takes values on a finite set $\set{X}$. These bounds
improve those in Theorems~\ref{thm1 - CourtadeV14} and~\ref{thm2 - CourtadeV14},
and in Section~\ref{subsec: FV lossless codes} we use them to derive non-asymptotic
bounds for optimal fixed-to-variable lossless codes.

We start by generalizing \cite[Lemma~1]{CV2014a} from $\beta=1$ to arbitrary
$\beta \in \Reals$.
\begin{lemma} \label{lemma: l inequality}
For an optimal binary code, and for all $\beta \in \Reals$
\begin{align}
\label{eq: l inequality}
\sum_{x \in \set{X}} 2^{-\beta \, \ell(f_X^{\ast}(x))} &=
\begin{dcases}
(2^{\Delta} - 1) s_{\beta}^m
+ \frac{1-s_{\beta}^m}{1-s_{\beta}},
& \quad \beta \neq 1 \\[0.1cm]
m + 2^{\Delta} - 1, & \quad \beta = 1
\end{dcases} \\
\label{eq: t}
& \triangleq t(\beta, |\set{X}|)
\end{align} \\
where
\begin{align}
\label{eq: s_beta}
& s_{\beta} = 2^{1-\beta}, \\
\label{eq: integer m}
& m = \bigl\lfloor \log_2(1+|\set{X}|) \bigr\rfloor, \\
\label{eq: Delta}
& \Delta = \log_2(1+|\set{X}|) - m \in [0,1).
\end{align}
\end{lemma}

\begin{proof}
From Definition~\ref{def: opt}, there are $2^\ell$ elements
of $\set{X}$ which are assigned codewords of length $\ell$ with
$\ell \in \{0, 1, \ldots, m-1\}$, and $|\set{X}|-2^m+1$ elements
which are assigned codewords of length $m$. Hence, for
$\beta \in \Reals$,
\begin{align}
\sum_{x \in \set{X}} 2^{-\beta \, \ell(f_X^{\ast}(x))}
\label{eq: attained UB}
= (|\set{X}|-2^m+1) \, 2^{-\beta m} + \sum_{\ell=0}^{m-1} 2^\ell \, 2^{-\beta \ell}.
\end{align}
In view of \eqref{eq: integer m} and \eqref{eq: Delta}, for $\beta=1$,
\eqref{eq: l inequality} follows from \eqref{eq: attained UB}.
If $\beta \neq 1$, then from \eqref{eq: s_beta}--\eqref{eq: attained UB}
\begin{align}
\sum_{x \in \set{X}} 2^{-\beta \, \ell(f_X^{\ast}(x))}
&= \left[(1+|\set{X}|) 2^{-m}-1\right] \, 2^{(1-\beta)m} +
\frac{1-2^{(1-\beta)m}}{1-2^{1-\beta}} \\
& = (2^\Delta-1) s_{\beta}^m + \frac{1-s_{\beta}^m}{1-s_{\beta}}.
\end{align}
\end{proof}

\begin{lemma} \label{lemma: tightened bound on the cumulant}
Let $X$ be a random variable taking values on a finite set
$\set{X}$, and let $\rho \neq 0$.
Then, for an optimal binary code,
\begin{align}
\frac1{\rho} \, \log \expectation\bigl[2^{\rho \,
\ell(f_X^{\ast}(X))} \bigr]
\geq \sup_{\beta \in (-\rho, \infty) \setminus \{0\}}
\frac1{\beta} \left[H_{\frac{\beta}{\beta+\rho}}(X)
- \log \, t(\beta, |\set{X}|) \right],
\label{eq: tightened bound on the cumulant}
\end{align}
where $t(\cdot,\cdot)$ is defined in \eqref{eq: t}.
\end{lemma}
\begin{proof}
Let $g \colon \set{X} \to \{1, 2, 4, 8, \ldots\}$ be defined by
$g(x) = 2^{\ell(f_X^{\ast}(x))}$ for all $x \in \set{X}$.
The result in \eqref{eq: tightened bound on the cumulant}
follows by combining \eqref{eq: key result} and
Lemma~\ref{lemma: l inequality}.
\end{proof}

\begin{lemma} \label{lemma1: UB by guessing moments}
Let $X$ be a random variable taking values on
a finite set $\set{X}$, and let $g_X$ be a ranking function of
$X$. Then, for every optimal binary code and for all
$\rho > 0$,
\begin{align} \label{eq: UB by guessing moments}
2^{-\rho} \, \expectation [g_X^{\rho}(X)]
< \expectation \bigl[2^{\rho \, \ell(f_X^{\ast}(X))}\bigr]
\leq \expectation [g_X^{\rho}(X)].
\end{align}
\end{lemma}
\begin{proof}
Let $x_k$ be the $k$-th most probable outcome of $X$ under
a given tie-breaking rule.
For an optimal compression code, the length of the assigned
codeword of the element $x_k \in \set{X}$ satisfies
$\ell(f_X^{\ast}(x_k)) = \lfloor \log_2 k \rfloor$
(\cite[(15)]{KYSV14}). Hence, for $\rho > 0$,
\begin{align}
\expectation \bigl[2^{\rho \, \ell(f_X^{\ast}(X))} \bigr]
& = \sum_k P_X(x_k) \, 2^{\rho \, \lfloor \log_2 k \rfloor} \label{is} \\
& \leq \sum_k P_X(x_k) \, k^{\rho} \\
& = \expectation \bigl[g_X^{\rho}(X) \bigr].  \label{is2}
\end{align}
This proves the right side of \eqref{eq: UB by guessing moments}.
To show the left side of \eqref{eq: UB by guessing moments}, note that
\begin{align} \label{sv}
\lfloor \log_2 k \rfloor > \log_2 k - 1
\end{align}
which implies from \eqref{is}, \eqref{is2} and \eqref{sv} that for $\rho > 0$
\begin{align}
\expectation \bigl[2^{\rho \, \ell(f_X^{\ast}(X))} \bigr]
& > \sum_k P_X(x_k) \, 2^{\rho \, (\log_2 k - 1)} \\
&= 2^{-\rho} \, \expectation \bigl[g_X^{\rho}(X) \bigr].
\end{align}
\end{proof}

\begin{theorem} \label{theorem: tightened bounds on the cumulant}
Let $X$ be a random variable taking values on a finite set $\set{X}$.
Then, for every optimal binary code, the cumulant generating function
in \eqref{eq: cumulant} satisfies
\begin{align}
& \sup_{\beta \in (-\rho, \infty) \setminus \{0\}}
\frac1{\beta} \left[H_{\frac{\beta}{\beta+\rho}}(X)
- \log \, t(\beta, |\set{X}|) \right] \nonumber \\
\label{eq: normalized cumulant}
& \leq \frac{\Lambda^{\ast}(\rho)}{\rho}  \\
\label{eq: improved UB on normalized cumulant}
& \leq H_{\frac1{1+\rho}}(X) + \frac1{\rho} \, \log \biggl(
\frac1{1+\rho} \left[1 - \exp\left(-\rho H_{\frac1{1+\rho}}(X)
\right) \right] + \exp \left( (\rho-1)^+ \, H_{\frac1\rho}(X)
- \rho H_{\frac1{1+\rho}}(X) \right) \biggr),
\end{align}
for all $\rho > 0$,
where $t(\cdot)$ is given in \eqref{eq: t}. Moreover,
\eqref{eq: normalized cumulant} also holds for $\rho < 0$.
\end{theorem}
\begin{proof}
The lower bound in the left side of \eqref{eq: normalized cumulant}
is Lemma~\ref{lemma: tightened bound on the cumulant}, and the upper
bound in the right side of \eqref{eq: improved UB on normalized cumulant}
follows from Theorem~\ref{theorem: UB-NSI} and
Lemma~\ref{lemma1: UB by guessing moments}.
\end{proof}

\begin{remark}
For $\rho \in (-1, 0) \cup (0, \infty)$, loosening the
bound in the left side of \eqref{eq: normalized cumulant}
by the sub-optimal choice of $\beta=1$ and invoking
$t(1, |\set{X}|) \leq \log_2(1+|\set{X}|)$ (in view of Lemma~\ref{lemma: l inequality},
and since $2^x-1 \leq x$ for $x \in [0,1]$)
recovers the lower bound in \eqref{eq2: Th. 1 CV-ISIT14}.
\end{remark}

\begin{remark} \label{belgrad}
In view of \eqref{ramat-gan}, the second term in
the right side of \eqref{eq: improved UB on normalized cumulant}
is non-positive for all $\rho \geq 1$; due to the non-negativity
of the R\'{e}nyi entropy, this also holds for $\rho \in (0,1)$. Hence,
for $\rho > 0$, the upper bound in
\eqref{eq: improved UB on normalized cumulant} improves the
bound in the right side of \eqref{eq2: Th. 1 CV-ISIT14}.
\end{remark}

The Chernoff bound and \eqref{eq: improved UB on normalized cumulant} readily yield
the following lower bound.
\begin{theorem} \label{theorem: tightened LB on CDF}
Under the assumptions in Theorem~\ref{theorem: tightened bounds on the cumulant},
for $R < \log |\set{X}|$,
\begin{align}
\log \left( \frac1{\mathbb{P} \bigl[\ell(f_X^{\ast}(X)) > R \bigr]} \right)
& \geq \sup_{\rho > 0} \biggl\{ \rho R - \rho H_{\frac1{1+\rho}}(X)
- \log \biggl( \frac1{1+\rho} \left[1 - \exp\left(-\rho H_{\frac1{1+\rho}}(X) \right)
\right] \nonumber \\
& \hspace*{5cm} + \exp \left( (\rho-1)^+ \, H_{\frac1\rho}(X) - \rho H_{\frac1{1+\rho}}(X)
\right) \biggr) \biggr\}.
\label{eq: tightened LB on CDF}
\end{align}
\end{theorem}

\begin{remark}
The bound in \eqref{eq: tightened LB on CDF} is strictly tighter than
the right side in \eqref{eq2: Th. 2 in CV2014a} unless $X$ is deterministic.
To show this, note that in view of
Remark~\ref{belgrad} and \eqref{eq1: Th. 2 in CV2014a}--\eqref{eq2: Th. 2 in CV2014a},
the bound in \eqref{eq: tightened LB on CDF} cannot be looser than the right side in
\eqref{eq2: Th. 2 in CV2014a}. Moreover,
since the function $\rho R - \rho H_{\frac1{1+\rho}}(X)$ is concave in
$\rho > 0$, its local (and hence global) maximum is achieved
at $\rho^{\ast} > 0$; from the proof of \cite[Theorem~2]{CV2014a}, this
value of $\rho^{\ast}$ satisfies
\begin{align}
\label{eq1: rho star}
& \alpha = \frac1{1+\rho^{\ast}}, \quad H(X_{\alpha}) = R, \\
\label{eq2: rho star}
& \rho^{\ast} R - \rho^{\ast}
H_{\frac1{1+\rho^\ast}}(X) = D(X_\alpha \| X)
\end{align}
where $X_{\alpha}$ has the scaled probability mass function in
\eqref{eq: scaled distribution}. Hence, by replacing the supremum over
$\rho > 0$ with the value at
$\rho = \rho^{\ast}$, \eqref{eq2: rho star} yields the following loosened bound:
\begin{align}
\log \left( \frac1{\mathbb{P}
\bigl[\ell(f_X^{\ast}(X)) > R \bigr]} \right)
& \geq D(X_\alpha \| X)
- \log \biggl( \frac1{1+\rho^{\ast}} \left[1 - \exp\left(-\rho^{\ast}
H_{\frac1{1+\rho^{\ast}}}(X) \right) \right] \nonumber \\
& \hspace*{3.5cm} + \exp \left( (\rho^{\ast}-1)^+
\, H_{\frac1{\rho^\ast}}(X) - \rho^{\ast} H_{\frac1{1+\rho^{\ast}}}(X)
\right) \biggr).
\label{eq: weaker LB on CDF}
\end{align}
Hence, it is enough to prove that
\eqref{eq: weaker LB on CDF} is strictly tighter than \eqref{eq2: Th. 2 in CV2014a}.
This follows from the strict inequality in \eqref{ramat-gan} whenever $X$
is non-deterministic, thus implying that, for all $\rho > 0$,
\begin{align}
\sum_{x \in \set{X}} P_X^{\frac1{1+\rho}}(x) = \exp\left(\tfrac{\rho}{1+\rho} \;
H_{\frac1{1+\rho}}(X) \right) > 1.
\end{align}
\end{remark}

\subsection{Non-asymptotic bounds for fixed-to-variable lossless source codes}
\label{subsec: FV lossless codes}

We consider the fixed-to-variable-length
lossless compression in Definition~\ref{def: opt} where the
object to be compressed $x^n = (x_1, \ldots, x_n) \in \set{A}^n$ is a
string of length $n$ ($n$ is known to both encoder and decoder),
whose letters are drawn from a finite alphabet
$\set{X}$ according to the probability mass function
$P_{X^n}(x^n) = \overset{n}{\underset{i=1}{\prod}} P_X(x_i)$ for all $x^n \in \set{A}^n$.
We consider the following non-asymptotic measures for optimal fixed-to-variable
lossless compression:
\begin{itemize}
\item The cumulant generating function of the codeword lengths is given by
\begin{align}
\Lambda_n(\rho) :=
\frac1n \, \log \expectation \left[2^{\rho \, \ell(f_{X^n}^{\ast}(X^n))} \right],
\quad \rho \in \Reals.
\end{align}
\item The non-asymptotic version of the source reliability function is given by
\begin{align}
E_n(R) = \frac1n \, \log \left( \frac1{\mathbb{P} \bigl[ \frac1n \,
\ell(f_{X^n}^{\ast}(X^n)) \geq R \bigr]} \right).
\end{align}
\end{itemize}

\begin{theorem} \label{theorem: non-asymp. FV}
Consider a memoryless and stationary source of finite alphabet
$\set{A}$, and let $f_{X^n}^{\ast} \colon \set{A}^n \to \{0,1\}^{\ast}$
be an optimal compression code. Then, the following bounds hold:
\begin{enumerate}[a)]
\item For all $\rho > 0$
\begin{align}
& \sup_{\beta \in (-\rho, \infty) \setminus \{0\}}
\frac{\rho}{\beta} \left[H_{\frac{\beta}{\beta+\rho}}(X)
- \frac1n \, \log t(\beta, |\set{A}|^n) \right] \nonumber \\[0.1cm]
\label{eq2: FV normalized cumulant}
& \leq \Lambda_n(\rho)  \\
\label{eq3: FV normalized cumulant}
& \leq \rho H_{\frac1{1+\rho}}(X) + \frac1n \, \log \biggl(
\frac1{1+\rho} \left[1 - \exp\left(-n\rho H_{\frac1{1+\rho}}(X)
\right) \right] \nonumber \\
& \hspace*{3.7cm} + \exp \left( n \left[ (\rho-1)^+ \, H_{\frac1\rho}(X)
- \rho H_{\frac1{1+\rho}}(X) \right] \right) \biggr)
\end{align}
where $t(\cdot)$ is as defined in \eqref{eq: t}.
\item For $R < \log |\set{A}|$
\begin{align}
& E_n(R) \geq \sup_{\rho > 0} \biggl\{ \rho R - \rho H_{\frac1{1+\rho}}(X)
- \frac1n \log \biggl( \frac1{1+\rho} \left[1 - \exp\left(-n\rho H_{\frac1{1+\rho}}(X) \right)
\right] \nonumber \\
\label{eq5: FV reliability}
& \hspace*{5cm} + \exp \left( n \left[ (\rho-1)^+ \, H_{\frac1\rho}(X)
- \rho H_{\frac1{1+\rho}}(X) \right] \right) \biggr) \biggr\}.
\end{align}
\end{enumerate}
\end{theorem}

\begin{proof}
Items a) and b) follow, respectively, from Theorems~\ref{theorem: tightened bounds on the cumulant}
and~\ref{theorem: tightened LB on CDF} with $|\set{X}|$ replaced by $|\set{A}|^n$, and since
$H_{\alpha}(X^n) = n H_{\alpha}(X)$ holds for i.i.d. random variables and for all $\alpha > 0$
(see Lemma~\ref{lemma: RE - indep.}).
\end{proof}

\begin{remark} \label{remark: asym}
The non-asymptotic bounds on the cumulant generating function
in \eqref{eq2: FV normalized cumulant}--\eqref{eq3: FV normalized cumulant}
recover the known asymptotic result in \cite[(29)]{CV2014a} where for all $\rho > 0$
\begin{align} \label{eq4: FV normalized cumulant}
\Lambda(\rho) := \lim_{n \to \infty} \Lambda_n(\rho) = \rho H_{\frac1{1+\rho}}(X),
\end{align}
which, incidentally, coincides with Arikan's asymptotic fundamental limit for
$\underset{n \to \infty}{\lim} \frac1{n} \log \expectation[g_{X^n}^{\rho}(X^n)]$
when $X^n$ is i.i.d. \cite{Arikan96}.
To this end, note that $\underset{n \to \infty}{\lim} \tfrac1n \, \log t(1, |\set{A}|^n) = 0$,
and selecting $\beta=1$ in the left side of \eqref{eq2: FV normalized cumulant}
yields
\begin{align} \label{eq: liminf cumulant}
\varliminf_{n \to \infty} \Lambda_n(\rho) \geq \rho H_{\frac1{1+\rho}}(X).
\end{align}
Moreover, since $H_{\alpha}(X)$ is monotonically non-increasing in $\alpha$,
\begin{align} \label{lll}
\rho H_{\frac1{1+\rho}}(X) - (\rho-1)^+ \, H_{\frac1\rho}(X) \geq \min\{\rho,1\} \, H_{\frac1\rho}(X)
\end{align}
and, if $X$ is non-deterministic, then \eqref{eq3: FV normalized cumulant}
and \eqref{lll} yield
\begin{align} \label{eq: limsup cumulant}
\varlimsup_{n \to \infty} \Lambda_n(\rho) \leq \rho H_{\frac1{1+\rho}}(X),
\end{align}
recovering \eqref{eq4: FV normalized cumulant} from \eqref{eq: liminf cumulant} and
\eqref{eq: limsup cumulant}.
Furthermore, \eqref{eq5: FV reliability} and \eqref{lll} imply that
\begin{align} \label{eq6: FV reliability}
E(R) &:= \varliminf_{n \to \infty} E_n(R) \\
&\geq \sup_{\rho > 0} \Bigl\{ \rho R - \rho H_{\frac1{1+\rho}}(X) \Bigr\}.
\end{align}
\end{remark}

Although, as noted in Remark~\ref{remark: asym}, the improvement in the bounds
afforded by Theorem~\ref{theorem: non-asymp. FV} washes out asymptotically, the
following example illustrates the improvement in the non-asymptotic regime.

\begin{example} \label{example: DMS}
Consider a discrete memoryless source which emits a string of $n$ letters
from the alphabet $\set{A} = \{a, b, c\}$ with $P_X(a) = \tfrac47$,
$P_X(b) = \tfrac27$ and $P_X(c) = \tfrac17$.

The bounds on the cumulant generating function in \cite[Theorem~1]{CV2014a}
(see \eqref{eq2: Th. 1 CV-ISIT14}) are given by
\begin{align} \label{eq: looser cumulant bounds}
\rho \, H_{\frac1{1+\rho}}(X) - \tfrac{\rho}{n} \, \log \log_2 (1+|\set{A}|^n)
\leq \Lambda_n(\rho) \leq \rho H_{\frac1{1+\rho}}(X)
\end{align}
for $\rho > 0$. Figure~\ref{figure: cumulant bounds} compares \eqref{eq: looser cumulant bounds}
with the improved bounds in \eqref{eq2: FV normalized cumulant}--\eqref{eq3: FV normalized cumulant}.
For $n=10$, the upper and lower bounds are compared to the exact normalized
cumulant (see the left plot in Figure~\ref{figure: cumulant bounds}); this
indicates that the lower bound in \eqref{eq2: FV normalized cumulant} can be tight
even for small values of $n$. The match between the upper and lower bounds in
\eqref{eq2: FV normalized cumulant}--\eqref{eq3: FV normalized cumulant} improves
by increasing $n$, and the tightening obtained by
the lower bound in \eqref{eq2: FV normalized cumulant} can be significant
for small values of $n$.

\begin{figure}[h]
\vspace*{-4.4cm}
\hspace*{-0.5cm}
\includegraphics[width=10cm]{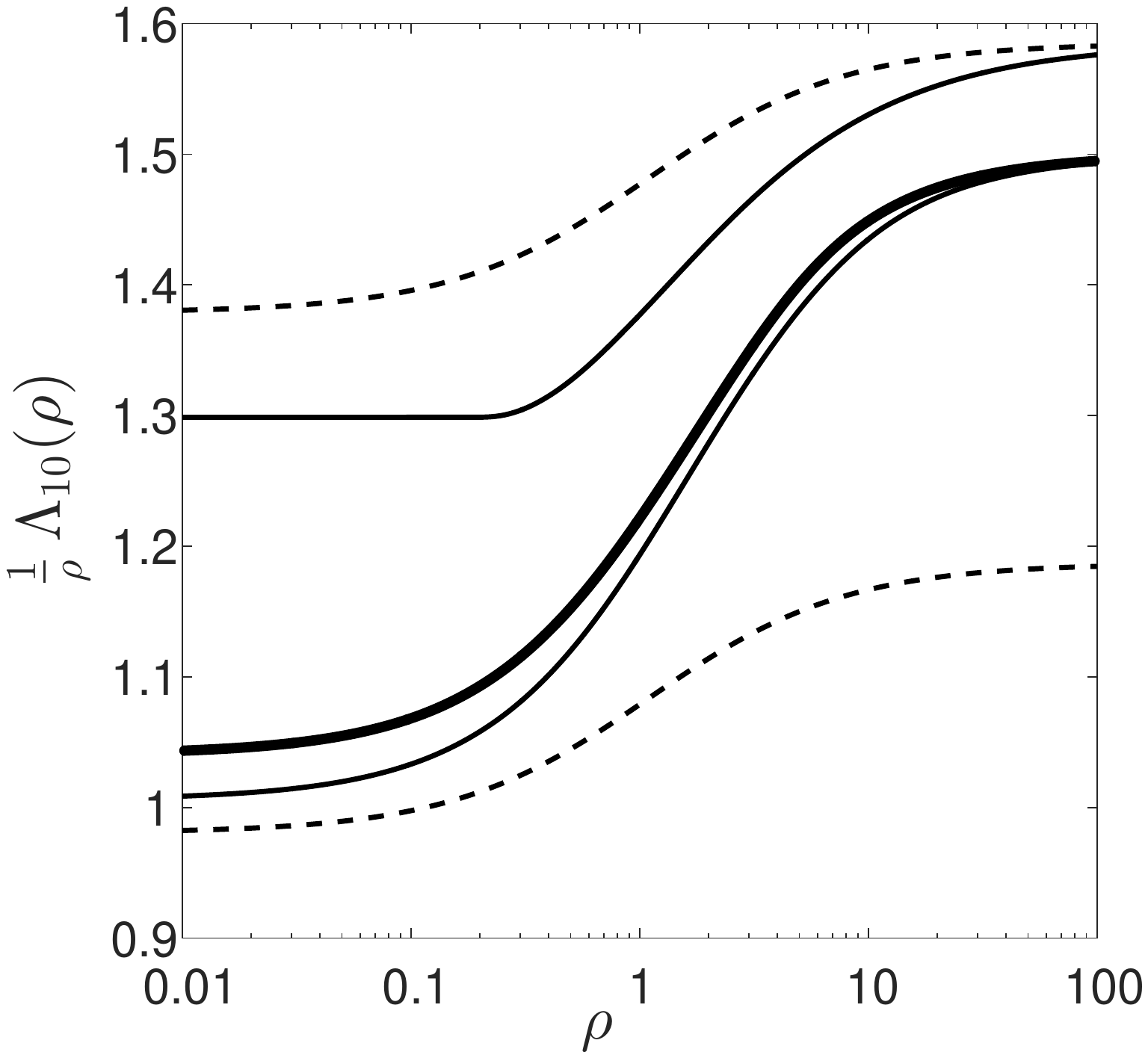}
\hspace*{-1.5cm}
\includegraphics[width=10cm]{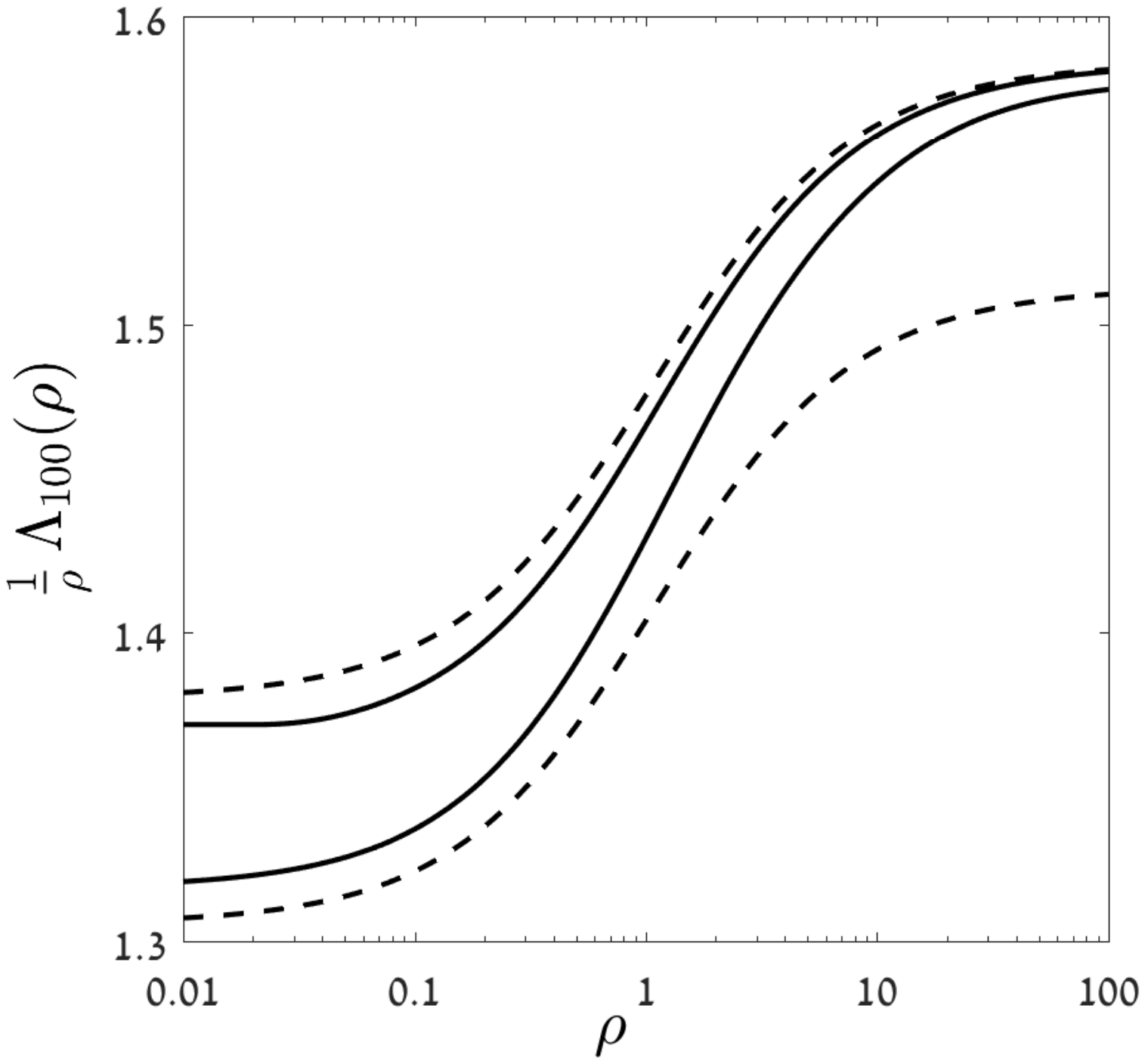}
\vspace*{-4.5cm}
\caption{\label{figure: cumulant bounds}
Bounds on the normalized cumulant generating function,
$\frac{\Lambda_n(\rho)}{\rho}$ (in bits), of the codeword
lengths of optimal
lossless compression of strings of length $n$ emitted
from the discrete memoryless source in
Example~\ref{example: DMS}. The dashed lines are the
bounds in \eqref{eq: looser cumulant bounds}, and the
thin solid lines refer to the improved bounds in
\eqref{eq2: FV normalized cumulant}--\eqref{eq3: FV normalized cumulant}.
The left plot corresponds to $n=10$, in which case we
can compute the exact normalized cumulant (thick
solid curve). The right plot corresponds to $n=100$.}
\end{figure}
\begin{figure}[h]
\hspace*{0.3cm}
\vspace*{-0.1cm}
\includegraphics[width=8.2cm]{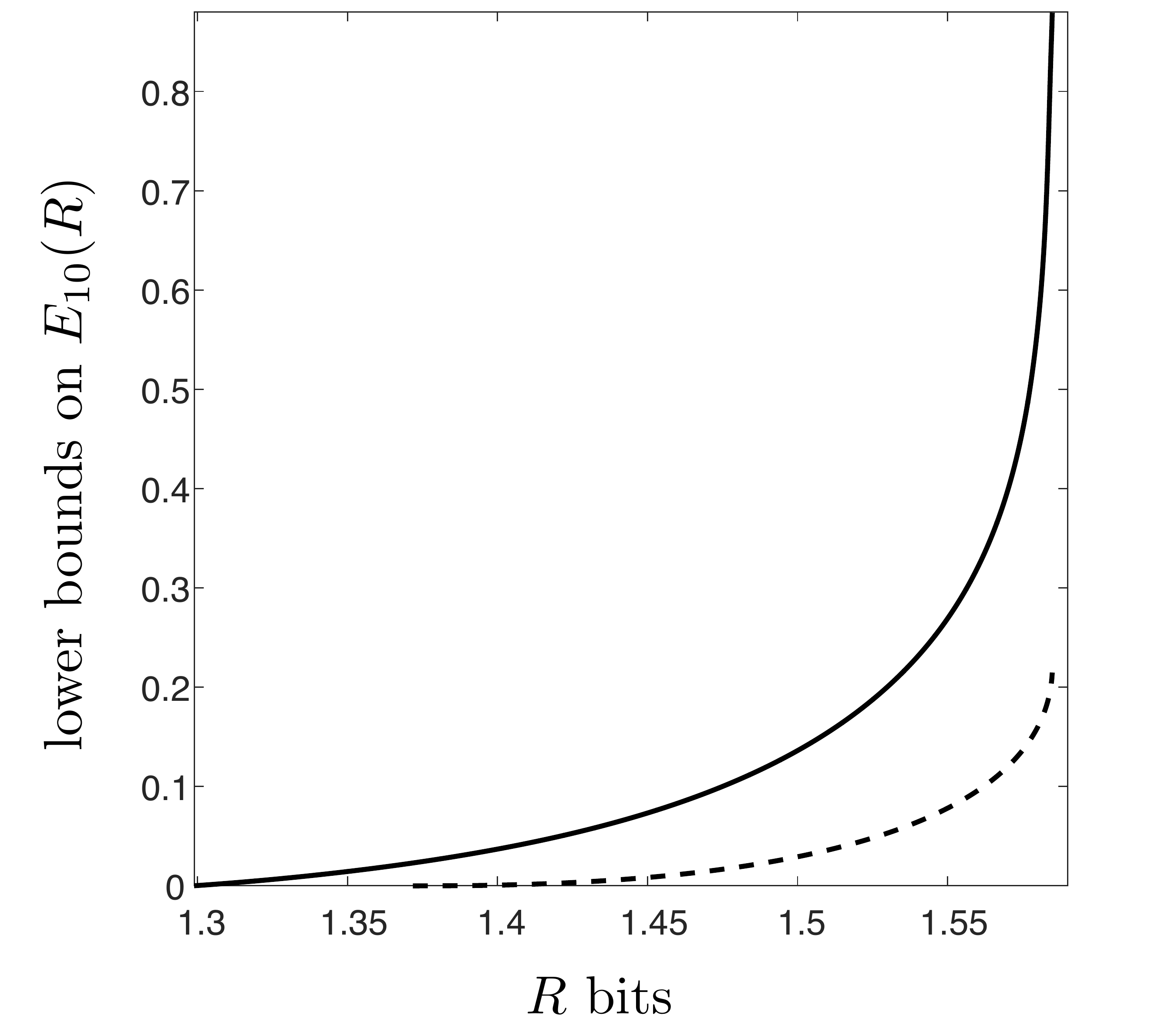}
\hspace*{0.5cm}
\includegraphics[width=7.9cm]{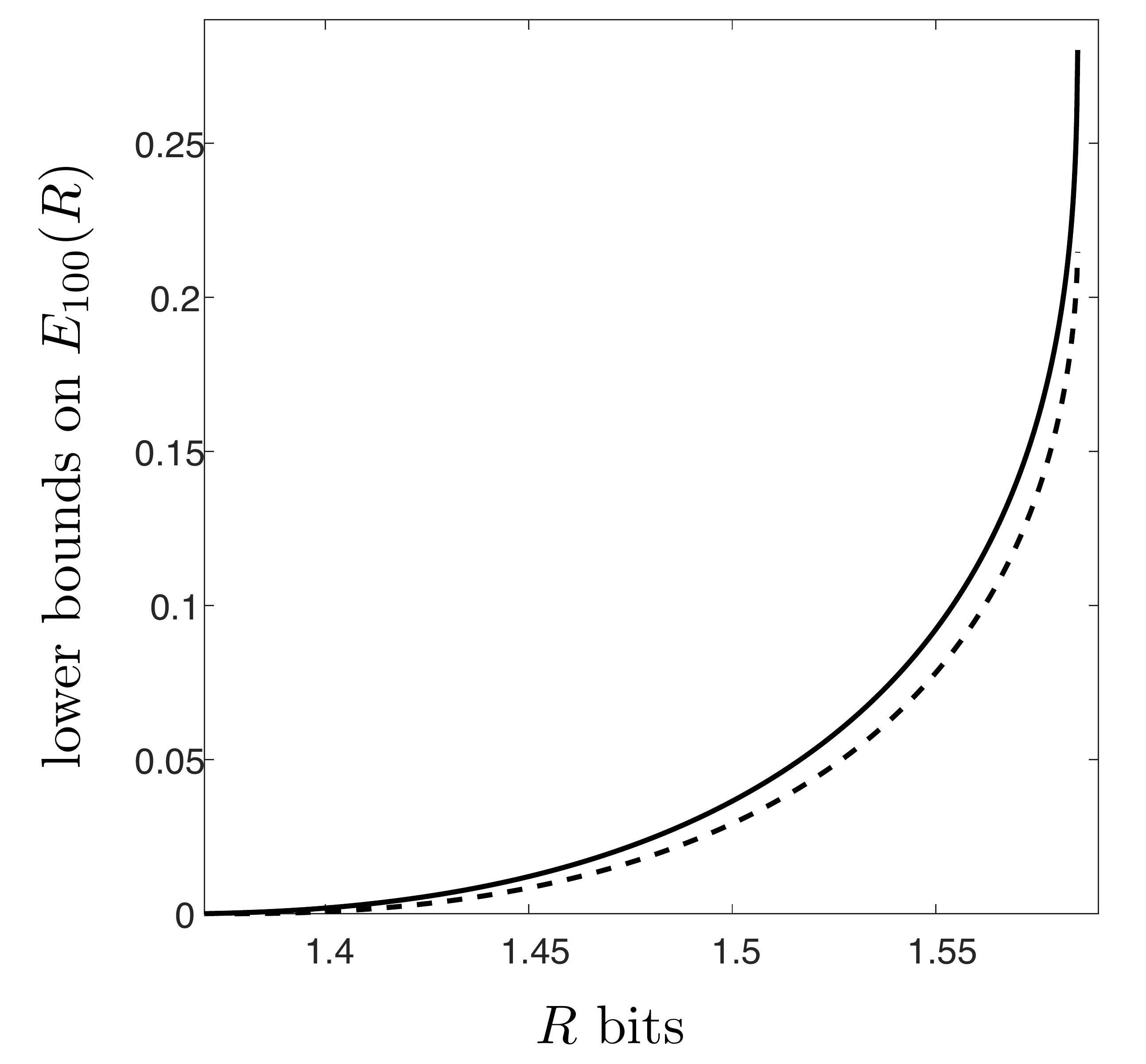}
\vspace*{-0.9cm}
\caption{\label{figure: source reliability bounds}
Lower bounds on the non-asymptotic reliability function, $E_n(R)$ (base~2),
for the discrete memoryless source in Example~\ref{example: DMS}. In each plot,
the dashed curve refers to the lower bound on $E_n(R)$ in Theorem~\ref{thm2 - CourtadeV14},
and the solid curve refers to the tighter lower bound in \eqref{eq5: FV reliability}.
The left and right plots correspond to $n=10$ and $n=100$, respectively.}
\end{figure}

Lower bounds on the non-asymptotic source reliability function $E_n(R)$ are plotted in
Figure~\ref{figure: source reliability bounds}. In this figure, the lower bound in
Theorem~\ref{thm2 - CourtadeV14} is compared to the tighter bound in \eqref{eq5: FV reliability},
illustrating the improvement for small to moderate values of $n$.
\end{example}

\section{Lower Bounds for Variable-Length Source Coding Allowing Errors}
\label{sec: almost-lossless source coding}

Following a recent study by Kuzuoka \cite{Kuzuoka16}, this section applies
our bounding techniques to derive
improved lower bounds on the cumulant generating function of the codeword lengths
for variable-length source coding allowing errors (which, in contrast to the
conventional fixed-to-fixed paradigm, are not necessarily detectable by the decoder)
by means of the smooth R\'{e}nyi entropy in Definition~\ref{def: smooth Renyi entropy}.

In contrast to \cite{Kuzuoka16}, the bounds in this subsection are
derived for source codes without the prefix condition when either the maximal
or average decoding error probabilities are limited not to exceed a given
value $\varepsilon \in [0,1)$.

\begin{theorem} \label{th1: almost lossless}
Let $X$ take values on a finite set $\set{X}$, and let
$f \colon \set{X} \to \set{C}$ be an encoder (possibly stochastic) with a finite codebook
$\set{C} \subseteq \{0,1\}^{\ast}$, and let $\ell \colon \set{C} \to \{0, 1, \ldots, \}$
be the length function of the codewords in $\set{C}$. Fix $\varepsilon \in [0,1)$ and
$\rho > 0$.
\begin{enumerate}[1)]
\item If the {\em average} decoding error probability cannot be larger than $\varepsilon$,
then
\begin{align} \label{170608_1}
\frac1\rho \, \log \expectation \left[ 2^{\rho \, \ell(f(X))} \right] \geq \sup_{\beta>0}
\frac1\beta \left[H_{\frac{\beta}{\beta+\rho}}^{(\varepsilon)}(X) - \log t(\beta, |\set{X}|) \right]
\end{align}
where $H_{\alpha}^{(\varepsilon)}(X)$ is the
$\varepsilon$-smooth R\'{e}nyi entropy of order $\alpha$,
and $t(\cdot)$ is given in \eqref{eq: t}.
\item If the {\em maximal} decoding error probability cannot be larger than $\varepsilon$,
then also
\begin{align} \label{170608_2}
\frac1\rho \, \log \expectation \left[ 2^{\rho \, \ell(f(X))} \right] \geq
\sup_{\beta \in (-\rho, 0)} \frac1\beta \left[H_{\frac{\beta}{\beta+\rho}}(X)
- \log t(\beta, |\set{X}|) \right] - \frac1\rho \, \log \frac1{1-\varepsilon}.
\end{align}
\end{enumerate}
\end{theorem}

\begin{proof} We first derive \eqref{170608_1}, and then rely on its proof for the derivation of \eqref{170608_2}.
\begin{enumerate}[1)]
\item \label{part 1 - 170806}
Let $Q \colon \set{X} \to \set{C}$ denote a transition probability matrix such that
$Q(c|x)$ is the probability that a codeword $c \in \set{C}$ is assigned
to $x \in \set{X}$ by the stochastic encoder.
Let $\psi \colon  \set{C} \to  \set{X}$ be the deterministic decoding function.
The average decoding error probability is given by
\begin{align} \label{170608_3}
P_{\mathrm{e}} &= \prob[ X \neq \psi(f(X))] \\
\label{170608_4}
&= \sum_{x \in \set{X}} P_X(x) \sum_{c: \, \psi(c) \neq x} Q(c|x).
\end{align}
In order to minimize $P_{\mathrm{e}}$ for a given (stochastic) encoder,
the decision relies on a MAP decoder:
\begin{align} \label{170608_5}
\psi(c) \in \arg\max_{x \in \set{X}} Q(c|x) P_X(x), \quad c \in \set{C}
\end{align}
where ties are arbitrarily resolved. Let
\begin{align} \label{170608_7}
\gamma(x) \triangleq Q\bigl(\psi^{-1}(x) | x \bigr) \in (0,1]
\end{align}
denote the probability that $x \in \set{X}$ is assigned to a codeword that is decoded into $x$.
Since the average decoding error probability satisfies $P_{\mathrm{e}} \leq \varepsilon$,
it follows from \eqref{170608_4} and \eqref{170608_7} that
\begin{align} \label{170608_8}
\sum_{x \in \set{X}} P_X(x) \gamma(x) \geq 1-\varepsilon.
\end{align}
For all $x \in \set{X}$, let
\begin{align} \label{170608_11}
\ell_{\psi}(x) \triangleq \min_{c \in \psi^{-1}(x)} \ell(c)
\end{align}
be the minimal length of the codewords for which the decoder chooses $x$, and let
\begin{align}  \label{170608_16}
\mu(x) \triangleq P_X(x) \gamma(x), \quad x \in \set{X}.
\end{align}
From \eqref{170608_7}, \eqref{170608_8} and \eqref{170608_16}
\begin{align}
& 0 \leq \mu(x) \leq P_X(x), \quad x \in \set{X}, \label{170608_17} \\
& \sum_{x \in \set{X}} \mu(x) \geq 1-\varepsilon, \label{170608_18}
\end{align}
which, by \eqref{eq: smooth RE 2}, yields
\begin{align}  \label{170608_19}
& \hspace*{-2.2cm} \mu \in \set{B}^{(\varepsilon)}(P_X).
\end{align}
For all $\rho > 0$
\begin{align}
\expectation\Bigl[ 2^{\rho \, \ell(f(X))} \Bigr]
&= \sum_{x \in \set{X}} P_X(x) \sum_{c \in \set{C}} Q(c|x) \, 2^{\rho \, \ell(c)} \label{170608_9} \\
&\geq \sum_{x \in \set{X}} P_X(x) \sum_{c \in \psi^{-1}(x)} Q(c|x) \, 2^{\rho \, \ell(c)}, \label{170608_10}
\end{align}
and, for every $x \in \set{X}$,
\begin{align}
\sum_{c \in \psi^{-1}(x)} Q(c|x) \, 2^{\rho \, \ell(c)} 
& \geq \sum_{c \in \psi^{-1}(x)} Q(c|x) \, 2^{\rho \, \ell_{\psi}(x)} \label{170608_13} \\
& = \gamma(x) \, 2^{\rho \, \ell_{\psi}(x)} \label{170608_14}
\end{align}
where \eqref{170608_13} holds due to \eqref{170608_11}, and
\eqref{170608_14} follows from \eqref{170608_7}.
Hence, for all $\rho > 0$,
\begin{align}  \label{170608_15}
\expectation\Bigl[ 2^{\rho \, \ell(f(X))} \Bigr] & \geq \sum_{x \in \set{X}} P_X(x) \gamma(x) \, 2^{\rho \, \ell_{\psi}(x)} \\
\label{170608_20}
&= \sum_{x \in \set{X}} \mu(x) \, 2^{\rho \, \ell_{\psi}(x)}
\end{align}
where \eqref{170608_15} follows from \eqref{170608_9}--\eqref{170608_14},
and \eqref{170608_20} is due to \eqref{170608_16}.
The finite sets $\{\psi^{-1}(x)\}_{x \in \set{X}}$ are
disjoint. For $x \in \set{X}$, let $c^{\ast}(x) \in \psi^{-1}(x)$ be a codeword which achieves the minimum
in the right side of \eqref{170608_11}, i.e.,
\begin{align}  \label{170608_21}
\ell_{\psi}(x) = \ell(c^{\ast}(x)).
\end{align}
Since the codewords $\{c^{\ast}(x)\}_{x \in \set{X}}$ are distinct,
it follows from \eqref{170608_21} and the proof of Lemma~\ref{lemma: l inequality}
that
\begin{align}  \label{170608_22}
\sum_{x \in \set{X}} 2^{-\beta \, \ell_{\psi}(x)} \leq t(\beta, |\set{X}|), \quad \beta \geq 0,
\end{align}
and
\begin{align}  \label{170608_23}
\sum_{x \in \set{X}} 2^{-\beta \, \ell_{\psi}(x)} \geq t(\beta, |\set{X}|), \quad \beta \leq 0
\end{align}
since any perturbation of the set of codeword lengths of a $P_X$-optimal code necessarily shifts
shorter to longer codewords.
Let
\begin{align}  \label{170608_24}
\alpha \triangleq \frac{\beta+\rho}{\beta}
\end{align}
with $\rho >0$ and $\beta \neq 0$, and define
the following probability mass functions on $\set{X}$:
\begin{align}
& R(x) = \frac{\mu^{\frac1\alpha}(x)}{\underset{a \in \set{X}}{\sum} \mu^{\frac1\alpha}(a)},
\label{170608_25} \\[0.1cm]
& S(x) = \frac{2^{-\beta \, \ell_\psi(x)}}{\underset{a \in \set{X}}{\sum} 2^{-\beta \, \ell_\psi(a)}}.
\label{170608_26}
\end{align}
Straightforward calculation shows that
\begin{align}
D_{\alpha}(R \| S)
&= \frac1{\alpha-1} \, \log \sum_{x \in \set{X}} R^{\alpha}(x) S^{1-\alpha}(x) \label{170608_27} \\
&= \frac{\beta}{\rho} \, \log \left( \sum_{x \in \set{X}} \mu(x) \, 2^{\rho \, \ell_\psi(x)} \right)
+ \log \left( \sum_{x \in \set{X}} 2^{-\beta \, \ell_\psi(x)} \right)
- \frac{\beta+\rho}{\rho} \, \log \left( \sum_{x \in \set{X}} \mu^{\frac{\beta}{\beta+\rho}}(x) \right).
\label{170608_28}
\end{align}
We now fix $\beta > 0$, and we proceed to upper bound each of the three terms in \eqref{170608_28}:
the first term is upper bounded using \eqref{170608_20}, the second term is upper bounded by
$\log t(\beta, |\set{X}|)$ in view of \eqref{170608_22}, and the third term satisfies
\begin{align} \label{170608_31}
& H_{\frac{\beta}{\beta+\rho}}^{(\varepsilon)}(X) \leq \frac{\beta+\rho}{\rho}
\; \log \sum_{x \in \set{X}} \mu^{\frac{\beta}{\beta+\rho}}(x),
\end{align}
which holds in view of Definition~\ref{def: smooth Renyi entropy} and \eqref{170608_19}.
Plugging those bounds into the right side of \eqref{170608_28}, and recalling that
$D_{\alpha}(R \| S) \geq 0$ (this follows from Lemma~\ref{lemma: RD}, and since
$\alpha>0$ in \eqref{170608_24} for $\beta, \rho>0$), we obtain
\begin{align}
\label{170608_33}
\frac1\rho \, \log \expectation\left[2^{\rho \, l(f(X))} \right]
& \geq \frac1\beta \left[H_{\frac{\beta}{\beta+\rho}}^{(\varepsilon)}(X) - \log t(\beta, |\set{X}|) \right].
\end{align}
Finally, \eqref{170608_1} follows by supremizing \eqref{170608_33} over the free parameter $\beta>0$.

\item \label{part 2 - 170806}
If the maximal decoding error probability does not exceed $\varepsilon$, then so is the average decoding error
probability, and we can rely on the results in Item~\ref{part 1 - 170806}) of this proof.
Fix $\beta \in (-\rho, 0)$. In view of \eqref{170608_24}, $\alpha < 0$ and
$D_{\alpha}(R\|S) \leq 0$ (Lemma~\ref{lemma: RD}). From \eqref{170608_28}, we get
\begin{align}
\label{170608_34}
\frac{\beta}{\rho} \, \log \left( \sum_{x \in \set{X}} \mu(x) \, 2^{\rho \, \ell_\psi(x)} \right)
+ \log \left( \sum_{x \in \set{X}} 2^{-\beta \, \ell_\psi(x)} \right)
- \frac{\beta+\rho}{\rho} \, \log \left( \sum_{x \in \set{X}} \mu^{\frac{\beta}{\beta+\rho}}(x) \right) \leq 0.
\end{align}
Due to the above assumption on the maximal decoding error probability, \eqref{170608_7} implies that for every $x \in \set{X}$
\begin{align}
\label{170608_35}
\gamma(x) \in [1-\varepsilon, 1],
\end{align}
and, consequently, \eqref{170608_16} implies that
\begin{align}
\label{170608_36}
(1-\varepsilon) P_X(x) \leq \mu(x) \leq P_X(x), \quad x \in \set{X}.
\end{align}
Since $\frac{\beta}{\beta+\rho}<0$ and $\frac{\rho}{\beta+\rho}>0$, \eqref{eq: Renyi entropy}
(with negative orders of the R\'{e}nyi entropy) and \eqref{170608_36} yield
\begin{align}
\label{170608_37}
H_{\frac{\beta}{\beta+\rho}}(X) & \leq \frac{\beta+\rho}{\rho} \,
\log \sum_{x \in \set{X}} \mu^{\frac{\beta}{\beta+\rho}}(x) \\
\label{eq: 20171017-b}
& \leq H_{\frac{\beta}{\beta+\rho}}(X) - \frac{\beta}{\rho} \, \log \frac1{1-\varepsilon}.
\end{align}
Consequently, we have
\begin{align}
\frac1\rho \, \log \expectation\left[2^{\rho \, l(f(X))} \right]
& \geq \frac1\rho \, \log \sum_{x \in \set{X}} \mu(x) \, 2^{\rho \, \ell_\psi(x)} \label{170608_38} \\[0.1cm]
& \geq \frac1\beta \left[ \frac{\beta+\rho}{\rho} \, \log \left( \sum_{x \in \set{X}} \mu^{\frac{\beta}{\beta+\rho}}(x) \right) -
\log \left( \sum_{x \in \set{X}} 2^{-\beta \, \ell_\psi(x)} \right) \right]  \label{170608_39} \\[0.1cm]
& \geq \frac1\beta \left[ \frac{\beta+\rho}{\rho} \, \log \left( \sum_{x \in \set{X}} \mu^{\frac{\beta}{\beta+\rho}}(x) \right) -
\log t(\beta, |\set{X}|) \right] \label{170608_40} \\[0.1cm]
& \geq \frac1\beta \left[ H_{\frac{\beta}{\beta+\rho}}(X) - \frac{\beta}{\rho} \, \log \frac1{1-\varepsilon} -
\log t(\beta, |\set{X}|) \right] \label{170608_41}
\end{align}
where \eqref{170608_38} holds due to \eqref{170608_20}; \eqref{170608_39} follows from \eqref{170608_34}
and since $\beta<0$; \eqref{170608_40} holds due to \eqref{170608_23}; finally, \eqref{170608_41} follows
from \eqref{eq: 20171017-b}. Finally, \eqref{170608_2} follows by supremizing the right side of \eqref{170608_41}
over the free parameter $\beta \in (-\rho, 0)$.
\end{enumerate}
\end{proof}

\begin{remark}
Note that the lower bound in \eqref{170608_1} is expressed in terms of the smooth R\'{e}nyi entropy of orders $\alpha \in (0,1)$,
and the lower bound in \eqref{170608_2} is expressed in terms of the conventional R\'{e}nyi entropy of negative orders. The calculation
of the bound in \eqref{170608_1} relies on Lemma~\ref{lemma: Koga13}\ref{lemma: Koga13-a}).
\end{remark}

\begin{remark} \label{remark: 20171018}
The combination of letting $\varepsilon=0$ in \eqref{170608_1} and \eqref{170608_2} recovers the result in Lemma~\ref{lemma: tightened bound on the cumulant}.
\end{remark}

\begin{remark}
If the maximal decoding error probability cannot be larger than $\varepsilon \in (0,1)$, then neither of the bounds in \eqref{170608_1} and \eqref{170608_2}
is superseded by the other, as can be verified by numerical experimentation.
\end{remark}

\begin{remark} \label{remark: Th. 2 by Kuzuoka}
For prefix codes, Kraft's inequality gives
$\underset{x \in \set{X}}{\sum} 2^{-\ell_\psi(x)} \leq 1$. Replacing $t(1, |\set{X}|)$ in the right side
of \eqref{170608_22} by 1
recovers the result in \cite[Theorem~2]{Kuzuoka16} from \eqref{170608_1}:
\begin{align} \label{eq: prefix, Kuzuoka}
\frac1\rho \, \log \expectation\left[2^{\rho \, \ell(f(X))} \right] \geq H_{\frac{1}{1+\rho}}^{(\varepsilon)}(X), \quad \rho > 0
\end{align}
where the average error probability cannot be larger than $\varepsilon$.
An analogous result of \eqref{eq: prefix, Kuzuoka}, for a lower bound on the normalized cumulant generating function
of the codeword lengths without imposing the prefix condition, is given by
\begin{align} \label{eq: 20171029-e}
\frac1\rho \, \log \expectation\left[2^{\rho \, \ell(f(X))} \right] \geq H_{\frac{1}{1+\rho}}^{(\varepsilon)}(X) - \log t(1, |\set{X}|), \quad \rho > 0
\end{align}
where $t(1, |\set{X}|)$ is given in \eqref{eq: t}. Since \eqref{eq: 20171029-e} is obtained by replacing the
supremization over $\beta>0$ in the right side of \eqref{170608_1} with its value at $\beta=1$,
Theorem~\ref{th1: almost lossless} gives a better bound than \eqref{eq: 20171029-e}.
\end{remark}

\begin{figure}[h] \label{figure: 20171029}
\vspace*{-3.7cm}
\begin{center}
\hspace*{-0.8cm}
\includegraphics[width=14cm, angle=90]{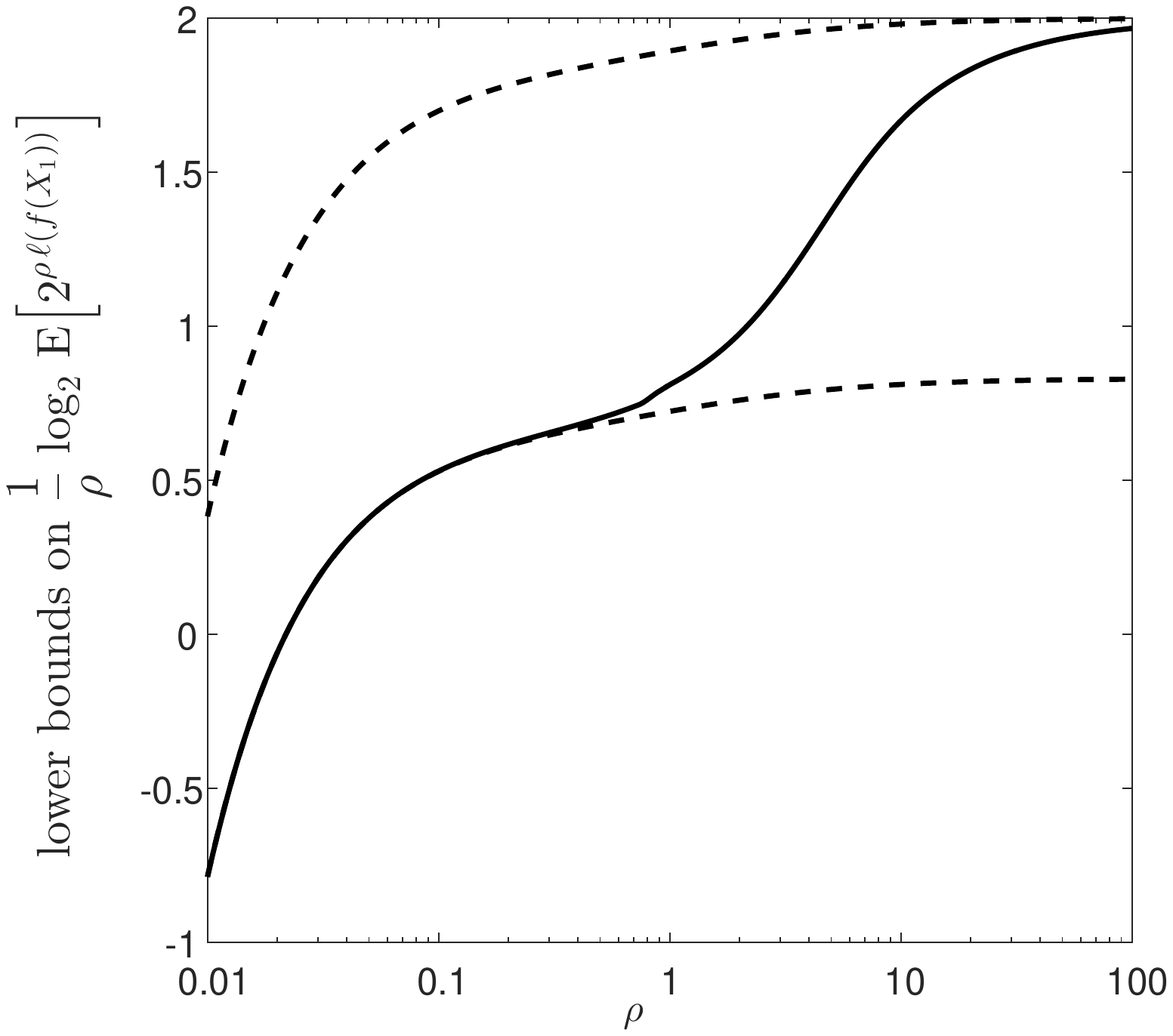}
\hspace*{-1.8cm}
\includegraphics[width=14cm, angle=90]{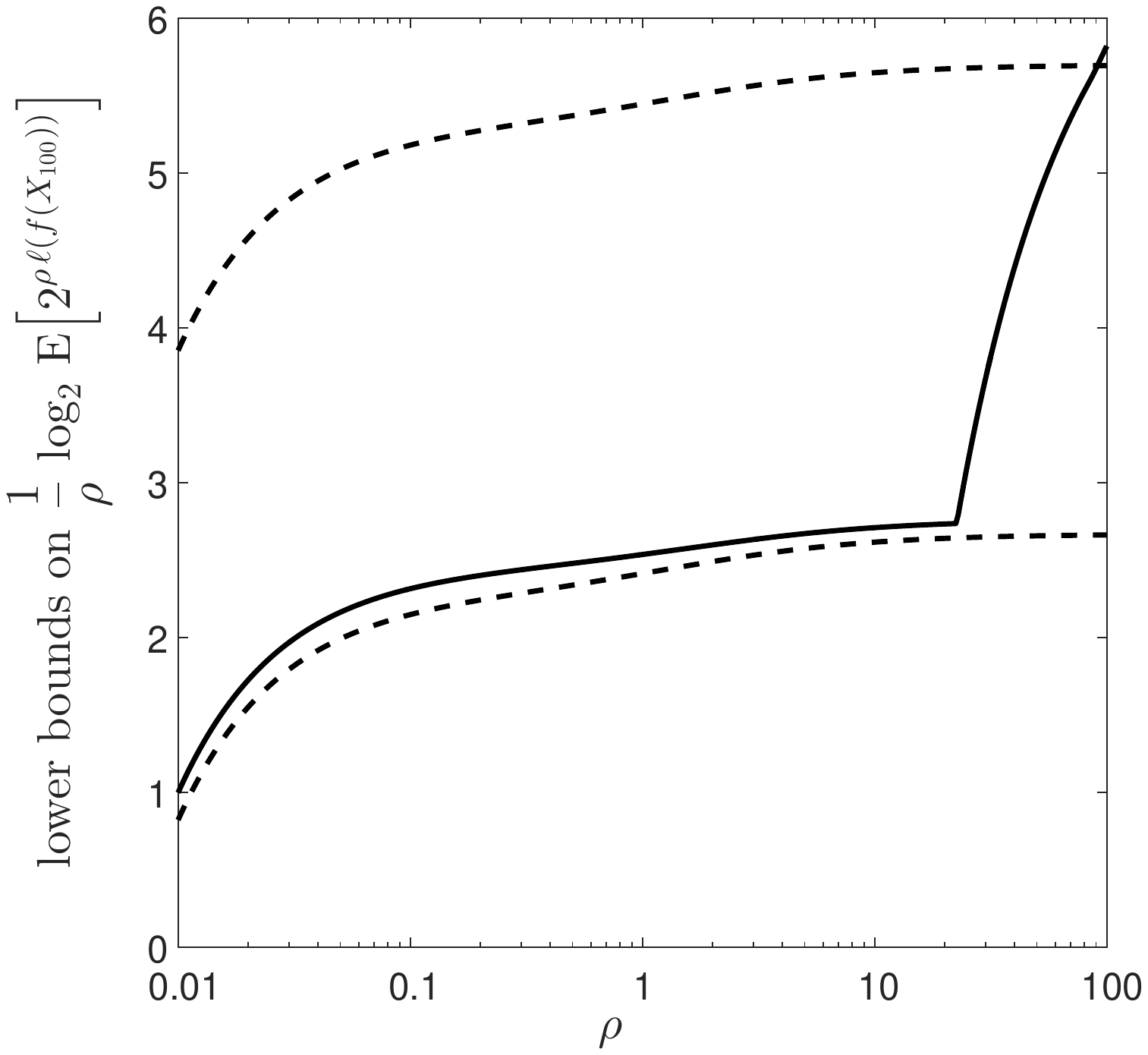} \\[-7cm]
\hspace*{-1cm}
\includegraphics[width=10cm]{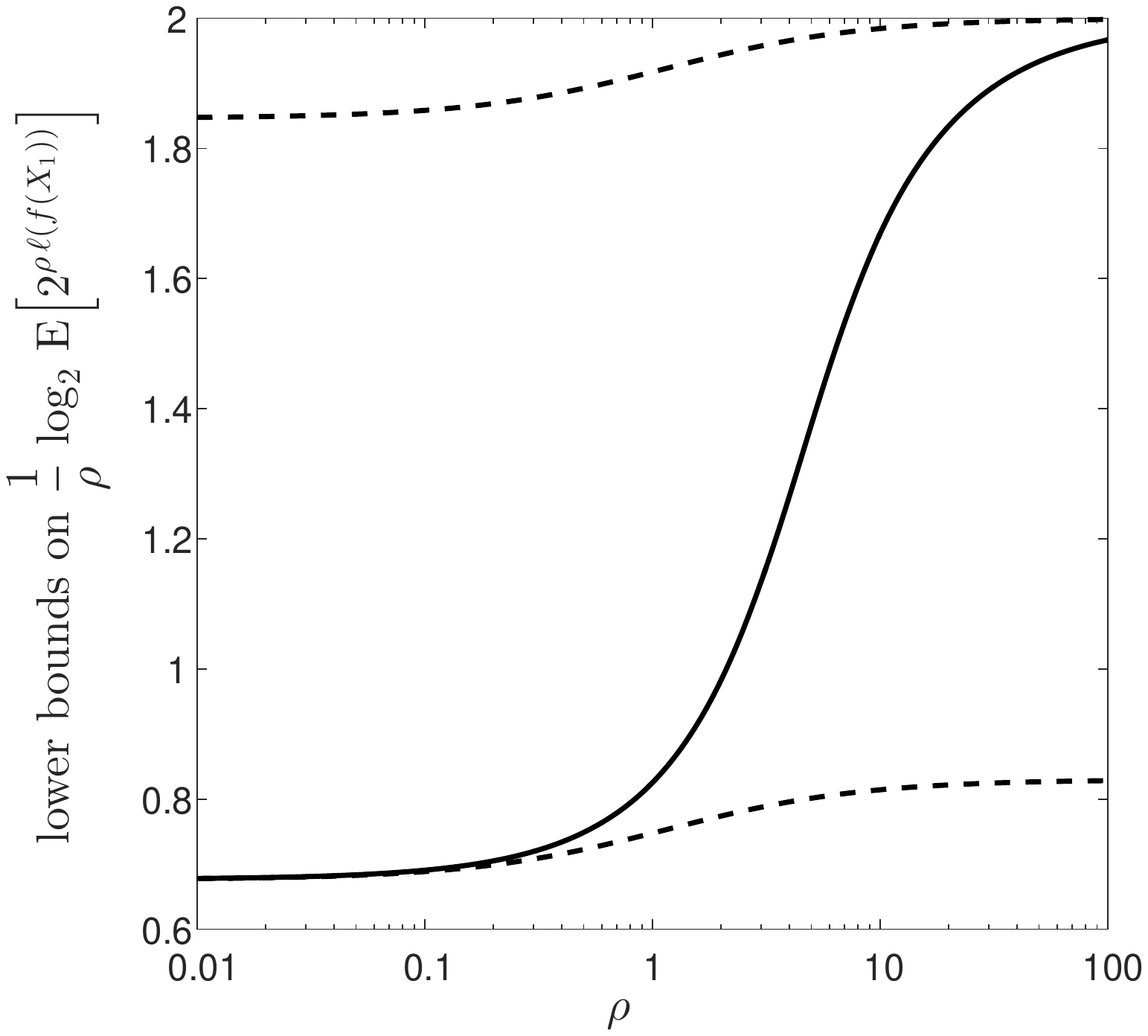}
\hspace*{-1.8cm}
\includegraphics[width=14cm, angle=90]{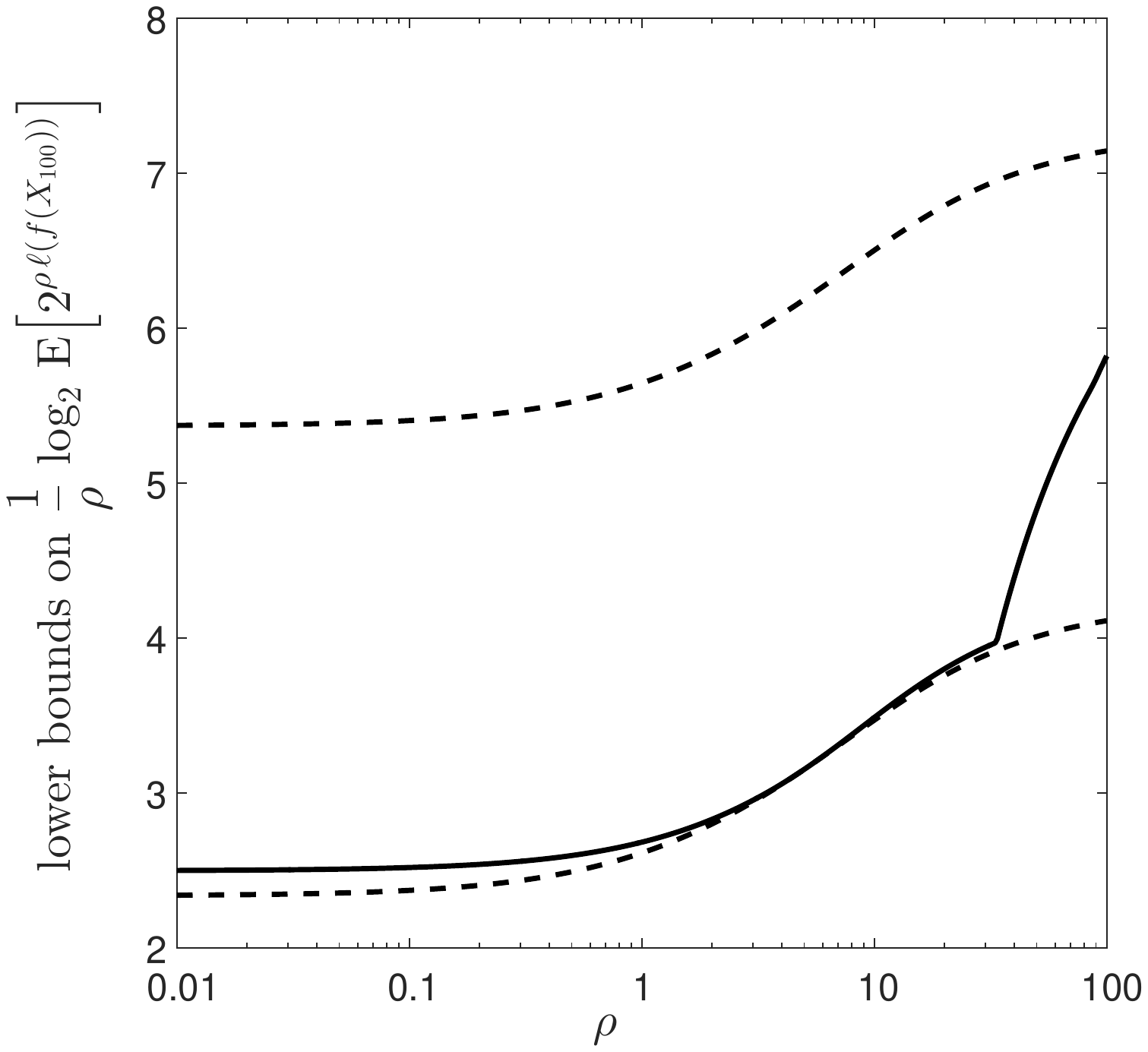}
\end{center}
\vspace*{-3.5cm}
\caption{\label{figure: 20171029} Example~\ref{example: 20171029}: Lower bounds on
$\frac1\rho \log_2 \expectation\Bigl[2^{\rho \, \ell(f(X_n))} \Bigr]$
with a maximal error probability of $\varepsilon = 0.01$ (upper plots),
and the lossless case: $\varepsilon = 0$ (lower plots).
The left and right plots correspond to $n=1$ and $n=100$, respectively.
The upper and lower dashed lines in each plot refer, respectively, to \eqref{eq: prefix, Kuzuoka}
and \eqref{eq: 20171029-e}; the solid lines refer to the combined
lower bounds in \eqref{170608_1} and \eqref{170608_2} (Theorem~\ref{th1: almost lossless}).}
\end{figure}
\begin{example} \label{example: 20171029}
Let $U_1, \ldots, U_n$ be i.i.d. random variables taking values in a set of cardinality~4,
and having the probability mass function
\begin{align}
P_{U_1} = \begin{array}{cccc}
[\tfrac{4}{10} & \tfrac{3}{10} & \tfrac{2}{10} & \tfrac{1}{10}].
\end{array}
\end{align}
Let
\begin{align}
X_n = \overset{n}{\underset{i=1}{\sum}} U_i,   \quad n \in \naturals.
\end{align}
The probability mass function $P_{X_n}$ is equal to $P_{U_1}$ convolved with itself
$n-1$ times, and $X_n$ takes $M_n = 3n+1$ values.
Assume that a maximal error probability of $\varepsilon \in [0,1)$ is
allowed in decoding $X_n$.
Figure~\ref{figure: 20171029} compares the lower bound in \eqref{eq: prefix, Kuzuoka},
for binary prefix codes, with the lower bounds in Theorem~\ref{th1: almost lossless}
and \eqref{eq: 20171029-e} for binary codes without the prefix condition.
The upper and lower plots in Figure~\ref{figure: 20171029} correspond to $\varepsilon = 0.01$
and $\varepsilon=0$ (i.e., the lossless case), respectively; the left and right plots correspond,
respectively, to $n=1$ and $n=100$ (note that, by the central limit theorem, $X_n$ is
close to Gaussian for large $n$).
The gain of the combined lower bounds in \eqref{170608_1} and \eqref{170608_2} (the solid
line in each plot) over the bound in \eqref{eq: 20171029-e} (the dashed lower line in
the corresponding plot) is illustrated in Figure~\ref{figure: 20171029} by comparing
the left and right plots; by increasing the value of $n$, the solid line becomes more
steep at sufficiently large values of $\rho$.
\end{example}

\appendices

\section{Completion of the Proof of Theorem~\ref{theorem: improving Arikan's bound}}
\label{appendix: UB on harmonic sum}

We first prove \eqref{eq1: UB on harmonic sum}.
\begin{itemize}
\item For $\beta=1$, \eqref{eq1: UB on harmonic sum} holds due to the
upper bound on the harmonic sum in \cite[Theorem~1]{GouQi11};
\item For $\beta > 1$
\begin{align}
\sum_{j=1}^M \frac1{j^\beta} &= \zeta(\beta) - \sum_{j=M+1}^{\infty} \frac1{j^\beta} \nonumber \\
&= \zeta(\beta) - \tfrac12 (M+1)^{-\beta} - \sum_{j=M+1}^{\infty} \tfrac12 \left(j^{-\beta} + (j+1)^{-\beta} \right) \label{eq1: sum}
\end{align}
where $\zeta(\beta) = \sum_{n=1}^{\infty} \frac1{n^\beta}$ for $\beta>1$ denotes
the Riemann zeta function; due to the convexity of $f_{\beta}(t) = t^{-\beta}$ in $(0, \infty)$,
for all $j \in \naturals$,
\begin{align}  \label{eq2: integral}
\tfrac12 \left(j^{-\beta} + (j+1)^{-\beta} \right) \geq \int_j^{j+1} t^{-\beta} \, \mathrm{d}t.
\end{align}
Hence, from \eqref{eq1: sum} and \eqref{eq2: integral}, for all $\beta > 1$,
\begin{align}
\sum_{j=1}^M \frac1{j^\beta} &\leq
\zeta(\beta) - \tfrac12 (M+1)^{-\beta} - \int_{M+1}^{\infty} t^{-\beta} \, \mathrm{d}t \label{shavuot0} \\
&= \zeta(\beta) - \tfrac12 (M+1)^{-\beta} - \frac{(M+1)^{1-\beta}}{\beta-1} \label{shavuot1}
\end{align}
and
\begin{align}
\sum_{j=1}^M \frac1{j^\beta} \leq \sum_{j=1}^M \frac1{j} \leq u_M(1). \label{20171011a}
\end{align}
Combining \eqref{shavuot0}--\eqref{20171011a} proves \eqref{eq1: UB on harmonic sum} for $\beta > 1$.
\item For $\beta > 0$,
from the convexity of $f_{\beta}(t) = t^{-\beta}$ in $(0, \infty)$, Jensen's inequality yields
\begin{align}
\int_{j-\frac12}^{j+\frac12} t^{-\beta} \, \mathrm{d}t \geq \frac1{j^\beta} \label{20171011b}
\end{align}
for all $j \in \naturals$, which implies that
\begin{align} \label{estimate integral}
\sum_{j=1}^M \frac1{j^\beta} & \leq 1
+ \int_{\frac32}^{M+\frac12} t^{-\beta} \, \mathrm{d}t \\
\label{shavuot2}
&= 1 + \tfrac1{1-\beta} \left[\bigl(M + \tfrac12 \bigr)^{1-\beta}
- \bigl(\tfrac32\bigr)^{1-\beta}\right].
\end{align}
This proves \eqref{eq1: UB on harmonic sum} for $\beta \in (0, 1)$. It can be
verified numerically that the upper bound in \eqref{shavuot1} supersedes the bound
\eqref{shavuot2} for $\beta > 1$; for this reason, we ignore the
bound in \eqref{shavuot2} for the derivation of $u_M(\beta)$ for $\beta > 1$.
\end{itemize}

We next prove \eqref{eq2: UB on harmonic sum}.
\begin{itemize}
\item For $\beta \in [-1,0)$, \eqref{eq2: UB on harmonic sum} follows
from the concavity of $f(t) = t^{-\beta}$ for $t > 0$; by Jensen's
inequality, we obtain the opposite inequalities in \eqref{20171011b}
and \eqref{estimate integral} for $\beta \in [-1,0)$.
\item For $\beta=-1$, $\sum_{j=1}^M \frac1{j^\beta} = \tfrac12 \, M(M+1).$
\item For $\beta \in (-\infty, -1)$, due to the convexity of
$f(t) = t^{-\beta}$ for $t > 0$,
\begin{align}
\sum_{j=1}^M \frac1{j^\beta} &= \tfrac12 +
\sum_{j=1}^{M-1} \tfrac12 \bigl(j^{-\beta} + (j+1)^{-\beta}\bigr)
+ \tfrac12 \, M^{-\beta} \\
& \geq \tfrac12 + \sum_{j=1}^{M-1} \int_j^{j+1} t^{-\beta} \, \mathrm{d}t
+ \tfrac12 \, M^{-\beta} \\
& = \int_1^M t^{-\beta} \, \mathrm{d}t + \tfrac12 \bigl(1+M^{-\beta}\bigr) \\[0.1cm]
& = \frac{M^{1-\beta}-1}{1-\beta} + \tfrac12 \bigl(1+M^{-\beta}\bigr).
\end{align}
\end{itemize}
Finally, Theorem~\ref{theorem: improving Arikan's bound} follows from
Theorem~\ref{thm: key result}, \eqref{eq1: UB on harmonic sum}
and \eqref{eq2: UB on harmonic sum}.

\section{Proof of Theorem~\ref{thm: 0-2}}
\label{appendix: UB-NSI}

We first prove \eqref{eq1.5: UB-NSI} by relying on the following result:
\begin{lemma} \label{lemma: 0<rho<1}
For $\rho \in (0,1)$ and $u \geq 1$
\begin{align}  \label{eq: 0<rho<1}
u^{\rho} \leq \frac{u^{1+\rho} - (u-1)^{1+\rho}}{1+\rho} +
\frac{\rho}{1+\rho} \; 1\{1 \leq u < 2\} +
\left( 2^{\rho} - \frac{2^{1+\rho}-1}{1+\rho} \right) \, 1\{u \geq 2\}.
\end{align}
\end{lemma}
\begin{proof}
For $\rho \in (0,1)$, let $f_{1, \rho} \colon [1,\infty) \to \Reals$ and
$f_{2, \rho} \colon [1,\infty) \to \Reals$ be given by
\begin{align}
f_{1,\rho}(u) &= \frac{u^{1+\rho} - (u-1)^{1+\rho}}{1+\rho} + \frac{\rho}{1+\rho} - u^\rho, \label{eq: f1} \\
f_{2,\rho}(u) &= \frac{u^{1+\rho} - (u-1)^{1+\rho}}{1+\rho} + 2^{\rho} - \frac{2^{1+\rho}-1}{1+\rho} - u^\rho. \label{eq: f2}
\end{align}
For $\rho \in (0,1)$ and $u \in [1, \infty)$
\begin{align}
f_{1,\rho}'(u) = f_{2,\rho}'(u) &= u^\rho - (u-1)^\rho - \rho u^{\rho-1} \\[-0.2cm]
&= \rho c^{\rho-1} - \rho u^{\rho-1}, \quad c \in (u-1, u) \label{mvt}   \\[-0.2cm]
&> 0,
\end{align}
where \eqref{mvt} holds by the mean value theorem of calculus;
moreover, since $f_{1,\rho}(1)=f_{2,\rho}(2)=0$,
\begin{align}
\label{eq1: f}
f_{1,\rho}(u) \geq 0, \quad u \in [1, \infty), \\[-0.2cm]
\label{eq2: f}
f_{2,\rho}(u) \geq 0, \quad u \in [2, \infty).
\end{align}
This gives
\begin{align}
\label{eq3: f}
0 &\leq f_{1,\rho}(u) \, 1\{1 \leq u < 2\} +
\min\{f_{1,\rho}(u), f_{2,\rho}(u)\} \, 1\{u \geq 2\} \\[-0.1cm]
\label{eq4: f}
&= f_{1,\rho}(u) \, 1\{1 \leq u < 2\} +
f_{2,\rho}(u) \, 1\{u \geq 2\}
\end{align}
where \eqref{eq3: f} follows from \eqref{eq1: f} and \eqref{eq2: f},
and \eqref{eq4: f} follows from \eqref{eq: f1}, \eqref{eq: f2}, and
since
\begin{align}
2^\rho - \frac{2^{1+\rho}-1}{1+\rho} - \frac{\rho}{1+\rho}
= \frac{(\rho-1) (2^{\rho}-1)}{1+\rho} < 0
\end{align}
for $\rho \in (0,1)$. Finally, \eqref{eq: 0<rho<1} is equivalent
to the non-negativity of the right side of \eqref{eq4: f}.
\end{proof}

From Lemma~\ref{lemma: 0<rho<1}, for $\rho \in (0,1)$,
\begin{align}
\expectation\left[g_X^{\rho}(X)\right]
& \leq \frac1{1+\rho} \; \expectation\left[g_X^{1+\rho}(X)
- \bigl(g_X(X)-1\bigr)^{1+\rho}\right] \nonumber \\[0.1cm]
& \hspace*{0.4cm} + \frac{\rho \, \prob[g_X(X)=1]}{1+\rho}
+ \left( 2^{\rho} - \frac{2^{1+\rho}-1}{1+\rho}
\right) \; \prob[g_X(X) \geq 2] \label{eq3: for 0<rho<1} \\[0.1cm]
& = \frac1{1+\rho} \; \expectation\left[g_X^{1+\rho}(X)
- \bigl(g_X(X)-1\bigr)^{1+\rho}\right]
+ \frac{\rho \, p_{\max}}{1+\rho} + \left( 2^{\rho} -
\frac{2^{1+\rho}-1}{1+\rho} \right) (1-p_{\max})
\label{eq4: for 0<rho<1}
\\[0.1cm]
\label{eq5: for 0<rho<1}
& \leq \frac1{1+\rho} \, \exp\left(\rho H_{\frac1{1+\rho}}(X) \right)
+ \frac{\rho \, p_{\max}}{1+\rho} + \left( 2^{\rho} -
\frac{2^{1+\rho}-1}{1+\rho} \right) (1-p_{\max})
\end{align}
where \eqref{eq3: for 0<rho<1} follows from
\eqref{eq: 0<rho<1} by substituting $u = g_X(X)$, and taking
expectations on both sides of the inequality;
\eqref{eq4: for 0<rho<1} holds since $\prob[g_X(X)=1] = p_{\max}$
(the first guess of $X$ is a mode of $P_X$); \eqref{eq5: for 0<rho<1}
follows from \eqref{eq: from Boztas' paper}, which is then simplified to
\eqref{eq1.5: UB-NSI}.

We next prove \eqref{eq2: UB-NSI}.
\begin{lemma} \label{lemma: 1<rho<2}
If $\rho \in [1,2]$ and $u \geq 1$, then
\begin{align} \label{for 1<rho<2}
u^{\rho} \leq \frac{u^{1+\rho}-(u-1)^{1+\rho}}{1+\rho} +
\frac{u^\rho-(u-1)^\rho}\rho + \frac{\rho^2-\rho-1}{\rho(1+\rho)}.
\end{align}
\end{lemma}
\begin{proof}
For $\rho \in [1,2]$, let $f \colon [1,\infty) \to \Reals$ be given by
\begin{align}
f_{\rho}(u) = \frac{u^{1+\rho}-(u-1)^{1+\rho}}{1+\rho} +
\frac{u^\rho-(u-1)^\rho}\rho - u^{\rho} + \frac{\rho^2-\rho-1}{\rho(1+\rho)}, \quad u \in [1, \infty).
\end{align}
For all $u \in [1, \infty)$,
\begin{align}
f_{\rho}'(u) &= u^\rho - (u-1)^\rho + u^{\rho-1}-(u-1)^{\rho-1} - \rho u^{\rho-1} \\
&\geq 1 + \rho (u-1)^{\rho-1} + u^{\rho-1}-(u-1)^{\rho-1} - \rho u^{\rho-1} \label{eq1: diff f} \\
&= 1 + (\rho-1) \left( (u-1)^{\rho-1} - u^{\rho-1} \right) \\
&\geq 2-\rho \geq 0 \label{eq2: diff f}
\end{align}
where \eqref{eq1: diff f} follows from the convexity of $u \mapsto u^{\rho}$
in $(0, \infty)$ for $\rho \geq 1$, and \eqref{eq2: diff f} holds since
\begin{align}
0 \leq u^{\rho-1} - (u-1)^{\rho-1} \leq 1
\end{align}
for $\rho \in [1,2]$ and $u \in [1, \infty)$. It can be verified that $f_{\rho}(1)=0$,
which implies that $f_{\rho}(u) \geq 0$ for $\rho \in [1,2]$ and $u \geq 1$.
\end{proof}

Replacing $x$ in \eqref{for 1<rho<2} with $g_X(X) \geq 1$, and taking expectations
on both sides of \eqref{for 1<rho<2} yield, via \eqref{eq: from Boztas' paper} and
\eqref{tel-aviv}, the result in \eqref{eq2: UB-NSI}.

If $X$ is deterministic, then $H_{\alpha}(X)=0$ for all $\alpha > 0$,
$g_X(X)=1$, and $p_{\max}=1$, which imply that \eqref{eq1.5: UB-NSI}
and \eqref{eq2: UB-NSI} hold with equality (note that
$\frac1{1+\rho} + \frac1{\rho} + \frac{\rho^2-\rho-1}{\rho(\rho+1)} = 1$
holds for $\rho \neq -1, 0$).

\section{Proof of Theorem~\ref{thm: rho>=2}}
\label{appendix: thm for rho>=2}

\begin{lemma} \label{lemma: rho>=2}
Under the assumptions in Theorem~\ref{thm: Arikan's UB}, if $\rho \geq 2$, then
\begin{align}  \label{eq3: UB-NSI}
\expectation[g_{X}^{\rho}(X)] \leq \frac1{1+\rho} \,
\exp\Bigl(\rho H_{\frac1{1+\rho}}(X) \Bigr) +
\frac{\rho}{2} \; \expectation[g_{X}^{\rho-1}(X)]
- \frac{\rho(\rho-1)}{2(1+\rho)}
\end{align}
with equality if $X$ is deterministic.
\end{lemma}
\begin{proof}
If $\rho \geq 2$ and $u \geq 1$, then
\begin{align} \label{eq1: rho>=2}
u^\rho \leq \frac{u^{1+\rho}-(u-1)^{1+\rho}}{1+\rho}
+ \frac{\rho}{2} \; u^{\rho-1} - \frac{\rho(\rho-1)}{2(1+\rho)}.
\end{align}
To prove \eqref{eq1: rho>=2}, let $\xi \colon [1, \infty) \to \Reals$
be given by
\begin{align} \label{eq: xi fun}
\xi(u) = \frac{u^{1+\rho}-(u-1)^{1+\rho}}{1+\rho} - u^\rho
+ \frac{\rho u^{\rho-1}}{2} - \frac{\rho(\rho-1)}{2(1+\rho)}
\end{align}
for $u \geq 1$, and similarly to \eqref{eq: v fun}, we denote $v(u) = u^\rho$ for $u \geq 0$. Then,
for $u \geq 1$,
\begin{align}
\xi'(u) &= u^\rho - (u-1)^\rho - \rho u^{\rho-1} + \tfrac12 \, \rho (\rho-1) u^{\rho-2}  \\
&= v(u) - v(u-1) - v'(u) + \tfrac12 v''(u). \label{t'-v}
\end{align}
By a Taylor series expansion of $v(\cdot)$ around $u$,
\begin{align}
v(u-1) = v(u) - v'(u) + \tfrac12 v''(u) - \tfrac16 v^{(3)}(c)
\end{align}
for some $c \in (u-1, u) \subseteq (0, \infty)$. For $\rho \geq 2$, we have
$v^{(3)}(c) = \rho (\rho-1) (\rho-2) c^{\rho-3} \geq 0$, which implies from \eqref{t'-v} that
$\xi'(u) \geq 0$ for all $u \geq 1$.
It can be verified from \eqref{eq: xi fun} that $\xi(1)=0$, which implies that $\xi(\cdot) \geq 0$ in $[1, \infty)$;
this gives \eqref{eq1: rho>=2} from \eqref{eq: xi fun}.
By substituting $u = g_X(X)$ in \eqref{eq1: rho>=2} and taking expectations on both
sides of the inequality, it follows that for $\rho \geq 2$
\begin{align}
\label{eq2: rho>=2}
\expectation[g_{X}^{\rho}(X)] \leq & \, \frac{\expectation \bigl[g_{X}^{1+\rho}(X)\bigr] -
\expectation\bigl[ \bigl(g_{X}(X)-1 \bigr)^{1+\rho}\bigr]}{1+\rho}
+ \frac{\rho}{2} \; \expectation[g_{X}^{\rho-1}(X)] - \frac{\rho(\rho-1)}{2(1+\rho)}
\end{align}
which, from \eqref{eq: from Boztas' paper}, yields \eqref{eq3: UB-NSI}. If $X$ is deterministic,
then $g_X(X)=1$ and $H_{\alpha}(X)=0$ for $\alpha>0$, which implies equality in
\eqref{eq3: UB-NSI} since $\frac1{1+\rho} + \frac{\rho}{2} - \frac{\rho(\rho-1)}{2(1+\rho)} = 1$ holds for $\rho \neq -1$.
\end{proof}

We proceed to prove Theorem~\ref{thm: rho>=2}.
Combining the recursive upper bound in Lemma~\ref{lemma: rho>=2} with
Theorem~\ref{thm: 0-2}-\ref{label2: UB-NSI}) gives, after some straightforward
algebra, an upper bound on $\expectation[g_{X}^{\rho}(X)]$
for $\rho \geq 2$ which is of the form
\begin{align} \label{what is d?}
\expectation[g_{X}^{\rho}(X)] \leq \sum_{j=0}^{\lfloor \rho \rfloor}
c_j(\rho) \exp \left( (\rho-j) \, H_{\frac1{1+\rho-j}}(X)\right) + d(\rho),
\end{align}
where the sequence $\{c_j(\rho)\}$ is given in \eqref{eq: c_j}, and
$d(\rho)$ is an additive term which only depends on $\rho$ (but it does not on the distribution of $X$).
Since the results in Theorem~\ref{thm: 0-2}-\ref{label2: UB-NSI}) and Lemma~\ref{lemma: rho>=2}
are satisfied with equalities if $X$ is deterministic, then it follows that also
\eqref{what is d?} holds with equality in this special case. For such $X$, we have $g_X(X)=1$ and
$H_{\alpha}(X)=0$ for all $\alpha > 0$, which therefore implies that
\begin{align} \label{20171011c}
d(\rho) = 1 - \sum_{j=0}^{\lfloor \rho \rfloor} c_j(\rho).
\end{align}
The bound in \eqref{eq: rho>=2} is obtained by combining \eqref{what is d?} and \eqref{20171011c}.

\section{Auxiliary Results for Section~\ref{subsection: guessing moments versus epsilon} }
\label{app: d}

Lemmas~\ref{lemma: f}--\ref{lemma: f is convex} are used to derive the bounds on the optimal generalized
guessing moment in \eqref{NY} and \eqref{NY2}, and Lemma~\ref{lemma: 20171024} refers to the proof of
Theorem~\ref{thm: 20171022}.
\begin{lemma} \label{lemma: f}
Let $\ell \in \naturals$, $p_1 \geq p_2 \geq \ldots \geq p_\ell \geq 0$ with
$\overset{\ell}{\underset{i=1}{\sum}} p_i = 1$,
and let $f \colon \{1, \ldots, \ell\} \to \Reals$ satisfy
\begin{align}  \label{eq: 20171012b}
\frac1j \sum_{k=l-j+1}^{\ell} f(k) \geq f(\ell-j),
\end{align}
for all $j \in \{1, \ldots, \ell\}$. Then,
\begin{align}  \label{eq: 20171012c}
\sum_{i=1}^{\ell} p_i f(i) \leq \frac1{\ell} \sum_{i=1}^{\ell} f(i).
\end{align}
Furthermore, if the inequality in \eqref{eq: 20171012b} is strict for all $j$, then \eqref{eq: 20171012c}
holds with equality if and only if $p_i = \frac1{\ell}$ for all $i \in \{1, \ldots, \ell\}$.
\end{lemma}
\begin{proof}
Denote
\begin{align} \label{eq: seq. u}
u_j = \sum_{i=1}^{\ell-j} p_i f(i) + \frac1j \left( \sum_{i=\ell-j+1}^{\ell} p_i \right) \left( \sum_{i=\ell-j+1}^{\ell} f(i) \right), \quad j \in \{1, \ldots, \ell\}.
\end{align}
Its first and last terms are equal to the side of \eqref{eq: 20171012c}:
\begin{align}
& u_1 = \sum_{i=1}^{\ell} p_i f(i), \\
& u_{\ell} = \frac1{\ell} \sum_{i=1}^{\ell} f(i)
\end{align}
so, proving that $\{u_j\}$ is monotonically increasing is sufficient to show \eqref{eq: 20171012c}.
By its definition in \eqref{eq: seq. u}, straightforward calculation shows that for $j \in \{1, \ldots, \ell-1\}$
\begin{align}
u_{j+1} - u_j &= \frac1{j(j+1)} \left( j \, p_{\ell-j} - \sum_{i=\ell-j+1}^{\ell} p_i \right)
\left( \sum_{i=\ell-j+1}^{\ell} f(i) - j f(\ell-j) \right) \label{eq: product} \\
& \geq 0 \label{non-negative}
\end{align}
where \eqref{non-negative} holds in view of $p_1 \geq \ldots \geq p_{\ell}$ and
\eqref{eq: 20171012b}. Note that \eqref{eq: product} and \eqref{non-negative} imply that
$u_1 = u_l$ if and only if $u_{j+1}=u_j$ for all $j \in \{1, \ldots, l-1\}$. Hence, if the inequality
in \eqref{eq: 20171012b} is strict, then \eqref{eq: product} implies that $u_{j+1}=u_j$ for all $j$
if and only if the monotonic sequence $\{p_i\}$ is fixed, i.e., $p_1 = p_2= \ldots = p_{\ell} = \frac1{\ell}$
(since, by assumption, this sequence sums to~1).
\end{proof}

\begin{lemma} \label{lemma-2: f}
Let $l \in \naturals$, $q_1 \geq \ldots \geq q_{\ell} \geq 0$, and let
$f \colon \{1, \ldots, \ell\} \to \Reals$ be a strictly monotonically increasing function.
Then,
\begin{align} \label{eq: 20171012d}
\sum_{i=1}^{\ell} q_i \, f(1) \leq \sum_{i=1}^{\ell} q_i f(i) \leq \frac1{\ell} \left( \sum_{i=1}^{\ell} q_i \right) \left( \sum_{i=1}^{\ell} f(i) \right)
\end{align}
with equality in the left inequality if and only if $q_i = 0$ for all $i \in \{2, \ldots, \ell\}$,
and equality in the right inequality if and only if $q_1 = \ldots q_{\ell}$.
\end{lemma}
\begin{proof}
Since by assumption $f \colon \{1, \ldots, \ell\} \to \Reals$ is a strictly monotonically increasing function,
the inequality in \eqref{eq: 20171012b} is strict for all $j \in \{1, \ldots, \ell-1\}$. Let
$\{p_i\}$ be the normalized version of the non-negative, monotonically decreasing sequence $\{q_i\}$ such that
$\overset{\ell}{\underset{i=1}{\sum}} p_i = 1$. Hence,
\begin{align} \label{p}
p_i = \frac{q_i}{\overset{\ell}{\underset{j=1}{\sum}} q_j}, \quad i \in \{1, \ldots, \ell\},
\end{align}
and $p_1 \geq p_2 \geq \ldots \geq p_{\ell} \geq 0$. Lemma~\ref{lemma: f} and \eqref{p} give
\begin{align}
\label{eq0: cor 2}
\sum_{i=1}^{\ell} q_i f(i) &= \left( \sum_{i=1}^{\ell} p_i f(i) \right) \left( \sum_{i=1}^{\ell} q_i \right) \\
\label{eq1: cor 2}
&\leq \frac1{\ell} \left( \sum_{i=1}^{\ell} f(i) \right) \left( \sum_{i=1}^{\ell} q_i \right),
\end{align}
and, due to Lemma~\ref{lemma: f}, \eqref{eq1: cor 2} holds with equality if and only if $p_1 = \ldots = p_{\ell} = \frac1{\ell}$;
due to \eqref{p}, this holds if and only if $q_1 = \ldots = q_{\ell}$.
Moreover, the assumptions on $f \colon \{1, \ldots, \ell\} \to \Reals$ and $\{q_i\}$ imply that
\begin{align}
\label{eq2: cor 2}
\sum_{i=1}^{\ell} q_i f(i) \geq f(1) \sum_{i=1}^{\ell} q_i,
\end{align}
where \eqref{eq2: cor 2} holds with equality if and only if $q_i=0$ for all $i \in \{2, \ldots, \ell\}$.
\end{proof}

\begin{lemma} \label{lemma: f is convex}
The function $f_{\rho} \colon [0,1) \to [0, \infty)$ defined in \eqref{eq: f-thm10}--\eqref{eq: k-thm10}
is convex for all $\rho>0$.
\end{lemma}
\begin{proof}
From \eqref{eq: k-thm10}, the value of $k_u \in \naturals$ is fixed in each interval
$[1-\frac1m, 1-\frac1{m+1})$ with $m \in \naturals$. Hence, $f_{\rho}$ in \eqref{eq: f-thm10} is
a linear function in each such interval; its positive slope, denoted by
$s_{\rho}(m)$, is given by
\begin{align} \label{eq: slope-s}
s_{\rho}(m) = k_u (k_u+1)^{\rho} - \sum_{j=1}^{k_u} j^\rho.
\end{align}
By a transition from an interval $[1-\frac1m, 1-\frac1{m+1})$ to the successive interval
$[1-\frac1{m+1}, 1-\frac1{m+2})$, the value of the positive integer $k_u$ is increased by~1
(see \eqref{eq: k-thm10}); consequently, it can be verified from \eqref{eq: slope-s} that
for $\rho>0$
\begin{align}
s_{\rho}(m+1) > s_{\rho}(m), \quad m \in \naturals.
\end{align}
Hence, the slope of the linear function obtained by restricting $f_{\rho}$ to
the interval $[1-\frac1{m+1}, 1-\frac1{m+2})$ is larger than its slope in the interval
$[1-\frac1m, 1-\frac1{m+1})$. Hence, $f_{\rho}$ can be decomposed by linear functions
in each interval $[1-\frac1m, 1-\frac1{m+1})$ whose slopes are monotonically increasing
in $m \in \naturals$. It can be also verified from \eqref{eq: f-thm10} that the function
$f_{\rho}$ is continuous at the endpoints of these intervals, which therefore yields its
convexity on $[0,1) = \underset{m \in \naturals}{\bigcup} [1-\frac1m, 1-\frac1{m+1})$.
\end{proof}

\vspace*{0.2cm}
\begin{lemma} \label{lemma: 20171024}
The identity in \eqref{eq: 20171022-h} holds for every integer $M \geq 2$.
\end{lemma}
\begin{proof}
The result in \eqref{eq: 20171022-h} is trivial for $M=2$. Let $M \geq 3$, then
\begin{align}
\left| \begin{array}{cccc}
 1  &   1  &   \cdots  &   1 \\
 0  &  \log_{\mathrm{e}} 2  &   \cdots  &   \log_{\mathrm{e}} M \\
 0  &  \log_{\mathrm{e}}^2 2   &  \cdots  &   \log_{\mathrm{e}}^2 M \\
 0  & \vdots & \vdots & \vdots \\
 0  &  \log_{\mathrm{e}}^{M-1} 2   &  \cdots  &   \log_{\mathrm{e}}^{M-1} M
\end{array}
\right|
\label{eq: 20171024-a}
 & = \left| \begin{array}{ccc}
 \log_{\mathrm{e}} 2  &   \cdots  &   \log_{\mathrm{e}} M \\
 \log_{\mathrm{e}}^2 2   &  \cdots  &   \log_{\mathrm{e}}^2 M \\
 \vdots & \vdots & \vdots \\
 \log_{\mathrm{e}}^{M-1} 2   &  \cdots  &   \log_{\mathrm{e}}^{M-1} M
 \end{array}
 \right| \\[0.1cm]
 \label{eq: 20171024-b}
 & = \left(\prod_{k=2}^M \log_{\mathrm{e}} k \right)
 \left| \begin{array}{ccc}
 1  &   \cdots  &   1 \\
 \log_{\mathrm{e}} 2  &   \cdots  &   \log_{\mathrm{e}} M \\
 \vdots & \vdots & \vdots \\
 \log_{\mathrm{e}}^{M-2} 2   &  \cdots  &   \log_{\mathrm{e}}^{M-2} M
 \end{array}
 \right| \\[0.1cm]
 \label{eq: 20171024-d}
 & = \prod_{k=2}^M \log_{\mathrm{e}} k \;
 \prod_{2 \leq i < j \leq M} (\log_{\mathrm{e}} j - \log_{\mathrm{e}} i )
 \end{align}
where \eqref{eq: 20171024-a} holds by expanding according to the first column;
\eqref{eq: 20171024-b} holds by factoring $\log_{\mathrm{e}} k$ from the $(k-1)$-th row for $k=2, \ldots, M$;
finally, \eqref{eq: 20171024-d} relies on the Vandermonde determinant (e.g., \cite[p.~155]{Lang}).
\end{proof}

\end{document}